\documentclass[11pt,reqno,twoside]{amsart}
\usepackage[asymmetric,top=2.5cm,bottom=2.8cm,left=2.38cm,right=2.38cm]{geometry}

\usepackage{graphicx}
\usepackage{amssymb}
\usepackage{epstopdf}
\usepackage{float}
\usepackage{caption}
\usepackage{mathtools}
\usepackage{booktabs}
\usepackage{subcaption}
\usepackage{amsmath, amssymb, amsthm}
\usepackage{mathtools}
\usepackage{geometry}
\usepackage{hyperref}
\hypersetup{colorlinks=true, linkcolor=blue, citecolor=blue, urlcolor=blue}

\usepackage{graphicx}
\usepackage{subcaption}

\usepackage[T1]{fontenc}
\usepackage{lmodern}

\usepackage{lipsum} % 仅用于生成示例文本
\usepackage{hyperref}
\hypersetup{hypertex=true,
	colorlinks=true,
	linkcolor=blue,
	anchorcolor=blue,
	citecolor=blue}

\usepackage[usenames]{color}

\usepackage{xcolor}

\newtheorem{thm}{Theorem}[section]

\newtheorem{cor}[thm]{Corollary}
\newtheorem{lem}{Lemma}[section]
\newtheorem{prop}{Proposition}[section]
\theoremstyle{definition}
\newtheorem{defn}{Definition}[section]

\newtheorem{exm}{Example}
\newtheorem{rem}{Remark}[section]
\numberwithin{equation}{section}

\allowdisplaybreaks

\usepackage{xstring} % 用于字符串比较

\newcommand{\AuthorInfo}[4]{%
  \textsc{#1}%
  \IfStrEq{#4}{true}{$^{*}$}{}\\
  #2\\
  \textit{E-mail:} \texttt{#3}%
  \IfStrEq{#4}{true}{\\\textit{$^{*}$Corresponding author.}}{}
  \par\vspace{0.5em}
}

\title[Quasi-Minnaert Resonances in High-contrast acoustic Structures]{Quasi-Minnaert Resonances in High-contrast acoustic Structures  and Applications to Invisibility Cloaking}

\date{} % Activate to display a given date or no date (if empty),
\begin{document}

	\maketitle

	\begin{center}
\bigskip
\footnotesize
\AuthorInfo{Weisheng Zhou}{School of Mathematics, Jilin University, Changchun 130012, People's Republic of China}{wszhou1211@163.com, zhouws24@mails.jlu.edu.cn}{false}

\AuthorInfo{Huaian Diao}{School of Mathematics and Key Laboratory of Symbolic Computation and Knowledge Engineering of Ministry of Education, Jilin University, Changchun 130012, People's Republic of China}{diao@jlu.edu.cn, hadiao@gmail.com}{true}

\AuthorInfo{Hongyu Liu}{Department of Mathematics, City University of Hong Kong, Kowloon, Hong Kong SAR, People's Republic of China}{hongyu.liuip@gmail.com, hongyliu@cityu.edu.hk}{false}
\end{center}
\normalsize

	\begin{abstract}
		This paper investigates a novel quasi-Minnaert resonance phenomenon in acoustic wave propagation through high-contrast medium in both two and three dimensions, occurring in the sub-wavelength regime. These media are characterized by physical properties significantly distinct from those of a homogeneous background. The quasi-Minnaert resonance is defined by two primary features: boundary localization, where the $L^2$-norms of the internal total field and the external scattered field  exhibit pronounced concentration near the boundary, and surface resonance, marked by highly oscillatory behavior of the fields near the boundary. In contrast to classical Minnaert resonances, which are associated with a discrete spectral spectrum tied to physical parameters, quasi-Minnaert resonances exhibit analogous physical phenomena but with a continuous spectral spectrum. Using layer potential theory and rigorous asymptotic analysis, we demonstrate that the coupling between a high-contrast material structure, particularly with radial geometries, and a carefully designed incident wave is critical for inducing quasi-Minnaert resonances. Extensive numerical experiments, involving radial geometries (e.g., unit disks and spheres) and general-shaped geometries (e.g., hearts, Cassini ovals, and clovers in $\mathbb{R}^2$, and spheres in $\mathbb{R}^3$), validate the occurrence of these resonances. Furthermore, we numerically demonstrate that quasi-Minnaert resonances induce an invisibility cloaking effect in the high-contrast medium. These findings have significant implications for mathematical material science and the development of acoustic cloaking technologies.		
		\medskip
		
 
 \medskip


	\end{abstract}
	
	\section{Introduction}

	The Minnaert resonance, originally introduced in \cite{M} to model the acoustic resonance of gas bubbles in liquids, plays a fundamental role in characterizing high-contrast acoustic scatterers within mathematical material science. For the case of a single bubble in a homogeneous medium, the formula for the Minnaert resonance frequencies of arbitrarily shaped bubbles was rigorously analyzed in \cite{ammari} using layer potential methods and Gohberg-Sigal theory. The acoustic scattering phenomenon induced by a large ensemble of bubbles in liquids, along with the Minnaert resonance within a subwavelength phononic bandgap, was rigorously investigated in \cite{AFL17,AFH20}. An explicit characterization via a variational formulation was subsequently provided in \cite{FH}. Recent investigations \cite{LM,LM2024} have demonstrated that the Minnaert resonant frequency corresponds to a pole of the resolvent operator associated with the acoustic scattering problem. Extensions of this framework to elastic scattering \cite{HHJ2023,WDSK25} and acoustic-elastic coupling scattering \cite{HHJ2022} have further broadened the applicability of Minnaert resonances. Applications of  Minnaert resonances to inverse problems, in both frequency and time domains, have been explored in \cite{AAM2021,SM2023}.   The Minnaert
	resonance phenomena have a very wide range of applications in the industry, including medical treatment \cite{SRNLLZ}, superabsorption \cite{VAMHPA} and superresolution imaging \cite{DDBPE}.  
	
	The wave field generated by high-contrast structures in the sub-wavelength regime exhibits boundary localization and high-frequency oscillation phenomena, as demonstrated in \cite{MRFMVA,MTJV}. Recently, we introduced the concept of \textit{quasi-Minnaert resonances} to characterize these effects in the context of elastic wave scattering (cf.~\cite{DTL}). Unlike classical Minnaert resonances \cite{ammari,FH,HHJ2022,HHJ2023}, which are primarily described in terms of the ratio of physical parameters between the high-contrast medium and the homogeneous background, quasi-Minnaert resonances depend on both the geometry of the high-contrast structure and the incident wave field in the sub-wavelength regime. In particular, the boundary localization and high oscillation of the wave field---referred to as \textit{surface resonance} in our context---are rigorously analyzed in \cite{DTL}.

	In this paper, we investigate quasi-Minnaert resonances in acoustic wave propagation through high-contrast structures in both two and three dimensions. Our primary objective is to explore the careful design of incident waves and high-contrast acoustic structures to induce quasi-Minnaert resonances. We thoroughly examine the interplay among the incident wave, the high-contrast structure, and the resulting resonances, with a particular focus on radial geometries. In contrast to prior studies on quasi-Minnaert resonances in elastic scattering \cite{DTL}, this work rigorously analyzes quasi-Minnaert resonances for high-contrast acoustic structures in $\mathbb{R}^2$ and $\mathbb{R}^3$. Additionally, we conduct extensive numerical experiments to validate our theoretical findings. Furthermore, we demonstrate numerically that invisibility cloaking \cite{ACKLM13,BILSA, ZLKC} can be achieved by leveraging quasi-Minnaert resonances through careful selection of the incident wave and high-contrast structure.

	Notably, for a fixed high-contrast structure, an appropriately designed incident wave can rigorously achieve boundary localization at a prescribed level, as established in Theorems~\ref{3dthm1} and~\ref{2dthm1}. Conversely, given a target boundary localization level for the internal total wave field  or the external scattered wave field, Corollary~\ref{cor:3.2} demonstrates how to select the density contrast parameter $\delta$ and adjust the incident wave  to realize the desired localization. The emergence of surface resonances in both internal and external fields is rigorously proven in Theorem~\ref{thm:gradient_u} under the setting of a fixed high-contrast structure with a suitably chosen incident wave. Moreover, for a given density contrast $\delta$, surface resonances can be induced by appropriately modifying the incident wave.

The quasi-Minnaert resonances introduced in this work are characterized by the simultaneous occurrence of boundary localization and surface resonance in both the internal total field  and the external scattered field, resulting from the subtle interplay between the acoustic high-contrast structure and the incident wave. Unlike classical Minnaert resonances, which correspond to discrete frequencies \cite{ammari}, the quasi-Minnaert resonance frequencies form a continuous spectrum. To validate these theoretical findings, we conduct extensive numerical experiments, including those involving general-shaped high-contrast structures, such as the unit disk, heart, Cassini oval, and clover in $\mathbb{R}^2$, and spherical geometries in $\mathbb{R}^3$.

Our analysis relies on layer potential techniques and the spectral theory of the Neumann--Poincar\'{e} operator, enabling a rigorous asymptotic analysis of the associated Helmholtz systems. Additionally, we numerically demonstrate that quasi-Minnaert resonances can induce invisibility cloaking by carefully selecting the incident wave and high-contrast structure, as shown through the quantity defined in \eqref{eq:invisibility_metric}. These results have significant implications for the development of cloaking technologies.

		The structure of this paper is organized as follows: 
	Section~\ref{sec2} introduces the mathematical framework and layer potential theory. 
	Section~\ref{sec3} analyzes the simultaneous boundary localization of both the internal total field and the external scattered field in the sub-wavelength regime in $\mathbb{R}^3$, and further explores surface resonance based on the established boundary localization phenomenon. 
	Section~\ref{sec4} investigates quasi-Minnaert resonance in $\mathbb{R}^2$. 
	Section~\ref{sec6} validates the theoretical results through numerical simulations and shows that quasi-Minnaert resonance induces a cloaking effect in the high-contrast medium.

	\section{Mathematical Setup}\label{sec2}
	
	In this section, we present the mathematical framework of our study. The high-contrast medium is modeled as a bounded, simply connected domain $\Omega \subset \mathbb{R}^d$ ($d = 2, 3$), where the boundary $\partial\Omega$ belongs to the class $C^{1,s}$ for some $0 < s < 1$. The acoustic properties of the medium $\Omega$ are characterized by the density $\rho_b \in \mathbb{R}_+$ and bulk modulus $\kappa_b \in \mathbb{R}_+$, while the background medium in $\mathbb{R}^d \setminus \Omega$ is described by the density $\rho \in \mathbb{R}_+$ and bulk modulus $\kappa \in \mathbb{R}_+$. Throughout this paper, we denote the acoustic material configuration by $(\Omega; \rho_b, \kappa_b)$ and the background by $(\mathbb{R}^d \setminus \overline{\Omega}; \rho, \kappa)$. Let $u^i$ be an incident acoustic wave that satisfies the homogeneous background Helmholtz equation
	\begin{align} \label{eq:u_i}
		\nabla \cdot \left( \frac{1}{\rho} \nabla u^i(\mathbf{x}) \right) + \frac{\omega^2}{\kappa} u^i(\mathbf{x}) = 0 \quad \text{in } \mathbb{R}^d,
	\end{align}
	where $\omega \in \mathbb{R}_+$ denotes the angular frequency. The interaction of $u^i$ with the inhomogeneity $\Omega$ gives rise to a scattered field $u^s$, and the total field is denoted by $u = u^i + u^s$.
	
	The acoustic wave scattering problem is governed by the following system:
	\begin{equation} \label{eq:system}
		\begin{cases}
			\nabla \cdot \left( \frac{1}{\rho} \nabla u(\mathbf{x}) \right) + \frac{\omega^2}{\kappa} u(\mathbf{x}) = 0, & \text{in } \mathbb{R}^d \setminus \overline{\Omega}, \\
			\nabla \cdot \left( \frac{1}{\rho_b} \nabla u(\mathbf{x}) \right) + \frac{\omega^2}{\kappa_b} u(\mathbf{x}) = 0, & \text{in } \Omega, \\
			u|_+ = u|_-, & \text{on } \partial\Omega, \\
			\frac{1}{\rho} \frac{\partial u|_+}{\partial \nu} = \frac{1}{\rho_b} \frac{\partial u|_-}{\partial \nu}, & \text{on } \partial\Omega, \\
			u^s := u - u^i & \text{satisfies the radiation condition}.
		\end{cases}
	\end{equation}
	Here, $\partial / \partial \nu$ denotes the outward normal derivative, and the notations $\left. \cdot \right|_{\pm}$ represent the trace limits from outside and inside $\partial \Omega$, respectively. We define the wavenumbers $k_b := \omega \sqrt{\rho_b / \kappa_b}$ in the medium $\Omega$ and $k := \omega \sqrt{\rho / \kappa}$ in the background. The Sommerfeld radiation condition \cite{JCN} ensures that the scattered field $u^s$ is outgoing and decays at infinity. It is expressed as:
	\begin{equation} \label{eq:src}
		\lim_{|\mathbf{x}| \to \infty} |\mathbf{x}|^{\frac{d-1}{2}} \left(  \nabla u^s(\mathbf{x})\cdot \hat {\mathbf{x} } - i k u^s(\mathbf{x}) \right) = 0,
	\end{equation}
	where $i = \sqrt{-1}$ denotes the imaginary unit and $\hat {\mathbf{x}}:=\frac{\mathbf{x}}{|\mathbf{x}|}$.  Acoustic wave scattering finds extensive applications in imaging, non-destructive testing, and geophysics. For related studies, we refer the reader to \cite{D1,D2,D3,D4,D5} and the references cited therein.

	In this article, we investigate acoustic wave scattering in a high-contrast medium $\Omega$, embedded within a homogeneous background material, under the sub-wavelength regime. A canonical physical example is a liquid medium containing dispersed gas bubbles, as studied in \cite{KA, PW, MRFMVA}. The configurations discussed herein satisfy the physical constraints of such systems. We focus on the scattering properties of a high-contrast medium within a homogeneous background, characterized by significant disparities in density and bulk modulus. Specifically, we assume:
	\begin{equation}
		\label{eq:delta}
		\delta = \frac{\rho_b}{\rho} \ll 1, \quad \delta' = \frac{\kappa_b}{\kappa} \ll 1,
	\end{equation}
	where $\delta$ represents the density contrast and $\delta'$ denotes the bulk modulus contrast between the background (with density $\rho_b$ and bulk modulus $\kappa_b$) and the medium $\Omega$ (with density $\rho$ and bulk modulus $\kappa$). The material structure $(\Omega; \rho_b, \kappa_b)$ is termed a high-contrast acoustic configuration, a framework previously employed to analyze Minnaert resonances in \cite{ammari}. The density contrast parameter $\delta$ is pivotal in our analysis, particularly for understanding boundary localization and surface resonance phenomena.

	\iffalse
	
	In this article, we study a configuration in which the acoustic high-contrast  medium $\Omega$ is embedded within a homogeneous background material in the sub-wavelength regime. A representative physical example is a liquid medium containing dispersed gas bubbles \cite{KA,PW,MRFMVA}. The configurations described below satisfy the physical requirements of such systems. In this paper we are concerned with the acoustic high-contrast medium scattering in a homogenous background under the following assumption: there are high-contrast differences in density and bulk modulus between the high-contrast medium and the homogenous background with 
	\begin{equation}\label{eq:delta}
		\delta = \frac{\rho_b}{\rho} \ll 1, \quad \delta' = \frac{\kappa_b}{\kappa} \ll 1,
	\end{equation}
	where $\delta$ or $\widetilde \delta$ characterizes the density contrast and $\delta'$ or $\widetilde \delta'$ is the bulk modulus contrast.  
	Accordingly, the material structure $(\Omega; \rho_b, \kappa_b)$ is referred to as a high-contrast acoustic configuration, which has been used to study the Minnaeart resonances in \cite{ammari}. We emphasize that the density contrast parameter $\delta$ plays a crucial role in the  our subsequent analysis for boundary localization and surface resonance. 
	\fi
	Recall that the speed of the background is given by $v := \sqrt{\kappa/\rho}$ and $v_b := \sqrt{\kappa_b/\rho_b}$ is the wave propagating speed in the high-contrast medium. We introduce the additional parameter
	\begin{equation}\label{eq:tau}
		\tau = \frac{v_b}{v} = \frac{k_b}{k} = \mathcal{O}(1),
	\end{equation}
	which indicates that the contrast in wave speeds inside and outside the medium is not significant. Moreover, we impose the scaling condition on the wave speeds:
	\begin{equation}\label{eq:wv}
		v = \mathcal{O}(1), \quad v_b = \mathcal{O}(1),
	\end{equation}
	ensuring that both sound speeds remain bounded independently of the high-contrast parameters $\delta$ and $\delta'$. Such material configurations have also been studied in the literature; see, e.g., \cite{ammari}. The sub-wavelength regime implies that the size of the high-contrast medium $\Omega$ is much smaller than the typical wavelength of the incident wave, quantified by the condition
	\begin{equation}\label{eq:sw}
		\omega \cdot \mathrm{diam}(\Omega) \ll 1,
	\end{equation}
	where $\omega$ is the angular frequency. 
	
	To simplify the analysis, we adopt a coordinate normalization so that the diameter of $\Omega$ satisfies
	\begin{equation}\label{eq:diamOmega}
		\mathrm{diam}(\Omega) = \mathcal{O}(1).
	\end{equation}
	It follows from \eqref{eq:sw} and \eqref{eq:diamOmega} that the angular frequency accordingly  satisfy
	\begin{equation}\label{eq:omega}
		\omega = o(1).
	\end{equation}
	Consequently, the corresponding wavenumbers in the high-contrast medium and in the background satisfy the asymptotic  relation:
	\begin{equation}\label{eq:wn1}
		k_b = o(1), \quad k = o(1).
	\end{equation}

	Next, we introduce the layer potential operators for the Helmholtz equation. Let $\boldsymbol{G}^{k}(\mathbf{x})$ denote the fundamental solution of the Helmholtz equation in $\mathbb{R}^{d}$, which is given by  
	\[
	\boldsymbol{G}^{k}(\mathbf{x}) = \begin{cases}  
		-\frac{i}{4} H_0^{(1)}(k|\mathbf{x}|), & d = 2, \\  
		-\frac{e^{ik|\mathbf{x}|}}{4\pi|\mathbf{x}|}, & d = 3,  
	\end{cases}
	\]
	where $H_0^{(1)}(k|\mathbf{x}|)$ is the Hankel function of the first kind of order zero, and this fundamental solution satisfies the radiation condition \eqref{eq:src}. For any bounded Lipschitz domain $\Omega \subset \mathbb{R}^d$, we define the single-layer potential operator $\mathcal{S}^{k}_{\partial \Omega}: H^{-1/2}(\partial \Omega) \to H^{1}(\mathbb{R}^d \setminus\partial \Omega)$ as  
	\begin{align}\label{eq:S_k} 
		\mathcal{S}^{k}_{\partial \Omega}[\varphi](\mathbf{x}) = \int_{\partial \Omega} \boldsymbol{G}^{k}(\mathbf{x} - \mathbf{y}) \varphi(\mathbf{y}) \, \mathrm{d}s(\mathbf{y}).
	\end{align}
	Moreover, the single-layer potential operator satisfies the following jump relation on the boundary:
	\begin{equation}\label{eq:jump}
		\left.\frac{\partial\left({\mathcal{S}}_{\partial \Omega}^{k}[\varphi]\right)}{\partial \nu}\right|_{ \pm}(\mathbf{x})=\left( \pm \frac{1}{2} \mathcal{I}+\left({\mathcal{K}}_{\partial \Omega}^{k, *}\right)\right)[\varphi](\mathbf{x}).
	\end{equation} where
	\begin{align}\label{K_k*}
		\mathcal{K}_{\partial \Omega}^{k,*}[\varphi](\mathbf{x}) = \text{p.v.} \int_{\partial \Omega} \frac{\partial \boldsymbol{G}^{k}(\mathbf{x}-\mathbf{y})}{\partial \nu(\mathbf{x})} \varphi(\mathbf{y}) \, \mathrm{d}s(\mathbf{y}), \quad \mathbf{x} \in \partial \Omega.
	\end{align}
	The operator $\mathcal{K}_{\partial \Omega}^{k,*}: H^{-1/2}(\partial \Omega) \to H^{-1/2}(\partial \Omega)$ is known as the Neumann-Poincar\'e (NP) operator. The notation ``p.v." represents the Cauchy principal value. For further details on the mapping properties of the operators introduced above, we refer the readers to \cite{JCN,HH}.

	According to the layer potential theory described above, we can seek a solution to the system $\eqref{eq:system}$ in the following form
	\begin{equation}\label{eq:layer}
		u=\begin{cases}
			\mathcal S^{k_b}_{\partial \Omega}[\varphi_b](\mathbf{x}),  &\quad \mbox{in}\quad \Omega,\\
			\mathcal S^{k}_{\partial \Omega}[\varphi](\mathbf{x})+u^i,   &\quad \mbox{in}\quad \mathbb R^{d} \setminus \overline{\Omega},\\
		\end{cases}
	\end{equation}	
	where the density functions $\varphi_b,\varphi \in \mathit L^2(\partial \Omega)$. By utilizing the transmission conditions  $\eqref{eq:system}$
	and the jump relation $\eqref{eq:jump}$, we can obtain the following system of boundary integral equations:
	\begin{equation}\label{eq:5}
		\mathcal{A}(\omega, \delta)[\Phi](\mathbf{x})=F(\mathbf{x}), \quad \mathbf{x} \in \partial \Omega
	\end{equation}
	where 
	\begin{equation}\label{eq:integraleq}
		\begin{array}{c}
			\mathcal{A}(\omega, \delta)=\left(\begin{array}{cc}
				{\mathcal{S}}_{\partial \Omega}^{k_b} & -\mathcal{S}_{\partial \Omega}^{k} \\
				-\frac{I}{2}+{\mathcal{K}}_{\partial \Omega}^{k_b, *} & -\delta\left(\frac{I}{2}+\mathcal{K}_{\partial \Omega}^{k, *}\right)
			\end{array}\right), \\
			\Phi=\left(\begin{array}{c}
				\varphi_b \\
				\varphi
			\end{array}\right),\quad F=\left(\begin{array}{c}
				{u}^i \\
				\delta\partial_\nu u^i
			\end{array}\right).
		\end{array}
	\end{equation}
	It is worth noting that the scattering problem $\eqref{eq:system}$ that is equivalent to the boundary integral equations $\eqref{eq:5}$. 
	
	Throughout the paper, we define \(B_R\) as a three-dimensional ball or a two-dimensional disk of radius \(R \in \mathbb{R}_+\), centered at the origin. The boundary of \(B_R\) is denoted by \(\mathbb{S}_R := \partial B_R\), and we let \(\mathbb{S} := \mathbb{S}_1\) represent the unit sphere or circle. Furthermore, for a given bounded domain $\Omega$ in $\mathbb R^d$, we define two regions associated with \(\partial\Omega\). For any given \(R \in \mathbb{R}_+\) such that \(\Omega \subset B_R\) and  any given constants \(\zeta_i \) with $0<\zeta_i < 1$ $(i=1,2)$,  we define 
	\begin{equation}\label{df:bregion}
		\begin{aligned}
			\mathcal{N}_{1,\zeta_1}\left(\partial\Omega\right)&:=\{\mathbf{x}\in\Omega;\mathrm{dist}\left(\mathbf{x},\partial\Omega\right)<\zeta_1\},\\
			\mathcal{N}_{2,\zeta_2}\left(\partial\Omega\right)&:=\{\mathbf{x}\in B_R\setminus\overline{\Omega};\mathrm{dist}\left(\mathbf{x},\partial\Omega\right)<\zeta_2\},
		\end{aligned}
	\end{equation}
	where the parameters \(\zeta_1\) and \(\zeta_2\) quantifies the neighborhood of  the boundary. To rigorously define the quasi-Minnaert resonances, we first introduce two foundational concepts of boundary localization and surface resonance for the scattering problem \eqref{eq:system}.

	\begin{defn}\label{df:blocalization}
		Consider the acoustic scattering system \eqref{eq:system} induced by the incident wave $u^i$ satisfying \eqref{eq:u_i}, which generates the internal total field $u|_{\Omega }$ and the external scattered field $u^s|_{\mathbb{R}^d \setminus \overline{\Omega}}$. If there exist sufficiently small $\epsilon  \in \mathbb R_{+}$, and  small parameters \(\zeta_1\) and \(\zeta_2\), such that 
		\begin{equation}\label{eq:ibl}
			\frac{\|u\|_{L^2\left(\Omega \setminus \mathcal{N}_{1,\zeta_1}\left(\partial\Omega\right)\right)}}{\|u\|_{L^2\left(\Omega\right)}} \leq \epsilon,\quad
			\frac{\|u^s\|_{L^2\left(B_{R}\setminus\left(\mathcal{N}_{2,\zeta_2}\left(\partial\Omega\right)\cup \Omega\right)\right)}}{\|u^s\|_{L^2\left(B_R \setminus\Omega\right)}}\leq \epsilon,
		\end{equation}
		then the internal total field $u|_{\Omega}$ and the external scattered field $u|_{\mathbb{R}^d \setminus \overline{\Omega}}$ are referred to as internally and externally boundary localized, respectively. The level of the boundary localization is characterized by the parameter $\epsilon$.
	\end{defn}

	The following definition is related to the surface resonance for the internal total and external wave field associated with \eqref{eq:system}. 
	\begin{defn}\label{df:bresonance}
		Consider the acoustic scattering system \eqref{eq:system} induced by the incident wave $u^i$ satisfying \eqref{eq:u_i}. We say that the internal total field $u|_{\Omega}$ and the external scattered field $u^s|_{\mathbb{R}^{d}\setminus\Omega}$ exhibit surface resonance if there exists a non-trivial  incident wave $u^i$ such that
		\begin{equation}
			\frac{\|\nabla u\|_{L^2\left(\mathcal{N}_{1,\zeta_1}\left(\partial\Omega\right)\right)}}{\| u^{i}\|_{L^2\left(\Omega\right)}}\gg 1,\quad \frac{\|\nabla u^s\|_{L^2\left(	\mathcal{N}_{2,\zeta_2}\left(\partial\Omega\right)\right)}}{\| u^{i}\|_{L^2\left(\Omega\right)}}\gg 1.
		\end{equation}
		where $\mathcal{N}_{1,\zeta_1}\left(\partial\Omega\right)$ and $\mathcal{N}_{2,\zeta_2}\left(\partial\Omega\right)$ are defined by \eqref{df:bregion} with small parameters $\zeta_i\in \mathbb R_+$ ($i=1,2$). 
	\end{defn}

	\begin{rem}\label{rem:21}
		In this paper, we investigate acoustic scattering in the sub-wavelength regime. The diameter of the high-contrast medium $\Omega$ is significantly smaller than the wavelength of the incident wave, as characterized by \eqref{eq:sw}. Consequently,  the $L^2$-norm of the incident wave $u^i$ in $\Omega$ is small. To facilitate analysis, it is advantageous to normalize the $L^2$-norm of $u^i$ in $\Omega$. Definition~\ref{df:blocalization} relies on the conditions specified in \eqref{eq:ibl}. Notably, the ratios in \eqref{eq:ibl} remain invariant when both the numerator and denominator are divided by the $L^2$-norm of the incident wave $u^i$ in $\Omega$.
		
		We emphasize that surface resonance is characterized by the blow-up of the gradient norms of the internal total field $u|_{\Omega}$ and the external scattered field $u^s|_{\mathbb{R}^d \setminus \Omega}$ in a neighborhood of the boundary of the high-contrast structure $\Omega$, relative to the $L^2$-norm of the incident wave in $\Omega$. This phenomenon manifests as highly oscillatory behavior of $u|_{\Omega}$ and $u^s|_{\mathbb{R}^d \setminus \Omega}$ near the boundary of $\Omega$. Both boundary localization and surface resonance, as introduced in Definitions~\ref{df:blocalization} and~\ref{df:bresonance}, are defined relative to the $L^2$-norm of the incident wave $u^i$ in the high-contrast medium $\Omega$.
	\end{rem}

	\iffalse
	It is emphasized that if the surface resonance occurs, then it characterizes the blow-up property of the norm of gradient of the generated internal total field \( u|_{\Omega} \) and external scattered field \( u^s|_{\mathbb{R}^{d} \setminus \Omega} \) in the neighborhood of the boundary of a high-contrast structure \( \Omega \), compared with the $L^2$-norm incident wave in \( \Omega \). This phenomenon reveals the highly oscillatory behavior of \( u|_{\Omega} \) and \( u^s|_{\mathbb{R}^{d} \setminus \Omega} \) near the boundary of \( \Omega \).
	\fi

	Based on Definition \ref{df:blocalization} and Definition \ref{df:bresonance}, we introduce the definition of the quasi-Minnaert resonance.

	\begin{defn}\label{df:quasi-minnaert}
		Under the sub-wavelength regime, we say that a quasi-Minnaert resonance occurs if there exists a non-trivial incident wave \( u^{{i}} \) such that the internal total field \( u|_{\Omega} \) and the external scattered field \( u^s|_{\mathbb{R}^d \setminus \overline{D}} \), governed by the acoustic scattering system \eqref{eq:system}, simultaneously satisfy boundary localization in Definition \ref{df:blocalization} and surface resonance in Definition \ref{df:bresonance}.
		The operating frequency \( \omega \) corresponding to such an incident wave is defined as the quasi-Minnaert resonance frequency, and the high-contrast medium \( \Omega \) that guarantees this phenomenon is classified as the quasi-Minnaert resonator.
	\end{defn}

	\begin{rem}
		The quasi-Minnaert resonance introduced by Definition \ref{df:quasi-minnaert} differs from the Minnaert resonance \cite{ammari} in two principal aspects. First, the Minnaert resonance only depends on the physical configuration of the high-contrast medium \( \Omega \), whereas the excitation of the quasi-Minnaert resonance additionally depends on the specific chocie of the incident wave \( u^i \). Second, the Minnaert resonance frequency \( \omega \) occurs at discrete intervals and depends on the physical parameters of the high-contrast medium \( \Omega \). For instance, in the three-dimensional case, the resonance frequency \( \omega \) satisfies the relationship \( \frac{\omega}{\sqrt{\delta}} = \mathcal{O}(1) \) (cf. \cite[Theorem 3.1]{ammari}), where \( \delta \) denotes the density contrast parameter. Conversely, the quasi-Minnaert resonance frequencies exhibit a continuous spectral spectrum in the sub-wavelength regime.
	\end{rem}

	\section{Quasi-Minnaert resonances in $\mathbb R^3$}\label{sec3}

	In this section, we examine the case where the high-contrast medium $\Omega$ exhibits radial geometry in $\mathbb{R}^3$. After normalization, we assume that $\Omega$ is the unit ball in $\mathbb{R}^3$. We establish boundary localization for the internal total field $u|_{\Omega}$ and the external scattered field $u^s|_{\mathbb{R}^3 \setminus \Omega}$ in Theorems~\ref{3dthm1} and~\ref{3dthm2}, respectively, by appropriately selecting the incident wave for a fixed high-contrast density ratio $\delta$. Conversely, Corollary~\ref{cor:3.2} demonstrates that a prescribed level of boundary localization can be achieved by suitably choosing both the density ratio $\delta$ and the corresponding incident wave. Furthermore, in Theorem~\ref{thm:gradient_u}, we prove that the internal total field $u$ and the external scattered field $u^s$ exhibit surface resonance when the incident wave is carefully selected for a fixed high-contrast structure. Similarly, surface resonance can be induced by varying the density ratio $\delta$ while adapting the incident wave accordingly. Finally, we show that quasi-Minnaert resonances can emerge through the interplay between the high-contrast density ratio $\delta$ and the incident wave. Our analysis relies on layer potential theory and asymptotic expansions of Hankel and Bessel functions.

	In the following discussion, we introduce some spectral properties of layer potentials along with relevant notations. Let $Y_n^m$, with $n \in \mathbb{N}_0 := \mathbb{N} \cup \{0\}$ and $-n \leq m \leq n$, denote the spherical harmonic functions, which form an orthogonal basis for $L^2(\mathbb{S})$. Let $j_n(t)$ and $h_n(t)$ denote the spherical Bessel function and the spherical Hankel function of the first kind of order $n$, respectively. For a fixed $n \in \mathbb{N}$ and $0 < |z| \ll 1$, the following asymptotic expansions hold (cf.\ \cite{NIST}):
	\begin{equation}\label{eq:j_n_eps}
		j_n(z) = \frac{z^n}{(2n+1)!!} \left(1 - \frac{z^2}{2(2n+3)} + \mathcal{O}\left(\frac{z^4}{n^2}\right)\right),
	\end{equation}
	and
	\begin{equation}\label{eq:h_n_eps}
		h_n(z) = \frac{(2n-1)!!}{\mathrm{i} z^{n+1}} \left(1 - \frac{z^2}{2(-2n+1)} + \mathcal{O}\left(\frac{z^3}{n^2}\right)\right).
	\end{equation}
	Let $f_n(z)$ denote either $j_n(z)$ or $y_n(z)$, where $y_n(z)$ is the spherical Bessel function of the second kind of order $n$. The following recursive relations will be useful in our analysis:
	\begin{equation}\label{eq:recursive_eq}
		\begin{aligned}
			f_n'(z) &= f_{n-1}(z) - \frac{n+1}{z} f_n(z), \\
			f_n'(z) &= -f_{n+1}(z) + \frac{n}{z} f_n(z).
		\end{aligned}
	\end{equation}

	%%%%%%%%%%%%%%%%%%%
	\iffalse
	
	In the following discussion, we introduce some spectral properties of layer potentials and relevant notations. Let  $Y_n^m$ with $n \in \mathbb{N}_0:=\mathbb{N}\cup\{0\}$ and $-n\leq m\leq n$ be the spherical harmonic functions, which forms an orthogonal basis of $L^2\left(\mathbb{S}\right)$. Let $j_n(t)$ and $h_n(t)$ denote the spherical Bessel function and Hankel function of the first kind of order $n$, respectively. For a fixed $n\in \mathbb{N}$ with $0<|z|\ll1$, it holds that (cf.\cite{NIST})
	\begin{equation}\label{eq:j_n_eps}
		j_n(z)=\frac{z^n}{(2n+1)!!}\left(1-\frac{z^2}{2(2n+3)}+\mathcal{O}(\frac{z^4}{n^2})\right),
	\end{equation}
	and 
	\begin{align}\label{eq:h_n_eps}
		&h_{n}(z) =\frac{(2n-1)!!}{\text{i}z^{n+1}}\left(1-\frac{z^2}{2(-2n+1)}+\mathcal{O}(\frac{z^3}{n^2})\right).
	\end{align}
	Let $f_n(z)$ denote either  $j_n(z)$ or $y_n(z)$, where $y_n(z)$ is the  Bessel  function of the second kind of order $n$.  In the following analysis, we need some recursive relationships:
	\begin{equation}
		\begin{array}{l}\label{eq:recursive_eq}
			f_n^{\prime}(z)=f_{n-1}(z)-((n+1) / z) f_n(z), \\
			f_n^{\prime}(z)=-f_{n+1}(z)+(n / z) f_n(z).
		\end{array}
	\end{equation}

	\begin{rem}
		By integrating asymptotic expansions \eqref{eq:j_n_eps} and \eqref{eq:h_n_eps} with recursive relation \eqref{eq:recursive_eq}, one can derive the corresponding asymptotic expansions for $j^{\prime}_n\left(z\right)$ and 
		$h^{\prime}_n\left(z\right)$.
	\end{rem}
	
	\fi 
	
	%%%%%%%%%%%%%5

	In this section, we focus on the case that  the high-contrast medium $\Omega$ is the unit ball.  Lemmas~\ref{lem:3d_s} and~\ref{lem:3d_k} present the spectral properties of the single-layer operator $\mathcal{S}^{k}_{\partial \Omega} $ and Neumann--Poincar\'e operator $\mathcal{K}_{\partial \Omega}^{k,*}$  defined in \eqref{eq:S_k} and \eqref{K_k*}, respectively. These results can be found in~\cite{XYC}.

	\begin{lem}\label{lem:3d_s}
		The eigensystem of the single-layer operator $\mathcal{S}_{\partial \Omega }^k$ defined in \eqref{eq:S_k} is given as follows
		\begin{equation}
			\mathcal{S}_{\partial \Omega }^k\left[Y_n^m\right](\mathbf{x})=-\mathrm{i} k  j_n(k ) h_n(k ) Y_n^m, \quad \mathbf{x} \in \partial \Omega,
		\end{equation}
		Moreover, the following two identities hold
		\begin{equation}
			\mathcal{S}_{\partial\Omega }^{k_b}\left[Y_n^m\right](\mathbf{x})=-\mathrm{i} k  j_n(k_b|\mathbf{x}|) h_n(k_b ) Y_n^m, \quad \mathbf{x} \in \Omega ,
		\end{equation}
		and
		\begin{equation}
			\mathcal{S}_{\partial \Omega }^k\left[Y_n^m\right](\mathbf{x})=-\mathrm{i} k  j_n(k ) h_n(k|\mathbf{x}|) Y_n^m, \quad \mathbf{x} \in \mathbb{R}^3 \backslash \Omega,
		\end{equation}
	\end{lem}
	%\begin{lem}\label{lem:3d_k_r}
	%	The spectral of $(\mathcal{K}_{B_{R}}^{k})^{*}$ 
	%	is the same with the spectral of $(\mathcal{K}_{B_1}^{kR})^*$.
	%\end{lem}
	\begin{lem}\label{lem:3d_k}
		The eigensystem of the N-P operator $\mathcal{K}_{\Omega }^{k,*}$ defined in \eqref{K_k*} is given as follows
		\begin{equation}
			\mathcal{K}_{\partial\Omega }^{k,*}[Y_n^m]=\lambda_{1,n,m}Y_n^m,\quad \mathbf{x} \in \partial \Omega.
		\end{equation}
		where the eigenvalues satisfy:
		\begin{equation}
			\lambda_{1,n,m}=-\frac{1}{2}-ik^2j_n(k)h_n^{(1)^{\prime}}(k)=\frac{1}{2}-ik^2j_n'(k)h_n^{(1)}(k).
		\end{equation}
	\end{lem}

	In this section, we consider the incident wave \( u_n^i \) given by
	\begin{equation}\label{eq:3d_ui}
		u^i_n(\mathbf{x}) = \sum_{m=-n}^{n} \beta_n^m j_n(k|\mathbf{x}|) Y_n^m(\hat{\mathbf{x}}),
	\end{equation}
	which is an entire solution to \eqref{eq:u_i}. Here, \(\left(\beta_n^{-n}, \cdots, \beta_n^n\right) \in \mathbb{C}^{2n+1}\) is a non-trivial constant vector.  We emphasize that the index \(n\), which characterizes \(u^i_n\), plays a crucial role in achieving the boundary localization of the resulting internal total wave field associated with the high-contrast medium \(\Omega\) and the incident wave \(u_n^i\), as established in Theorem~\ref{3dthm1}. Throughout the remainder of this paper, we denote by $u_n$ and $u_n^s$ the internal total wave field and the external scattered wave field, respectively, associated with the incident wave $u_n^i$.  In the following theorem, the boundary localization level \(\epsilon\) and the high-contrast density parameter \(\delta\) are prescribed. We denote $\lceil\cdot \rceil$ by the ceiling function.

	\begin{thm}\label{3dthm1}
		Consider the acoustic scattering problem in \( \mathbb{R}^3 \) as described by the system in \eqref{eq:system}, where the high-contrast medium \( \Omega \) is the unit ball. Under the assumptions given in \eqref{eq:sw} and \eqref{eq:wn1}, for any fixed \( \gamma_1 \in (0, 1) \) and sufficiently small \( \epsilon > 0 \), there exists an incident wave \( u^i_n \) defined in \eqref{eq:3d_ui} with \( n > n_1 \), where \( n_1 \) is defined as
		\begin{align}\label{n_1}
			n_1 =  \lceil\frac{1}{2}\left(\frac{\ln{\epsilon}}{\ln{\gamma_1}} - 3 \right)\rceil+1,
		\end{align}
		For this choice of \( n \), the following estimate holds:
		\begin{align}\label{eq:thm 31 ratio}
			\frac{\| u_n \|^2_{L^2\left(\Omega \setminus \mathcal{N}_{1,1-\gamma_1}(\partial \Omega)\right)}}{\| u_n \|^2_{L^2(\Omega)}} \leq  \mathcal{O}(\epsilon) + \mathcal{O}(\epsilon \omega^2) \ll 1,
		\end{align}
		where \( \mathcal{N}_{1, 1-\gamma_1}(\partial \Omega) \) is  defined in \eqref{df:bregion} and \( u_n \)  is the internal field to \eqref{eq:system} associated with $u_n^i$ and $\Omega$. This implies that, for sufficiently small \( \epsilon \), the internal field \( u_n \) is sharply localized near the boundary \( \partial \Omega \), and its contribution within the bulk of \( \Omega \) is negligible.
		
		Furthermore, the internal total field \( u_n \) in \( \Omega \) can be represented as a spherical harmonic expansion:
		\begin{equation}\label{internalus}
			u_n(\mathbf{x}) = \sum_{m=-n}^{n} \Psi_{1,n,m}(|\mathbf{x}|) Y_n^m(\hat{\mathbf{x}}),
		\end{equation}
		where  the coefficients \( \Psi_{1,n,m}(|\mathbf{x}|)  \) are given by:
		\[
		\Psi_{1,n,m}(|\mathbf{x}|) = -ik_b j_n(k_b |\mathbf{x}|) h_n(k_b) \left( -\frac{\delta(2n + 1)\beta_n^m k^n}{[\delta(n+1) + n](2n - 1)!!} + \mathcal{O} \left( \frac{\delta \beta_n^m k^{n+2}}{(2n + 3)!!} \right) \right).
		\]
	\end{thm}

	%%%%%%%%%
	\iffalse
	\begin{thm}\label{3dthm1}
		Consider the acoustic scattering problem \eqref{eq:system} in $\mathbb{R}^3$, where the high-contrast medium $\Omega$ is a unit ball. Under the assumptions \eqref{eq:sw}-\eqref{eq:wn1}, for any $\gamma_1 \in(0,1)$ and a sufficiently small $\epsilon > 0$, there exists an incident wave defined in \eqref{eq:3d_ui} with $n>n_1$, where
		\begin{align}\label{n_1}
			n_1 =  \lceil\frac{1}{2}\left(\frac{\ln{\epsilon}}{\ln{\gamma_1}} - 3 \right)\rceil+1,
		\end{align}
		such that
		\begin{equation*}
			\frac{\|u_n\|_{L^2\left(\Omega \setminus \mathcal{N}_{1,1-\gamma_1}\left(\partial\Omega\right)\right)}}{\|u_n\|^2_{L^2\left(\Omega\right)}} = \mathcal{O}\left(\epsilon\right)+\mathcal{O}\left(\epsilon\omega^2\right)\ll 1.
		\end{equation*}
		where  $\mathcal{N}_{1,1-\gamma_1}\left(\partial\Omega\right)$ is defined in  \eqref{df:bregion}. Moreover, the internal total field $u$ in $\Omega$ is represented as
		\begin{equation}\label{internalus}
			\begin{split}
				u_n(\mathbf{x}) =
				\sum_{m=-n}^{n} \Psi_{1,n,m}\left(\mathbf{x}\right) Y_n^m(\hat{\mathbf{x}})
			\end{split}
		\end{equation}
		where the coefficients are given by 
		\begin{equation*}
			\Psi_{1,n,m}\left(\mathbf{x}\right) = -ik_b j_n(k_b |\mathbf{x}|) h_n(k_b) \left(-\frac{\delta\left(2n + 1\right)\beta_n^m k^n}{\left[\delta(n + 1) + n\right](2n - 1)!!} + \mathcal{O}\left(\frac{\delta \beta_n^m k^{n + 2}}{(2n + 3)!!}\right) \right).
		\end{equation*}
	\end{thm}
	\fi
	%%%%%%%%%%%

	\begin{proof}

		Since $Y_n^m(\hat{\mathbf{x}})$ forms an orthogonal basis of $L^2(\partial \Omega)$, according to \eqref{eq:layer}, the densities \(\varphi_b\) and \(\varphi\), which define the internal total and external scattered wave fields, respectively, can be expressed as
		%According to the layer potential theory, we need to use the density function defined in \eqref{eq:5} to obtain the asymptotic expression of the internal total field. Utilizing the orthogonality of the functions $Y_n^m(\hat{\mathbf{x}})$, we choose the density function  as follows
		\begin{align}
			\varphi_b(\hat{\mathbf{x}}) & =\sum_{m=-n}^{n} \varphi_{b, n,m} Y_n^m(\hat{\mathbf{x}}), \label{eq:phi_b}\\
			\varphi(\hat{\mathbf{x}}) & =\sum_{m=-n}^{n} \varphi_{n,m} Y_n^m(\hat{\mathbf{x}}),\label{eq:phi}
		\end{align}
		where $\varphi_{b, n,m}$ and $\varphi_{n,m}$ are to be determined. 
		By using Lemma \ref{lem:3d_s} and \ref{lem:3d_k}, substituting \eqref{eq:3d_ui}, \eqref{eq:phi_b} and \eqref{eq:phi} into the integral equation \eqref{eq:5}, we can deduce that
		\begin{equation}\label{subs:2.28}
			\mathbf{A}_{n}\mathbf{X}_{n,m}=\mathbf{b}_{n,m}
		\end{equation}
		where
		\begin{equation*}
			\mathbf{A}_{n}=\begin{pmatrix}
				a_{11}^{n}& a_{12}^{n}\\
				a_{21}^{n}& a_{22}^{n}
			\end{pmatrix},\quad \mathbf{X}_{n,m}=\begin{pmatrix}
				\varphi_{b,n,m}\\
				\varphi_{n,m}
			\end{pmatrix},\quad
			\mathbf{b}_{n,m}=\begin{pmatrix}
				\beta_n^mj_n(k)\\
				\delta k\beta_n^mj^{'}_n(k).
			\end{pmatrix}
		\end{equation*}
		with
		\begin{equation*}
			\begin{aligned}
				a_{11}^{n}&=-ik_bj_n(k_b)h_n(k_b),\quad a_{12}^{n}=ikj_n(k)h_n(k),\\
				a_{21}^{n}&=-ik_b^2j^{\prime}_n(k_b)h_n(k_b), \quad a_{22}^{n}=i\delta k^2j_n(k)h_n^{\prime}(k).\\
			\end{aligned}
		\end{equation*}
		According to the assumptions \eqref{eq:delta} and \eqref{eq:wn1},  by using the asymptotic expansions \eqref{eq:j_n_eps} and \eqref{eq:h_n_eps} as well as the recurrence formula \eqref{eq:recursive_eq}, 
		the coefficients $a_{12}^{n}$ and $a_{22}^{n}$ can be asymptotically expanded in terms of $k$. Similarly, since $k_b = \tau k$ and $\tau = \mathcal{O}\left(1\right)$, the coefficients  $a_{11}^{n}$ and $a_{21}^{n}$  can also be asymptotically expanded with respect to $k$. Thus, the system $\eqref{subs:2.28}$ can be equivalently expressed as follows
		\begin{equation}\label{eq:2.31}
			\left(\mathbf{A}_{\mathbf{0}}^{n}+k^2 \mathbf{A}_{\mathbf{2}}^{n}+k^3\mathbf{A}_{\mathbf{3}}^{n}\right) \mathbf{X}_{n,m}=\mathbf{b}_{n,m},\quad n \ge 1
		\end{equation}
		where 
		\begin{equation}\label{eq:2.32}
			\begin{aligned}
				\mathbf{A}_{\mathbf{0}}^{n} &= \frac{1}{2n+1}\begin{pmatrix}
					-1 & 1\\
					-n & -\delta(n+1)
				\end{pmatrix},\\
				\mathbf{A}_{\mathbf{2}}^{n} &= \frac{1}{(2n+1)(2n+3)(1-2n)}\begin{pmatrix}
					2\tau^2 & -2\\
					(4n-1)\tau^2 &
					-4\delta n
				\end{pmatrix},\\
				\mathbf{A}_{\mathbf{3}}^{n} &= \begin{pmatrix}
					\mathcal{O}\left(\frac{\tau^3}{n^2}\right) & \mathcal{O}\left(\frac{1}{n^2}\right)\\
					\mathcal{O}\left(\frac{\tau^3}{n}\right) &
					\mathcal{O}\left(\frac{\delta}{n}\right)
				\end{pmatrix}.
			\end{aligned}
		\end{equation}
		It is clear that ${\rm det}(\mathbf{A}_{\mathbf{0}}^{n})=\frac{\delta(n+1)+n}{(2n+1)^2}\neq 0$ for $n\in \mathbb{N}$. Consequently, we get
		\begin{equation}\label{eq:2.33}
			(\mathbf{A}_{\mathbf{0}}^{n})^{-1} = \frac{(2n+1)}{\delta(n+1)+n}\begin{pmatrix}
				-\delta(n+1) & -1 \\
				n & -1
			\end{pmatrix}.
		\end{equation}
		Multiplying both sides of the equation $\eqref{eq:2.31}$ simultaneously by  $(\mathbf{A}_{\mathbf{0}}^{n})^{-1}$, 
		using $\eqref{eq:2.32}$ and $\eqref{eq:2.33}$, we can further deduce that
		\begin{align}
			\mathbf{X}_{n,m}&=\left(\mathbf{I}+k^2(\mathbf{A}_{\mathbf{0}}^{n})^{-1}\mathbf{A}_{\mathbf{2}}^{n}+k^3(\mathbf{A}_{\mathbf{0}}^{n})^{-1}\mathbf{A}_{\mathbf{3}}^{n}\right)^{-1}(\mathbf{A}_{\mathbf{0}}^{n})^{-1}\mathbf{b}_{n,m}\notag\\
			&=(\mathbf{A}_{\mathbf{0}}^{n})^{-1}\mathbf{b}_{n,m}-k^2(\mathbf{A}_{\mathbf{0}}^{n})^{-1}\mathbf{A}_{\mathbf{2}}^{n}(\mathbf{A}_{\mathbf{0}}^{n})^{-1}\mathbf{b}_{n,m}+\mathcal{O}(k^3)(\mathbf{A}_{\mathbf{0}}^{n})^{-1}\mathbf{b}_{n,m}.\label{eq:x_n_m3d}		
		\end{align}
		where
		\begin{align*}
			(\mathbf{A}_{\mathbf{0}}^{n})^{-1}\mathbf{b}_{n,m}&=\frac{2n+1}{\delta(n+1)+n}\begin{pmatrix}
				\delta k\beta_n^mj_{n+1}(k)-(2n+1)\delta\beta_n^m j_n(k)\\
				\delta k\beta_n^mj_{n+1}(k)+(1-\delta)n\beta_n^mj_{n}(k)
			\end{pmatrix},\\
			(\mathbf{A}_{\mathbf{0}}^{n})^{-1}\mathbf{A}_{\mathbf{2}}^{n,m}(\mathbf{A}_{\mathbf{0}}^{n})^{-1}\mathbf{b}_{n,m}&=\left[\frac{2n+1}{\delta(n+1)+n}\right]^2\begin{pmatrix}
				c_{n,m}\\
				d_{n,m}
			\end{pmatrix},
		\end{align*}
		and 
		\begin{align*}
			c_{n,m} &= \frac{2(\tau^2-1)\delta k\beta_n^mj_{n+1}(k)+2[(1-\delta)n-(2n+1)\tau^2\delta]\beta_nj_{n}(k)}{(1-2n)(2n+1)(2n+3)},\\
			d_{n,m} &= \frac{(\delta-\tau^2)(4n-1)\delta k\beta_n^mj_{n+1}(k)+[(\delta-1)n+(2n+1)\tau^2](4n-1)\delta\beta_nj_{n+1}(k)}{(1-2n)(2n+1)(2n+3)}.
		\end{align*}
		Substituting  \eqref{eq:j_n_eps} and \eqref{eq:h_n_eps} into \eqref{eq:x_n_m3d}, we obtain
		\begin{equation}\label{eq:3d_density}
			\mathbf{X}_{n,m}=\begin{pmatrix}
				-\frac{\delta\left(2n+1\right)\beta_n^m k^n}{\left[\delta(n+1)+n\right](2n-1)!!}\\
				-\frac{n(1-\delta)\beta_n^m k^n}{\left[\delta(n+1)+n\right](2n-1)!!}
			\end{pmatrix}+                                     
			\begin{pmatrix}
				\mathcal{O}\left(\frac{\delta\beta_n^m k^{n+2}}{(2n+3)!!}\right)\\
				\mathcal{O}\left(\frac{\delta\beta_n^m k^{n+2}}{(2n+3)!!}\right)
			\end{pmatrix}.
		\end{equation}

		%	In order to calculate $	\frac{\|u\|^2_{L^2\left(B_{\gamma_1}\right)}}{\|u\|^2_{L^2\left(\Omega\right)}}$, 
		
		In the following,  we shall derive  the asymptotic expansion of the internal total field $u$ defined in \eqref{eq:layer}. Substituting the density function defined in \eqref{eq:phi_b} along with the coefficient \eqref{eq:3d_density} into Lemma \ref{lem:3d_s}, and then applying \eqref{eq:j_n_eps} as well as \eqref{eq:h_n_eps}, it is clear that 
		\begin{align}\label{eq:3d_u_norm}
			u_n(\mathbf{x})&=\mathcal{S}_{\partial \Omega }^{k_b}\left[\varphi_b\right](\mathbf{x})\notag\\
			&=
			\sum_{m=-n}^{n} -i k_b \delta k^n j_n(k_b|\mathbf{x}|)h_n(k_b)\notag\\
			&\times\left(-\frac{\left(2n+1\right)\beta_n^m }{\left[\delta(n+1)+n\right](2n-1)!!}+\mathcal{O}\left(\frac{\beta_n^m k^{2}}{(2n+3)!!}\right)\right) Y_n^m(\hat{\mathbf{x}})\notag\\
			&=
			\sum_{m=-n}^{n} \frac{\delta k^n\vert \mathbf{x}\vert^n}{2n+1}\left(-1+\mathcal{O}\left(\tau^2k^2\right)\right)\notag\\
			&\times\left(-\frac{\delta\left(2n+1\right)\beta_n^m k^n}{\left[\delta(n+1)+n\right](2n-1)!!}+\mathcal{O}\left(\frac{\beta_n^m k^{2}}{(2n+3)!!}\right)\right) Y_n^m(\hat{\mathbf{x}})\notag\\
			&=
			\sum_{m=-n}^{n} \frac{\delta\beta_n^m k^n\vert \mathbf{x}\vert^n}{(2n-1)!!}\left(\frac{1}{\left(\delta(n+1)+n\right)^2}+\mathcal{O}\left(\frac{ k^{2}}{(2n+1)}\right)\right) Y_n^m(\hat{\mathbf{x}}).
		\end{align} 
		For any $t\in (0,1]$ and $n \ge 1$,  utilizing the assumptions \eqref{eq:sw}-\eqref{eq:wn1} and  \eqref{eq:3d_u_norm}, we can calculate that
		\begin{align}
			&	\|u_n\|^2_{L^2\left(B_t\right)}\notag\\
			&=\int_{B_t}	\left|\sum_{m=-n}^{n}\frac{\delta\beta_n^m k^n\vert \mathbf{x}\vert^n}{(2n-1)!!}\left(\frac{1}{\left(\delta(n+1)+n\right)^2}+\mathcal{O}\left(\frac{ k^{2}}{(2n+1)}\right)\right) Y_n^m(\hat{\mathbf{x}})\right|^2 d\mathbf{x}\notag\\
			&=\int_{B_t} \sum_{m=-n}^{n} \frac{\delta^2\vert\beta_n^m \vert^2 k^{2n}\vert \mathbf{x}\vert^{2n}}{[(2n-1)!!]^2}\left(\frac{1}{\left(\delta(n+1)+n\right)^2}+\mathcal{O}\left(\frac{k^{2}}{\left(2n+1\right)^2}\right)\right) |Y_n^m(\hat{\mathbf{x}})|^2d\mathbf{x}\notag\\
			&=\sum_{m=-n}^{n}\int_{0}^{2\pi}\int_{0}^{\pi}\int_{0}^{t}\frac{\delta^2\vert\beta_n^m \vert^2 k^{2n}r^{2n+2}}{[(2n-1)!!]^2}\left(\frac{1}{\left(\delta(n+1)+n\right)^2}+\mathcal{O}\left(\frac{k^{2}}{\left(2n+1\right)^2}\right)\right)\notag\\
			&\times|Y_n^m(\hat{\mathbf{x}})|^2\sin\theta drd\theta d\varphi\notag\\
			&=K_{1,n,m}(t)\left(1+\mathcal{O}\left(k^2\right)\right)\label{eq:3du_internal},
		\end{align}
		where
		\begin{align*}
			K_{1,n,m}(t) &= \sum_{m=-n}^{n} \frac{\delta^2\vert\beta_n^m \vert^2 k^{2n}t^{2n+3}}{[\delta(n+1)+n]^2[(2n-1)!!]^2(2n+3)}.
		\end{align*}
		For $\gamma_{1} \in(0,1)$, it is easy to see that 
		$$
		\frac{K_{1,n,m}(\gamma_{1})}{K_{1,n,m}(1)}=\gamma_{1}^{2n+3}. 
		$$
		Since $n \geq  n_1$, we know that $\gamma_{1}^{2n+3} \leq \epsilon$, which implies that $\frac{K_{1,n,m}(\gamma_{1})}{K_{1,n,m}(1)} \leq \epsilon$. Therefore, one has
		\begin{align*}
			\frac{\|u_n\|^2_{L^2\left(\Omega \setminus \mathcal{N}_{1,1-\gamma_1}\left(\partial\Omega\right)\right)}}{\|u_n\|^2_{L^2\left(\Omega\right)}}
			&= \frac{K_{1,n,m}(\gamma_{1})}{K_{1,n,m}(1)}\left(1+\mathcal{O}\left(k^{2}\right)\right)\left(1-\mathcal{O}\left(k^{2}\right)\right)\\
			&\leq  \mathcal{O}\left(\epsilon\right)+\mathcal{O}\left(\epsilon\omega^2\right).
		\end{align*}	
		where the internal total field satisfies the internal localization defined by Definition \ref{df:blocalization}. 
		
		The proof is complete.
	\end{proof}

	Similar to Theorem~\ref{3dthm1}, the following theorem establishes that the external scattered field \(u^s_n|_{\mathbb{R}^{d} \setminus \Omega}\), induced by the incident field \(u_n^i\) defined in \eqref{eq:3d_ui}, exhibits boundary localization. This is achieved by appropriately selecting the index \(n\) of \(u_n^i\), given a prescribed boundary localization level \(\epsilon\) and high-contrast density ratio \(\delta\).

	\begin{thm}\label{3dthm2}
		Consider the acoustic scattering problem \eqref{eq:system} in $\mathbb{R}^3$, and suppose that $B_R$ such that $\Omega \Subset B_R$ with $R\in (2,\infty)$, where the high-contrast medium \( \Omega \) is the unit ball. Under the assumptions \eqref{eq:sw}-\eqref{eq:wn1}, for any fixed $\gamma_2 \in(1,2)$ and sufficiently small $\epsilon > 0$, there exists an incident wave defined in \eqref{eq:3d_ui} with $n \geq n_2$,  where 
		\begin{align}\label{n2}
			n_2=\lceil-\left(\frac{1}{2} \frac{\ln{\epsilon}}{\ln{\gamma_2}}+1 \right)\rceil+1,
		\end{align}
		such that 
		\begin{align}\label{eq:usll1}
			\frac{\|u^s_n\|^2_{L^2\left(B_{R}\setminus\left(\mathcal{N}_{2,\gamma_2-1}\left(\partial\Omega\right)\cup \Omega\right)\right)}}{\|u^s_n\|^2_{L^2\left(B_R\setminus \Omega\right)}}\leq \mathcal{O}\left(\epsilon\right)+\mathcal{O}\left(\epsilon\omega^2\right)\ll 1,
		\end{align}
		with  $\mathcal{N}_{2,\gamma_{2}-1}\left(\partial\Omega\right)$ being defined in  \eqref{df:bregion}. Moreover, the scattered field $u_n^s$ in $B_R\setminus \Omega$ is represented as: 
		\begin{align}\label{eq:us}
			\begin{split}
				u^s_n(\mathbf{x})
				=\sum_{m=-n}^{n} \Psi_{2,n,m}(|\mathbf{x}|) Y_n^m(\hat{\mathbf{x}}).
			\end{split}
		\end{align} 
		where the coefficients are given by
		\begin{align}
			\Psi_{2,n,m}(|\mathbf{x}|) = -ikj_n(k)h_n(k|\mathbf{x}|)
			\left(-\frac{n(1-\delta)\beta_n^m k^n}{\left[\delta(n+1)+n\right](2n-1)!!}+\mathcal{O}\left(\frac{\delta\beta_n^m k^{n+2}}{(2n+3)!!}\right)\right).
		\end{align}
	\end{thm}

	\begin{proof}
		
		Similar to the derivation of \eqref{eq:3d_u_norm}, by using \eqref{eq:layer}, \eqref{eq:3d_density}, we  obtain an asymptotic  expression for the external scattered field in the domain \(B_R \setminus \Omega\):
		
		%Following a similar approach to the derivation of \eqref{eq:3du_internal}, we obtain an explicit expression for the external scattered field, as defined by \eqref{eq:layer}, in the domain \(B_R \setminus \Omega\). First, we substitute the density function given in \eqref{eq:phi} and the coefficients described in \eqref{eq:3d_density} into Lemma~\ref{lem:3d_s}. Then, by applying the asymptotic expansions in \eqref{eq:j_n_eps} and \eqref{eq:h_n_eps}, we derive

		%Following a similar approach to the proof of \eqref{eq:3du_internal}, we derive the explicit expression for the external scattered field defined by \eqref{eq:layer} in the domain \(B_R \setminus \Omega\). First, we substitute the density function defined by \eqref{eq:phi} and the coefficient described in \eqref{eq:3d_density} into Lemma \ref{lem:3d_s}. Subsequently, with the application of \eqref{eq:j_n_eps} and \eqref{eq:h_n_eps}, we obtain
		\begin{align}
			u^s_n(\mathbf{x}) &=\mathcal{S}_{\partial B_1}^{k}\left[\varphi\right](\mathbf{x})\notag\\
			&=\sum_{m=-n}^{n} -ikj_n(k)h_n(k|\mathbf{x}|)\notag\\
			&\times\left(-\frac{n(1-\delta)\beta_n^m k^n}{\left[\delta(n+1)+n\right](2n-1)!!}+\mathcal{O}\left(\frac{\delta\beta_n^m k^{n+2}}{(2n+3)!!}\right)\right) Y_n^m(\hat{\mathbf{x}})\notag\\
			&=\sum_{m=-n}^{n}-\frac{\beta_n^m k^n}{(2n+1)\vert \mathbf{x}\vert^{n+1}}\left(1+\mathcal{O}\left(k^2\right)\right)\notag\\ &\times\left(-\frac{n(1-\delta)}{\left[\delta(n+1)+n\right](2n-1)!!}+\mathcal{O}\left(\frac{\delta k^{2}}{(2n+3)!!}\right)\right)
			Y_n^m(\hat{\mathbf{x}})\notag\\
			&= \sum_{m=-n}^{n} \frac{\beta_n^m k^n}{(2n+1)!!|\mathbf{x}|^{n+1}}\left(-\frac{n(1-\delta)}{\delta(n+1)+n}+\mathcal{O}\left(k^2\right)\right) Y_n^m(\hat{\mathbf{x}}).\label{eq:3d_us_norm}
		\end{align}
		
		Subsequently, we demonstrate that the external scattered field \(u^s\) exhibits boundary localization near the exterior boundary of the domain \(\Omega\). For any \(1 \leq t_2 < t_1 < +\infty\) and \(n \geq 1\), using the asymptotic expansion in \eqref{eq:3d_us_norm}, we obtain
		\begin{align}
			&\|u^s_n\|^2_{L^2\left(B_{t_1}\setminus B_{t_2}\right)}\notag\\
			&=\int_{B_t} \left|\sum_{m=-n}^{n} \frac{\beta_n^m k^n}{(2n+1)!!|\mathbf{x}|^{n+1}}\left(-\frac{n(1-\delta)}{\delta(n+1)+n}+\mathcal{O}\left(k^2\right)\right) Y_n^m(\hat{\mathbf{x}})\right|^2d\mathbf{x}\notag\\
			&=\int_{B_t}\sum_{m=-n}^{n}\frac{|\beta_n^m|^2 k^{2n}}{[(2n+1)!!]^2|\mathbf{x}|^{2n+2}}
			\left(\frac{\left[n(1-\delta)\right]^2}{\left(\delta(n+1)+n\right)^2}+\mathcal{O}\left(k^2\right)\right)|Y_n^m(\hat{\mathbf{x}})|^2d\mathbf{x}\notag\\
			&=\sum_{m=-n}^{n}\int_{0}^{2\pi}\int_{0}^{\pi}\int_{t_2}^{t_1}\sum_{m=-n}^{n}\frac{|\beta_n^m|^2 k^{2n}}{[(2n+1)!!]^2r^{2n}}
			\left(\frac{\left[n(1-\delta)\right]^2}{\left(\delta(n+1)+n\right)^2}+\mathcal{O}\left(k^2\right)\right)\notag\\
			&\times|Y_n^m(\hat{\mathbf{x}})|^2\sin\theta drd\theta d\varphi\notag\\
			&=L_{1,n,q}(t_1,t_2)\left(1+\mathcal{O}\left(k^2\right)\right)\label{eq:3dus_norm}
		\end{align}
		where
		\begin{align*}
			L_{1,n,m}(t_1,t_2)&=\sum_{m=-n}^{n}\frac{[n(1-\delta)]^2|\beta_n^m|^2 k^{2n}}{[\delta(n+1)+n]^2[(2n+1)!!]^2(2n-1)}\left(\frac{1}{t_2^{2n-1}}-\frac{1}{t_1^{2n-1}}\right).
		\end{align*}
		For $t_1=R$ and $\gamma_2 \in(1,R)$, the function $L_{1,n,m}(t_1,t_2)$ is monotonically increasing with respect to $t_2$, it is easy to see that
		\begin{align*}
			\frac{\|u^s_n\|^2_{L^2\left(B_{R}\setminus\left(\mathcal{N}_{2,\gamma_{2}-1}\left(\partial\Omega\right)\cup \Omega\right)\right)}}{\|u^s_n\|^2_{L^2\left(B_R\setminus \Omega\right)}}
			&=\frac{L_{1,n,m}(R,\gamma_2)}{L_{1,n,m}(R,1)}\left(1+\mathcal{O}\left(k^{2}\right)\right)\left(1-\mathcal{O}\left(k^{2}\right)\right)\\
			&\leq \frac{1}{\gamma_2^{2n+1}}\left(1+\mathcal{O}\left(k^{2}\right)\right)\left(1-\mathcal{O}\left(k^{2}\right)\right)\\
		\end{align*}
		Using $n \geq \lceil-\left(\frac{1}{2} \frac{\ln{\epsilon}}{\ln{\gamma_2}}+1 \right)\rceil+1$, we obtain $\frac{1}{\gamma_{2}^{2n-1}} < \epsilon$, from which we have
		\begin{align*}
			\frac{\|u^s_n\|^2_{L^2\left(B_{R}\setminus\left(\mathcal{N}_{2,\gamma_{2}-1}\left(\partial\Omega\right)\cup \Omega\right)\right)}}{\|u^s_n\|^2_{L^2\left(B_R\setminus \Omega\right)}}
			&\leq \mathcal{O}\left(\epsilon\right)+\mathcal{O}\left(\epsilon\omega^2\right).
		\end{align*}
		
		The proof is complete.
	\end{proof}
	
	\begin{rem}\label{rem:3d_bl}
		
		While Theorems~\ref{3dthm1} and~\ref{3dthm2} are derived under the assumption that the domain \(\Omega\) is the unit sphere, we believe that these results can be extended to more general geometries by utilizing high-contrast material parameters. Extensive numerical experiments indicate that boundary localization persists for high-contrast structures with more general shapes. The rigorous study of boundary localization in such geometries will be pursued in future work.
		
		Moreover, according to Theorems~\ref{3dthm1} and~\ref{3dthm2}, under assumptions \eqref{eq:sw}--\eqref{eq:wn1}, we select the incident wave defined in \eqref{eq:3d_ui} with \(n\) satisfying
		\begin{equation*}
			n \geq \max\left(n_1, n_2\right).
		\end{equation*}
		Consequently, both the internal total field and the external scattered field exhibit boundary localization, as defined in Definition~\ref{df:blocalization}. For a fixed high-contrast parameter \(\delta\) and prescribed boundary localization level \(\epsilon\), one can appropriately choose the incident wave \(u^{i}_{n}\) such that both the internal total field and the external scattered field \(u^s_n\) decay rapidly away from the boundary of the high-contrast medium \(\Omega\), and primarily propagate along the boundary.

		%%%%%%%%%%%%%%%%%%%%
		\iffalse
		While Theorems \ref{3dthm1} and \ref{3dthm2} were  derived under the assumption of a unit spherical domain, we believe that these results can be extended to more general geometries by high-contrast material parameters. Extensive numerical examples show that the boundary localizations hold for general shape high-contrast structures.  The study of boundary localization in more general geometries will be considered in subsequent work. Moreover, according to Theorem \ref{3dthm1} and Theorem \ref{3dthm2}, under the assumptions \eqref{eq:sw}-\eqref{eq:wn1}, we select the incident wave defined in \eqref{eq:3d_ui} with $n$ satisfying the condition
		\begin{equation*}
			n \geq \max\left(n_1, n_2\right).
		\end{equation*}
		Consequently, the internal total field and the external scattered field exhibit  boundary localization as shown in Definition \ref{df:blocalization} simultaneously. For a fixed high-contrast parameter $\delta$ and the level of boundary localization $\epsilon$, we choose an appropriate incident wave $u^{i}_{n}$ such that the total internal field $u^i_n$ and the external scattering field $u^s_n$, which decays rapidly  away the  boundary of the high-contrast medium $\Omega$ and mainly propagate  along the boundary.
		\fi
		%%%%%%%%%%%%%%%%%%%%%%%%%%%
	\end{rem}
	
	%Distinct from Remark \ref{rem:3d_bl}, for a fixed boundary localization level $\epsilon$, 
	
	In the following corollary, for a prescribed boundary localization level \(\epsilon\), we adjust the high-contrast parameter \(\delta\) (modifying the material properties of the high-contrast medium) so that the incident wave satisfies the condition \eqref{eq:n27}, where the index \(n\) of the incident wave \(u^i_n\), defined in \eqref{eq:3d_ui}, is chosen according to \(\delta\). Under this configuration, both the internal total field \(u|_{\Omega}\) and the external scattered field \(u^s|_{B_R \setminus \overline{D}}\) exhibit boundary localization.
	
	\begin{cor}\label{cor:3.2}
		Consider the acoustic scattering problem \eqref{eq:system} in $\mathbb{R}^3$, and  suppose that $B_R$ such that $\Omega \Subset B_R$ with $R\in (2,\infty)$. Under the assumptions \eqref{eq:sw}-\eqref{eq:wn1}, for fixed parameters $\gamma_1 \in (0,1)$ and $\gamma_2 \in (1,2)$, and a fixed boundary localization level $\varepsilon$ such that the density contrast parameter $\delta$ satisfies
		\begin{align}\label{eq:327}
			\delta \leqslant \beta =\min\left\{\frac{2\ln \gamma_1}{\ln\varepsilon - 3\ln \gamma_1}, \frac{2\ln \gamma_2}{\ln{\gamma_2} - \ln \varepsilon} \right\}.
		\end{align}
		By selecting an incident wave  defined in  \eqref{eq:3d_ui} with index $n$ satisfying 
		\begin{align}\label{eq:n27}
			n \geqslant \frac{1}{\delta},
		\end{align}
		it follows that the internal total field $\mathbf{u}|_D$ and the external scattered field $\mathbf{u}^s|_{\mathbb{R}^3 \setminus \overline{D}}$ exhibit boundary localization.
	\end{cor}

	\begin{proof}
		When the condition \eqref{eq:327} holds, it can be verified that 
		\begin{align}\label{eq:326}
			\frac{1}{\delta} \geqslant \max\{n_1, n_2\}. 
		\end{align}
		where $n_1$ and $n_2$ are as defined in \eqref{n_1} and \eqref{n2}, respectively. Combining with \eqref{eq:n27} and \eqref{eq:326}, we deduce that the index $n$ satisfies the condition $n \geq \max\left(n_1, n_2\right)$ which satisfies the requirement in Remark \ref{rem:3d_bl}, thereby completing the proof of Corollary \ref{cor:3.2}. 
	\end{proof}

	%In Theorem \ref{thm:gradient_u}, we will show the occurrence of surface resonances for the total internal field and the external scattered field by choosing an appropriate index $n$ defining the incident wave $u_n^i$ for a fixed high-contrast density parameter $\delta$. 
	
	In Theorem~\ref{thm:gradient_u}, we demonstrate the occurrence of surface resonances for both the internal total field and the external scattered field by appropriately choosing the index \(n\) that defines the incident wave $u_n^i$ given by \eqref{eq:3d_ui}, for a fixed high-contrast density parameter \(\delta\).

	%change drastically, compared with $\| u^{i}_{n}(\mathbf{x})\|^2_{L^2\left(\Omega\right)}$. These significant variations clearly demonstrate the occurrence of surface resonance.

	\begin{thm}\label{thm:gradient_u}
		
		Consider the acoustic scattering problem \eqref{eq:system} in \(\mathbb{R}^3\). Suppose that $B_R$ such that $\Omega \Subset B_R$ with $R\in (2,\infty)$. We recall that \(\mathcal{N}_{1,1-\gamma_1}\left(\partial\Omega\right)\) and \(\mathcal{N}_{2,\gamma_2-1}\left(\partial\Omega\right)\) are defined in \eqref{df:bregion}, where \(\gamma_1 \in (0,1)\) and \(\gamma_2 \in (1, 2)\) are constants. Under the assumptions \eqref{eq:sw}--\eqref{eq:wn1}, for a sufficiently small boundary localization level \(\epsilon\ll 1\), we choose the incident wave defined in \eqref{eq:3d_ui} with index \(n\) satisfying the condition
		\begin{equation}\label{eq:brn}
			n \geq \max\left(n_1, n_2, \frac{1}{\delta^2}\right),
		\end{equation}
		where \(n_1\) and \(n_2\) are defined in \eqref{n_1} and \eqref{n2}, respectively. With this choice of \(n\), we obtain the following estimates:
		\begin{equation}\label{gradientnorm}
			\frac{\|\nabla u_n(\mathbf{x})\|_{L^2\left(\mathcal{N}_{1,1-\gamma_1}\left(\partial\Omega\right)\right)}}{\| u^{i}_{n}(\mathbf{x})\|_{L^2\left(\Omega\right)}} \ge \frac{n \delta}{3} \gg 1, \quad 
			\frac{\|\nabla u^s_n(\mathbf{x})\|_{L^2\left(\mathcal{N}_{2,\gamma_{2}-1}\left(\partial\Omega\right)\right)}}{\| u^{i}_{n}(\mathbf{x})\|_{L^2\left(\Omega\right)}} \ge \frac{n}{3} \gg 1.
		\end{equation}
		These estimates indicate the occurrence of the surface resonance.

		%%%%%%%%%
		\iffalse
		Consider the acoustic scattering problem \eqref{eq:system} in $\mathbb{R}^3$. Let $B_R$ denote the region containing the high-contrast medium $\Omega$. We note that  $\mathcal{N}_{1,1-\gamma_1}\left(\partial\Omega\right)$ and $\mathcal{N}_{2,\gamma_2-1}\left(\partial\Omega\right)$ are defined in  \eqref{df:bregion}, where $\gamma_1\in(0,1)$ and $\gamma_2\in(1,R)$, respectively. Under the assumptions \eqref{eq:sw}-\eqref{eq:wn1}, for a sufficiently small level of boundary localization $\epsilon$ and we select the incident wave defined in $\eqref{eq:3d_ui}$ with $n$ satisfying the condition
		\begin{equation}\label{eq:brn}
			n \geq \max\left(n_1, n_2, \frac{1}{\delta^2}\right),
		\end{equation} 
		where $n_1$ and $n_2$ are defined in \eqref{n_1} and \eqref{n2}, respectively.
		With this choice of $n$, we have
		\begin{equation}\label{gradientnorm}
			\frac{\|\nabla u_n(\mathbf{x})\|_{L^2\left(\mathcal{N}_{1,1-\gamma_1}\left(\partial\Omega\right)\right)}}{\| u^{i}_{n}(\mathbf{x})\|_{L^2\left(\Omega\right)}}\ge \frac{ n \delta}{3}\gg 1, \quad \frac{\|\nabla u^s_n(\mathbf{x})\|_{L^2\left(\mathcal{N}_{2,\gamma_{2}-1}\left(\partial\Omega\right)\right)}}{\| u^{i}_{n}(\mathbf{x})\|_{L^2\left(\Omega\right)}}\ge \frac{n}{3}\gg 1.
		\end{equation}
		This indicates the occurrence of quasi-Minnaert resonance as defined in Definition \ref{df:quasi-minnaert}.
		\fi
		%%%%%%%%%%%%%%
		
	\end{thm}

	\begin{rem}\label{rem:32}
		%By noting the definition $n_1$ and $n_2$ given by \eqref{n_1} and \eqref{n2}, since $n$ satisfies \eqref{eq:brn}, we know that both $u_n|_{\Omega}$ and $u_n^s|_{B_R \setminus \Omega }$ fulfill the boundary localizations. In view of Theorem \ref{thm:gradient_u}, we can say that the quasi-Minnaert resonances occur in the sense of  Definition~\ref{df:quasi-minnaert}.
		By recalling the definitions of \( n_1 \) and \( n_2 \) in \eqref{n_1} and \eqref{n2}, and since \( n \) satisfies \eqref{eq:brn}, using Theorems \ref{3dthm1} and \ref{3dthm2},  it follows that both \( u_n|_{\Omega} \) and \( u_n^s|_{B_R \setminus \Omega} \) exhibit boundary localizations. Invoking Theorem~\ref{thm:gradient_u}, we conclude that quasi-Minnaert resonances arise in the sense of Definition~\ref{df:quasi-minnaert} when the high-contrast density ratio \(\delta\) is fixed and the incident wave \( u_n^i \) is chosen as in \eqref{eq:3d_ui}, with \( n \) satisfying \eqref{eq:brn}.
	\end{rem}
	
	\begin{proof}[The proof of Theorem \ref{thm:gradient_u}]
		
		We first prove the surface resonance of the internal total wave field $u_n$ associated with $u_n^i$ and the high-contrast medium $\Omega$ when $n$ satisfies \eqref{eq:brn}. Hence, we need to 
		%Our initial objective is to rigorously prove the inequality $\frac{\|\nabla u_n(\mathbf{x})\|^2_{L^2\left(\mathcal{N}_{1,1-\gamma_{1}}\left(\partial\Omega\right)\right)}}{\| u^{i}_{n}(\mathbf{x})\|^2_{L^2\left(\Omega\right)}}\gg 1$. First,
		derive the asymptotic analysis for $\nabla u_n$ in $\mathcal{N}_{1,1-\gamma_1}\left(\partial\Omega\right)$. Utilizing the gradient formula in spherical coordinates, we obtain
		\begin{equation}\label{gradient3d}
			\begin{aligned}
				&\nabla j_n(k|\mathbf{x}|)Y_n^m(\hat{\mathbf{x}})\\ 
				&=kj_n^{'}(k|\mathbf{x}|)Y_n^m(\hat{\mathbf{x}})\hat{e}_r+j_n(k|\mathbf{x}|)\nabla Y_n^m(\hat{x})\\
				&=kj_n^{'}(k|\mathbf{x}|)Y_n^m(\hat{\mathbf{x}})\hat{e}_r+\frac{j_n(k|\mathbf{x}|)}{|\mathbf{x}|}\nabla_\mathbb{S} Y_n^m(\hat{\mathbf{x}}).
			\end{aligned}
		\end{equation}
		where 
		\begin{equation}
			\hat{e}_r=\left(\begin{array}{c}\sin\theta\cos\varphi\\\sin\theta\sin\varphi\\\cos\theta\end{array}\right),\quad \hat{e}_{\theta}=\left(\begin{array}{c}\cos\theta\cos\varphi\\\cos\theta\sin\varphi\\-\sin\theta\end{array}\right),\quad \hat{e}_{\varphi}=\left(\begin{array}{c}-\sin\varphi\\\cos\varphi\\0\end{array}\right).
		\end{equation}
		and $\nabla_\mathbb{S}$ represents the surface gradient.
		Substituting  \eqref{gradient3d} into \eqref{internalus}, it is clear that
		\begin{align}\label{eq:total_gradient_u_1}
			\nabla	u(\mathbf{x})
			&=S_{1}(\mathbf{x})+S_{2}(\mathbf{x}),
		\end{align}
		where
		\begin{align*}
			S_{1}(\mathbf{x})&=\sum_{m=-n}^{n} \frac{-i\delta\beta_n^mk_b k^nj_n(k_b|\mathbf{x}|)h_n(k_b)}{|\mathbf{x}|(2n-1)!!}\left(-\frac{(2n+1)}{[\delta(n+1)+n]}+\mathcal{O}\left(\frac{ k^{2}}{n^2}\right)\right)\nabla_\mathbb{S} Y_n^m(\hat{\mathbf{x}}),\\
			S_{2}(\mathbf{\mathbf{x}})&=
			\sum_{m=-n}^{n} \frac{-i\delta\beta_n^m k_b^2 k^{n}j_n^{'}(k_b|\mathbf{x}|)h_n(k_b)}{(2n-1)!!}\left(-\frac{(2n+1)}{[\delta(n+1)+n]}+\mathcal{O}\left(\frac{ k^{2}}{n^2}\right)\right) Y_n^m(\hat{\mathbf{x}})\hat{e}_r.
		\end{align*}
		with $\delta$ are defined in \eqref{eq:delta}. By applying  the recurrence relation \eqref{eq:recursive_eq} of the spherical Bessel function and similar arguments as in \eqref{eq:3d_u_norm}, we can rewrite $S_{1}(\mathbf{x})$ and $S_{2}(\mathbf{x})$ as follows:
		\begin{align}
			S_{1}(\mathbf{x})
			&= \sum_{m=-n}^{n} \frac{\delta\beta_n^m k^n\vert \mathbf{x}\vert^{n-1}}{(2n-1)!!}\left(\frac{1}{\delta(n+1)+n}+\mathcal{O}\left(\frac{k^{2}}{n}\right)\right)\nabla_\mathbb{S} Y_n^m(\hat{\mathbf{x}}),\\
			S_{2}(\mathbf{x})
			&= \sum_{m=-n}^{n}\frac{n\delta\beta_n^m k^n\vert\mathbf{x}\vert^{n-1}}{(2n-1)!!} \left(-\frac{1}{\delta(n+1)+n}+\mathcal{O}\left(\frac{k^2}{n}\right)\right)Y_n^m(\hat{\mathbf{x}})\hat{e}_r.
		\end{align}
		Subsequently, we proceed to derive that the $L^2$ norm of $\nabla u(\mathbf{x})$ in $\mathcal{N}_{1,1-\gamma_1}\left(\partial\Omega\right)$.  According to the
		definition \eqref{eq:total_gradient_u_1} of $\nabla u(\mathbf{x})$, and taking into account the orthogonality between $\hat{e}_r$ and $\nabla_\mathbb{S} Y_n^m(\hat{x})$, we arrive at the following conclusion:
		\begin{align}\label{eq:3d_gun}
			\|\nabla u(\mathbf{x})\|^2_{L^2\left(\mathcal{N}_{1,1-\gamma_{1}}\left(\partial\Omega\right)\right)}
			&=\int_{\mathcal{N}_{1,1-\gamma_{1}}\left(\partial\Omega\right)}|S_{1}(\mathbf{x})|^2d\mathbf{x}+
			\int_{\mathcal{N}_{1,1-\gamma_{1}}\left(\partial\Omega\right)}|S_{2}(\mathbf{x})|^2d\mathbf{x}.
		\end{align} 
		By using the equality 
		\begin{align*}
			\int_\mathbb{S} \nabla_\mathbb{S} Y_n^m(\hat{\mathbf{x}}) \nabla_\mathbb{S} Y_{n^{\prime}}^{m^{\prime}}(\hat{\mathbf{x}})d\mathbf{x}=n(n+1)\int_\mathbb{S} Y_n^m(\hat{\mathbf{x}}) Y_{n^{\prime}}^{m^{\prime}}(\hat{\mathbf{x}})d\mathbf{x},
		\end{align*}
		one further obtain
		\begin{align}
			&\int_{\mathcal{N}_{1,1-\gamma_{1}}\left(\partial\Omega\right)}|S_{1}(\mathbf{x})|^2d\mathbf{x} \notag\\
			&= 
			\int_{\mathcal{N}_{1,1-\gamma_{1}}\left(\partial\Omega\right)} \sum_{m=-n}^{n}\frac{\delta^2\vert\beta_n^m \vert^2k^{2n}\vert\mathbf{x}\vert^{2n-2}}{[(2n-1)!!]^2}\left(\frac{1}{\left(\delta(n+1)+n\right)^2}+\mathcal{O}\left(\frac{k^2}{n^2}\right)\right)\vert\nabla_\mathbb{S} Y_n^m(\hat{\mathbf{x}})\vert^2d\mathbf{x}\notag\\
			&=\int_{0}^{2\pi}\int_{0}^{\pi}\int_{\gamma_1}^{1}  \sum_{m=-n}^{n}\frac{\delta^2\vert\beta_n^m \vert^2k^{2n}r^{2n}}{[(2n-1)!!]^2}\left(\frac{1}{\left(\delta(n+1)+n\right)^2}+\mathcal{O}\left(\frac{k^2}{n^2}\right)\right)\notag\\
			&\times\vert\nabla_\mathbb{S} Y_n^m(\hat{\mathbf{x}})\vert^2r^2\sin\theta drd\theta d\varphi\notag\\
			&=n(n+1)\sum_{m=-n}^{n}\frac{\delta^2\vert\beta_n^m \vert^2k^{2n}(1-\gamma^{2n+1}_{1})}{[(2n-1)!!]^2(2n+1)}\left(\frac{1}{\left(\delta(n+1)+n\right)^2}+\mathcal{O}\left(\frac{k^2}{n^2}\right)\right)\label{eq:3d_s_norm},
		\end{align}
		and 
		\begin{align} 
			&\int_{\mathcal{N}_{1,1-\gamma_{1}}\left(\partial\Omega\right)}|S_{2}(\mathbf{x})|^2d\mathbf{x}\notag \\
			&= 
			\int_{\mathcal{N}_{1,1-\gamma_{1}}\left(\partial\Omega\right)} \sum_{m=-n}^{n} \frac{n^2\delta^2\vert\beta_n^m\vert^2 k^{2n}|\mathbf{x}|^{2n-2}}{\left[(2n-1)!!\right]^2}\left(\frac{1}{\left(\delta(n+1)+n\right)^2}+\mathcal{O}\left(\frac{k^2}{n^2}\right)\right)\vert Y_n^m(\hat{\mathbf{x}})\vert^2d\mathbf{x}\notag\\
			&=\int_{0}^{2\pi}\int_{0}^{\pi}\int_{\gamma_1}^{1}  \sum_{m=-n}^{n}\frac{n^2\delta^2\vert\beta_n^m\vert^2 k^{2n}r^{2n-2}}{\left[(2n-1)!!\right]^2}\left(\frac{1}{\left(\delta(n+1)+n\right)^2}+\mathcal{O}\left(\frac{k^2}{n^2}\right)\right)\notag\\
			&\times\vert Y_n^m(\hat{\mathbf{x}})\vert^2r^2\sin\theta drd\theta d\varphi\notag\\
			&=\sum_{m=-n}^{n}
			\frac{n^2\delta^2\vert\beta_n^m\vert^2 k^{2n}(1-\gamma^{2n+1}_{1})}{\left[(2n-1)!!\right]^2(2n+1)}
			\left(\frac{1}{\left(\delta(n+1)+n\right)^2}+\mathcal{O}\left(\frac{k^2}{n^2}\right)\right)\label{eq:3d_s1_norm}.
		\end{align}
		Substituting \eqref{eq:3d_s1_norm} and \eqref{eq:3d_s_norm} into \eqref{eq:3d_gun}, we can derive the following asymptotic analysis with respect to $k$:
		\begin{align}\label{eq:tgu2}
			\|\nabla u_n(\mathbf{x})\|^2_{L^2\left(\mathcal{N}_{1,1-\gamma_{1}}\left(\partial\Omega\right)\right)}
			&=\sum_{m=-n}^{n}
			\frac{n\delta^2\vert\beta_n^m\vert^2 k^{2n}(1-\gamma^{2n+1}_{1})}{\left[(2n-1)!!\right]^2\left(\delta(n+1)+n\right)^2}
			\left(1+\mathcal{O}\left(k^2\right)\right).
		\end{align} 
		In accordance with the definition of  $u_n^i$ and the asymptotic expansion of the spherical Bessel function as specified in \eqref{eq:j_n_eps}, we obtain
		\begin{align}\label{eq:ui2}
			\| u^{i}_{n}(\mathbf{x})\|^2_{L^2\left(\Omega\right)}
			&=\sum_{m=-n}^{n}
			\frac{\vert\beta_n^m\vert^2 k^{2n}}{\left[(2n+1)!!\right]^2(2n+3)}
			\left(1+\mathcal{O}\left(k^2\right)\right).
		\end{align} 
		In view of \eqref{eq:tgu2} and \eqref{eq:ui2}, we subsequently derive the following result
		\begin{align*}
			\frac{\|\nabla u_n(\mathbf{x})\|^2_{L^2\left(\mathcal{N}_{1,1-\gamma_1}\left(\partial\Omega\right)\right)}}{\| u^{i}_{n}(\mathbf{x})\|^2_{L^2\left(\Omega\right)}}
			&=\frac{\delta^2n(2n+1)^2(2n+3)(1-\gamma^{2n+1}_{1})}{\left(\delta(n+1)+n\right)^2}\left(1+\mathcal{O}\left(k^{2}\right)\right)\left(1-\mathcal{O}\left(k^{2}\right)\right)\\
			&=\frac{n^2\delta^2(2+\frac{1}{n})^2(2+\frac{3}{n})(1-\gamma^{2n+1}_{1})}{\left(\delta(1+\frac{1}{n})+1\right)^2}\left(1+\mathcal{O}\left(k^{2}\right)\right)\left(1-\mathcal{O}\left(k^{2}\right)\right)\\
		\end{align*}
		Under the conditions $n \ge n_1$ and $\delta \ll 1$, since $\epsilon$ is sufficient small,  the following inequality holds
		\begin{equation}
			\frac{(2+\frac{1}{n})^2(2+\frac{3}{n})(1-\gamma^{2n+1}_{1})}{\left(\delta(1+\frac{1}{n})+1\right)^2} \ge \frac{2}{9} %\frac{4(1-\epsilon)}{9} .
		\end{equation}
		%For sufficiently small  $\epsilon$, we further simplify to $\frac{(2+\frac{1}{n})^2(2+\frac{3}{n})(1-\gamma^{2n+1}_{1})}{\left(\delta(1+\frac{1}{n})+1\right)^2} \ge \frac{2}{9}$. Combining these results, we establish the estimates:
		Therefore, one has
		\begin{align}
			\frac{\|\nabla u_n(\mathbf{x})\|^2_{L^2\left(\mathcal{N}_{1,1-\gamma_1}\left(\partial\Omega\right)\right)}}{\| u^{i}_{n}(\mathbf{x})\|^2_{L^2\left(\Omega\right)}}
			&\ge\frac{2}{9}n^2\delta^2\left(1+\mathcal{O}\left(k^{2}\right)\right)\left(1-\mathcal{O}\left(k^{2}\right)\right)\notag\\
			&\ge \frac{n^2 \delta^2}{9}, \label{gunorm}
		\end{align}
		which implies 
		$$
		\frac{\|\nabla u_n(\mathbf{x})\|_{L^2\left(\mathcal{N}_{1,1-\gamma_1}\left(\partial\Omega\right)\right)}}{\| u^{i}_{n}(\mathbf{x})\|_{L^2\left(\Omega\right)}} \ge \frac{n \delta}{3} \gg 1. 
		$$

		%the energy of the total field gradient exhibits a more pronounced variation compared to that of the incident wave.

		Next, We proceed to show the occurrence of surface resonances for the external scattered wave $u_n^s$ associated with $u_n^i$ and the high-contrast medium $\Omega$ when the index $n$ fulfills \eqref{eq:brn}. Similarly, we first derive the asymptotic expansion of $\nabla	u^s_n(\mathbf{x})$  with respect to $k$ as follows:
		%derive the ratio $\frac{\|\nabla u^s_n(\mathbf{x})\|^2_{L^2\left(\mathcal{N}_{2,\gamma_{2}-1}\left(\partial\Omega\right)\right)}}{\| u^{i}_{n}(\mathbf{x})\|^2_{L^2\left(\Omega\right)}}\gg 1$. From \eqref{eq:ui2}, the explicit expression for $\| u^{i}_{n}(\mathbf{x})\|^2_{L^2\left(\Omega\right)}$ is already established. To compute the ratio, we focus on deriving $\|\nabla u^s_n(\mathbf{x})\|^2_{L^2\left(\mathcal{N}_{2,\gamma_{2}-1}\left(\partial\Omega\right)\right)}$. Substituting  \eqref{gradient3d} into \eqref{eq:us} and arguments analogous in \eqref{eq:total_gradient_u_1}, the gradient of the scattered field in $\mathcal{N}_{2,\gamma_{2}-1}\left(\partial\Omega\right)$ is expressed as
		\begin{align}
			\nabla	u^s_n(\mathbf{x})&=\sum_{m=-n}^{n}\frac{\beta_n^m k^n}{(2n+1)!!|\mathbf{x}|^{n+2}} \left(-\frac{n(1-\delta)}{\delta(n+1)+n}+\mathcal{O}\left(k^2\right)\right)\nabla_\mathbb{S} Y_n^m(\hat{\mathbf{x}})\notag\\
			&+\frac{n\beta_n^m k^{n}}{(2n+1)!!|\mathbf{x}|^{n+2}} \left(\frac{(n+1)(1-\delta)}{[\delta(n+1)+n]}+\mathcal{O}\left(k^2\right)\right)Y_n^m(\hat{\mathbf{x}})\hat{e_r}.\notag
		\end{align}
		Using the similar argument for deriving  \eqref{eq:tgu2}, it can be deduced that
		\begin{align}\label{eq:total_gradient_u_2}
			\|\nabla u^s_n(\mathbf{x})\|^2_{L^2\left(\mathcal{N}_{2,\gamma_{2}-1}\left(\partial\Omega\right)\right)}
			&=\sum_{m=-n}^{n}
			\frac{n\vert\beta_n^m\vert^2 k^{2n}\left(1-\frac{1}{\gamma^{2n+1}_{2}}\right)}{\left[(2n-1)!!\right]^2\left(\delta(n+1)+n\right)^2}
			\left(1+\mathcal{O}\left(k^2\right)\right)
		\end{align}  
		Combining \eqref{eq:total_gradient_u_2} and \eqref{eq:ui2}, it is readily to know that
		\begin{align}
			\frac{\|\nabla u^s_n(\mathbf{x})\|^2_{L^2\left(\mathcal{N}_{2,\gamma_{2}-1}\left(\partial\Omega\right)\right)}}{\| u^{i}_{n}(\mathbf{x})\|^2_{L^2\left(\Omega\right)}}
			&=\frac{n(2n+1)^2(2n+3)\left(1-\frac{1}{\gamma^{2n+1}_{2}}\right)}{\left(\delta(n+1)+n\right)^2}\left(1+\mathcal{O}\left(k^{2}\right)\right)\left(1-\mathcal{O}\left(k^{2}\right)\right)\notag\\
			&=\frac{n^2(2+\frac{1}{n})^2(2+\frac{3}{n})\left(1-\frac{1}{\gamma^{2n+1}_{2}}\right)}{\left(\delta(1+\frac{1}{n})+1\right)^2}\left(1+\mathcal{O}\left(k^{2}\right)\right)\left(1-\mathcal{O}\left(k^{2}\right)\right)\notag\\
			&\ge\frac{2}{9}n^2\left(1+\mathcal{O}\left(k^{2}\right)\right)\left(1-\mathcal{O}\left(k^{2}\right)\right)\notag\\
			&\ge\frac{n^2}{9} \label{gusnorm}
		\end{align}
		where the first inequality can be proved in a similar manner for \eqref{gunorm} by noting $n \geq n_2$.

		%is based on $\frac{(2+\frac{1}{n})^2(2+\frac{3}{n})(1-\gamma^{2n+1}_{1})}{\left(\delta(1+\frac{1}{n})+1\right)^2} \ge \frac{2}{9}$ with $\delta,\epsilon\ll 1$ and $n \geq \lceil-\left(\frac{1}{2} \frac{\ln{\epsilon}}{\ln{\gamma_2}}+1 \right)\rceil+1$.
		
		The proof is complete.
	\end{proof}

	In Theorem \ref{thm:gradient_u}, we proved that by appropriately selecting the index $n$ of the incident field $u_n^i$, both the interior wave field $u_n|_{\Omega}$ and the exterior scattered field $u^s_n|_{B_R \setminus \Omega }$ can simultaneously exhibit boundary localization and surface resonance near the boundary of the high-contrast medium $\Omega$. In the following corollary, we will demonstrate that when the index $n$ of the incident wave field depends only on the high-contrast density parameter $\delta$, the resulting interior total field and exterior scattered field will only generate surface resonance, without guaranteeing the occurrence of boundary localization.

	%In Corollary \ref{cor:gbr}, we illustrates that even if $n$ does not satisfy the inequality \eqref{eq:brn}, we can still find a special incident wave $u^i_n$ such that surface resonance occurs.

	\begin{cor}\label{cor:gbr}
		
		Under the same assumptions as Theorem \ref{thm:gradient_u}, if we select the incident wave $u_n^i$ defined in $\eqref{eq:3d_ui}$ with $n$ only satisfying the condition
		\begin{align}\label{eq:n delta con}
			n \geqslant \frac{1}{\delta^2},
		\end{align}
		then the internal total field \( u_n|_{\Omega} \) and the external scattered field \( u^s_n|_{\mathbb{R}^d \setminus \overline{D}} \)  satisfy \eqref{gradientnorm}.
	\end{cor}
	\begin{proof}
		Employing an approach analogous to the proof of Theorem \ref{thm:gradient_u}, we derive the inequalities \eqref{gunorm} and \eqref{gusnorm}. In view of condition \eqref{eq:n delta con}, the following results hold
		\begin{equation*}
			\frac{\|\nabla u_n(\mathbf{x})\|_{L^2\left(\mathcal{N}_{1,1-\gamma_1}\left(\partial\Omega\right)\right)}}{\| u^{i}_{n}(\mathbf{x})\|_{L^2\left(\Omega\right)}}\ge \frac{1}{3\delta}\gg 1, \quad \frac{\|\nabla u^s_n(\mathbf{x})\|_{L^2\left(\mathcal{N}_{2,\gamma_2-1}\left(\partial\Omega\right)\right)}}{\| u^{i}_{n}(\mathbf{x})\|_{L^2\left(\Omega\right)}}\ge \frac{1}{3\delta^2}\gg 1
		\end{equation*}
		since $\delta \ll 1$. 
		
		The proof is complete.
	\end{proof}
	
	\iffalse

	\begin{rem}\label{rem:4.1}
		Considering that \(\epsilon \ll 1\) characterizes the level of boundary localization and the index $n$ of the incident wave satisfies \eqref{eq:brn}, it can be deduced that
		\begin{align}\label{eq:n44}
			n \geqslant \max \{n_1,n_2\} \gg 1,
		\end{align}
		where \(n_1\) and \(n_2\) depend on the choice of $\epsilon$. This inequality implies the occurrence of boundary localization. When $n$ only satisfies \(n \ge \frac{1}{\delta^2}\), we can demonstrate that surface resonance occurs. Nevertheless, this condition is insufficient to demonstrate the emergence of boundary localization for the internal total field $u_n$ and the external scattered field \(u^{s}_n\). 
		
	\end{rem}
	
	\fi

	In Theorem \ref{thm:gradient_u}, from Remark \ref{rem:32}, when the high-contrast density parameter $\delta$ is fixed, by modulating the incident wave $u_n^i$ given by \eqref{eq:3d_ui}, the corresponding high-contrast medium characterized by \((\Omega;\rho_b, \kappa_b)\)  is made to be a quasi-Minnaert resonator. Similarly, we can adjust the high-contrast density parameter $\delta$ and chose the corresponding the index $n$ defining $u_n^i$ to achieve quasi-Minnaert resonances. 
	In fact, according to Corollary \ref{cor:3.2} and Corollary \ref{cor:gbr}, we can adjust the high-contrast density parameter \(\delta\) to induce the quasi-Minnaert resonance. We summarize this as the following theorem.

	\begin{prop}\label{pro:4.2}  
		Under the same assumptions as Theorem \ref{thm:gradient_u}, for the incident wave defined in \eqref{eq:3d_ui} interacting with the scattering system \eqref{eq:system}, when the high contrast density parameter \(\delta\) satisfies \eqref{eq:327} and \eqref{eq:n delta con}, the high-contrast medium $\Omega$ behaves as a quasi-Minnaert resonator.
		%			Under the same assumptions as Theorem \ref{thm:gradient_u}, when the incident wave defined in \eqref{eq:3d_ui} interacts with the scattering system \eqref{eq:system}, causing the acoustic medium to exhibit the characteristics of a quasi-Minnaert resonator. then the density parameter \(\delta\), which has a significantly high contrast relative to the background medium, must satisfy the relationships \eqref{eq:327} and \eqref{eq:n delta con}.
	\end{prop}

	%It is noted that \( n_1 \) and \( n_2 \), given by \eqref{eq:n1 n2 def}, play an important role in proving the \(\mathcal{O}(\varepsilon)\)-level boundary localization of \(\mathbf{u}|_D\) and \(\mathbf{u}^s|_{\mathbb{R}^3 \setminus D}\) in the sense of Definition \ref{def:surface localized}. 

	%Consequently, both boundary localization and surface resonance occur for \(\mathbf{u}|_D\) and \(\mathbf{u}^s|_{\mathbb{R}^3 \setminus D}\) in the sense of Definitions \ref{def:surface localized} and \ref{def:surface resonant}. This implies that the quasi-Minnaert resonance is associated with \(\mathbf{u}^i\) given by \eqref{eq:u^i 4.1} in the sub-wavelength regime. 

	%\begin{rem}\label{rem:3d_br}
	%	According to Theorem\ref{3dthm1}, \ref{3dthm2} and \ref{thm:gradient_u}, under the assumptions \eqref{eq:sw}-\eqref{eq:wn1}, we can select a high-mode incident wave, enabling the high contrast bubble $\Omega$ to become a quasi-Minnaert resonator. However, we need to emphasize that here $\Omega$ is a unit sphere. In fact, we believe this also holds for other shapes. We hold the belief that the findings of this theorem can be applied to high-contrast inclusions with more general shapes. Nevertheless, we development work to a future study.
	%\end{rem}
	
	\section{Quasi-Minnaert resonances in $\mathbb R^2$}\label{sec4}

	In this section, we establish the main results in $\mathbb{R}^2$. The proofs follow a similar approach to those presented for the three-dimensional case in the preceding section. To avoid redundancy, we provide only the essential ingredients required to complete the proofs. 
	In what follows, we denote  $J_{n}(t)$ and $H_{n}^{(1)}(t)$, $n\in \mathbb{N}$, as the Bessel function of order $n$ and the Hankel function of the first kind of order $n$, respectively. For a fixed $n$ with $0<|t|\ll1$, the following asymptotic expansions for $J_n(t)$ and $H_n^{(1)}(t)$ holds:
	%Based on the findings in \cite{XYC}, we can deduce that for a fixed $n$ with $0<|t|\ll1$, it holds that
	\begin{equation}\label{eq:jn} 
		J_n(t)= \frac{1}{\Gamma(n+1)}\Big(\frac{t}{2}\Big)^n\left(1-\frac{1}{n+1}\Big(\frac{t}{2}\Big)^2\\ +\mathcal{O}\left(\frac{1}{2(n+1)(n+2)}\Big(\frac{t}{2}\Big)^4\right)\right),\quad\text{for} \quad n \ge 1
	\end{equation}
	and
	\begin{equation}\label{eq:hn} 
		H_n^{(1)}(t)=-i\frac{\Gamma(n)}{\pi}\Big(\frac{2}{t}\Big)^n\left(1+\frac{1}{n-1}\Big(\frac{t}{2}\Big)^2 +
		\mathcal{O}\left(\frac{1}{(n-1)(n-2)}\Big(\frac{t}{2}\Big)^3\right)\right),\quad \text{for}\quad n \ge 2.
	\end{equation}
	where $\Gamma$ is the Gamma function and $\gamma=0.5772...$ is the Euler-Mascheroni constant.

	To facilitate exposition, we assume that the acoustic high-contrast medium $\Omega$ exhibits radial geometry. After a suitable coordinate transformation, we may take $\Omega$ to be the unit disk in $\mathbb{R}^2$. The spectral properties of layer potential operators in $\mathbb{R}^2$ are detailed in \cite{ammari} and are summarized in the following lemma. Let $(r,\theta)$ be the polar coordinate of $\mathbf x=(x_1,x_2) \in \mathbb R^2$.

	%Similar to the three dimensional case, the reference \cite{ammari} gives a spectral theory for the layer potential operator in the two-dimensional case with the following lemma.
	\begin{lem}\label{lem:2d_s_k}
		For the single-layer potential operator \(\mathcal{S}_{\partial \Omega}^{k}\) defined on \eqref{eq:S_k} and the N-P operator \((K_{\partial \Omega}^k)^*\) defined on \eqref{K_k*}, it holds that 
		
		%The spectral systems of the operator \(\mathcal{S}_{\partial B_1}^{k}\) defined on \eqref{eq:S_k} and \((K_{\partial B_1}^k)^*\) defined on \eqref{K_k*} are shown as follows
		\begin{equation}
			\mathcal{S}_{\partial \Omega}^{k}[e^{in\theta}](\mathbf{x})=-\frac{i\pi}{2}J_{n}(k)H_{n}^{(1)}(k\vert \mathbf{x}\vert)e^{in\theta},\quad \vert \mathbf{x}\vert>1
		\end{equation}
		and
		\begin{equation}
			(K_{\partial \Omega}^k)^*[e^{in\theta}]=\Big(-\frac{1}{2}-\frac{i\pi}{2}kJ_{n}(k)H_{n}^{(1)^{\prime}}(k)\Big)e^{in\theta}=\Big(\frac{1}{2}-\frac{i\pi}{2}kJ'_{n}(k)H_{n}^{(1)}(k)\Big)e^{in\theta}.
		\end{equation}
	\end{lem}
	\begin{rem}
		Analogous  to the proof of Lemma $\ref{lem:3d_s}$, we have the following two identities hold
		\begin{equation}
			\mathcal{S}_{\partial \Omega}^{k}[e^{in\theta}](\mathbf{x})=-\frac{i\pi}{2}J_{n}(k)H_{n}^{(1)}(k)e^{in\theta_{x}},\quad \vert \mathbf{x}\vert=1
		\end{equation}
		and 
		\begin{equation}
			\mathcal{S}_{\partial B_1}^{k}[e^{in\theta}](\mathbf{x})=-\frac{i\pi}{2}J_{n}(k\vert \mathbf{x}\vert)H_{n}^{(1)}(k)e^{in\theta_{x}},\quad \vert \mathbf{x}\vert<1.
		\end{equation}
	\end{rem}

	In order to generate the boundary localization and surface resonance for the internal total wave field and external scattered wave in $\mathbb R^2$, we choose the incident wave \(u^i_n\) as
	\begin{equation}\label{eq:2d_ui}
		{u}^i_n(\mathbf{x})=\beta_{n}J_{n}(k|\mathbf{x}|) {e}^{{i}n\theta}, 
	\end{equation}
	where $\beta_{n}$ is non-zero constant.

	%Subsequently, we utilize the spectral theory of 2D layer potentials to prove Theorem \ref{2dthm1} and illustrate that the incident wave $u^i$ defined in \eqref{eq:2d_ui} can make both the internal total field and the external scattered field locally localized on the boundary simultaneously.

	\begin{thm}\label{2dthm1}
		Consider the acoustic scattering problem \eqref{eq:system} in $\mathbb{R}^2$. Let $B_R$ be a disk centered at the origin with radius $R$ satisfying \(\Omega \Subset B_R\), where  \(R \in (2, \infty)\). Suppose that  $\mathcal{N}_{1,1-\gamma_1}\left(\partial\Omega\right)$ and $\mathcal{N}_{2,\gamma_2-1}\left(\partial\Omega\right)$ are defined in  \eqref{df:bregion}, where $\gamma_1\in(0,1)$ and $\gamma_2\in(1,2)$ are constants. Under the assumptions \eqref{eq:sw}-\eqref{eq:wn1}, for a fixed sufficiently small $\epsilon > 0$, if the  incident wave defined in \eqref{eq:2d_ui} with index $n$ satisfying the condition
		\begin{equation}\label{eq:thm 41 n1n2}
			n \geq \max\left(\widetilde{n}_{1},\widetilde{n}_2\right)
		\end{equation} 
		with 
		\begin{align}\label{2dn}
			\widetilde{n}_1=\lceil\frac{1}{2}\frac{\ln{\epsilon}}{\ln{\gamma_1}} - 3 \rceil + 2, \quad \widetilde{n}_2=-\lceil\frac{1}{2}\frac{\ln{\epsilon}}{\ln{\gamma_2}} + 3 \rceil + 2, 
		\end{align}
		then the following estimates hold
		\begin{align}\label{eq:thm 41 ratio}
			r_{u_n,\gamma_1 }:=\frac{\|u_n\|^2_{L^2\left(\Omega \setminus \mathcal{N}_{1,1-\gamma_{1}}\left(\partial\Omega\right)\right)}}{\|u_n\|^2_{L^2\left(\Omega\right)}} &\leq  \mathcal{O}\left(\epsilon\right)+\mathcal{O}\left(\epsilon\omega^2\right)\ll 1,\\
			r_{u_n^s,\gamma_2 }:=\frac{\|u^s_n\|^2_{L^2\left(B_{R}\setminus\left(\mathcal{N}_{2,\gamma_{2}-1}\left(\partial\Omega\right)\cup \Omega\right)\right)}}{\|u^s_n\|^2_{L^2\left(B_{R}\setminus \Omega\right)}}&\leq \mathcal{O}\left(\epsilon\right)+\mathcal{O}\left(\epsilon\omega^2\right)\ll 1. \notag
		\end{align}
	\end{thm}

	\begin{proof}
		The proof of this theorem is similar to the counterpart of Theorem \ref{3dthm1} and \ref{3dthm2}. In what follows, we mainly
		give several major ingredients of the proof. In view of \(u^i_n\) given by \eqref{eq:2d_ui}, for the internal total wave field $u_n $ and external scattered wave field $u_n^s $ given by \eqref{eq:layer}, the associated  density functions $\varphi_b(\hat{\mathbf{x}}) $ and $\varphi(\hat{\mathbf{x}})$, corresponding to $u_n $ and $u_n^s $  respectively,  can be written as
		\begin{equation}\label{eq:2d_phi}
			\varphi_b(\hat{\mathbf{x}}) = \varphi_{b, n} \mathrm{e}^{\mathrm{i}n\theta},\quad
			\varphi(\hat{\mathbf{x}}) = \varphi_{n} \mathrm{e}^{\mathrm{i}n\theta}.
		\end{equation}
		where $\varphi_{b, n}$ and $\varphi_{n}$ are constants to be determined.
		By virtue of Lemma \ref{lem:2d_s_k}, substituting \eqref{eq:2d_phi} into \eqref{eq:5}, using \eqref{eq:jn} and \eqref{eq:hn}, for $n>2$ we obtain the asymptotic expansions of $\varphi_{b, n}$ and $\varphi_{n}$:
		\begin{equation}\label{eq:2d_density}
			\begin{pmatrix}
				\varphi_{b,n}\\
				\varphi_{n}
			\end{pmatrix}=\begin{pmatrix}
				-\frac{\delta\beta_nk^{n}}{(\delta+1)2^{n-2}(n-1)!}\\
				\frac{(1-\delta)\beta_{n}k^n}{(\delta+1)2^{n-1}(n-1)!}
			\end{pmatrix}+                                     
			\begin{pmatrix}
				\mathcal{O}\left(\frac{\delta\beta_n k^{n+2}}{2^{n}(n+1)!}\right)\\
				\mathcal{O}\left(\frac{\delta\beta_n k^{n+2}}{2^{n}(n+1)!}\right)
			\end{pmatrix}.
		\end{equation}
		In the following, we only prove that $u_n$ is internally boundary-localized, and
		$u^s_n$ can be proved by following a similar argument.
		For any $t\in[0,1]$, similar to \eqref{eq:3du_internal}, it yields that
		
		%by using the layer potential operators of the total field defined by \eqref{eq:layer} in $\Omega$, and combining with \eqref{eq:2d_phi}, \eqref{eq:2d_density}, \eqref{eq:jn} and \eqref{eq:hn}, we can calculate that
		
		\begin{align}
			\|u_n\|^2_{L^2\left(B_t\right)}=M_{1,n}(t)\left(1+\mathcal{O}\left(k^2\right)\right),
		\end{align}
		where
		\begin{equation}\label{eq:2d1M}
			\begin{aligned}
				M_{1,n}(t) &=  \frac{\pi\delta^2 \vert\beta_n\vert^2k^{2n}t^{2n+2}}{\left[(\delta+1)2^{n-1}n!\right]^2\left(n+1\right)}.
			\end{aligned}
		\end{equation}
		For any $\gamma_{1} \in(0,1)$, due to $n>\widetilde{n}_{1}$ and $\frac{M_{1,n}(\gamma_{1})}{M_{1,n}(1)}=\gamma_{1}^{2n+2}$, we know that $\frac{M_{1,n}(\gamma_{1})}{M_{1,n}(1)} \leq \epsilon$. Therefore, one has
		\begin{align*}
			\frac{\|u_n\|^2_{L^2\left(\Omega \setminus \mathcal{N}_{1,1-\gamma_{1}}\left(\partial\Omega\right)\right)}}{\|u\|^2_{L^2\left(\Omega\right)}}
			&=\frac{M_{1,n}(\gamma_{1})}{M_{1,n}(1)}\left(1+\mathcal{O}\left(k^{2}\right))(1-\mathcal{O}\left(k^{2}\right)\right)\\
			&\leq  \mathcal{O}\left(\epsilon\right)+\mathcal{O}\left(\epsilon\omega^2\right).
		\end{align*}
		
		The proof is complete.
		%		\begin{equation}
			%			\begin{aligned}
				%				\|u^s\|^2_{L^2\left(B_{t_1}\setminus B_{t_2}\right)} 
				%				&=N_{1,n}(t_1, t_2)+N_{2,n}(t_1, t_2).\\
				%			\end{aligned}
			%		\end{equation}
		%		where
		%		\begin{equation}\label{eq:2dN}
			%			\begin{aligned}
				%				N_{1,n}(t_1, t_2) &=  \frac{(1-\delta)^2\vert\beta_{n} \vert^2 k^{2n}}{\left[(\delta+1)2^{n}n!\right]^2}\left(\frac{1}{\gamma_2^{2n-3}}-\frac{1}{\gamma_1^{2n-3}}\right),\\
				%			\end{aligned}
			%		\end{equation}
		%		By straightforward calculations, one can verify that
		%		\begin{equation}
			%			\begin{aligned}
				%				\frac{\|u^s\|^2_{L^2\left(B_{\gamma_1}\setminus B_{\gamma_2}\right)}}{\|u^s\|^2_{L^2\left(B_{\gamma_1}\setminus \Omega\right)}}
				%				&=\frac{\frac{N_{1,n}(\gamma_1, \gamma_2)}{N_{1,n}(\gamma_1,1)}+\mathcal{O}\left(\frac{N_{2,n}(\gamma_1, \gamma_2)}{N_{1,n}(\gamma_1,1)}\right)}{1+\mathcal{O}\left(\frac{N_{2,n}(\gamma_1,1)}{N_{1,n}(\gamma_1,1)}\right)}\\
				%				&\leq\frac{\left(1+\mathcal{O}\left(\omega^{2}\right)\right)}{\gamma_2^{2n-3}\left(1+\mathcal{O}\left(\omega^{2}\right)\right)}\\
				%				&= \mathcal{O}\left(\epsilon\right)+\mathcal{O}\left(\epsilon\omega^2\right).
				%			\end{aligned}
			%		\end{equation}
	\end{proof}

	%\begin{cor}\label{cor:4.2}
	%	Consider the acoustic scattering problem \eqref{eq:system} in $\mathbb{R}^2$, and suppose that $B_R$ represents the region containing the acoustic wave medium $\Omega$. Under the assumptions \eqref{eq:sw}-\eqref{eq:wn1}, for fixed parameters $\gamma_1 \in (0,1)$ and $\gamma_2 \in (1,R)$, and a fixed boundary localization level $\varepsilon$ such that the density contrast parameter $\delta$ satisfies
	%	\begin{align}\label{eq:327}
		%		\delta \leqslant \beta =\min\left\{\frac{2\ln \gamma_1}{\ln\varepsilon - 3\ln \gamma_1}, \frac{2\ln \gamma_2}{3\ln{\gamma_2} - \ln \varepsilon} \right\}.
		%	\end{align}
	%	By selecting an incident wave  defined in  \eqref{eq:2d_ui} with index $n$ satisfying 
	%	\begin{align}\label{eq:n27}
		%		n \geqslant \frac{1}{\delta},
		%	\end{align}
	%	it follows that the internal total field ${u}|_D$ and the external scattered field ${u}^s|_{\mathbb{R}^3 \setminus \overline{D}}$ exhibit boundary localization.
	%\end{cor}
	%
	%
	%\begin{proof}
	% Analogous to the reasoning in Corollary \ref {cor:3.2}.
	%%\end{proof}

	Analogous to Theorem \ref{thm:gradient_u}, we establish that both the internal total wave field and the external scattered field generated by the incident wave defined by \eqref{eq:2d_ui}, exhibit surface resonances.

	\begin{thm}\label{thm:2d_gradient}
		Consider the acoustic scattering problem \eqref{eq:system} in $\mathbb{R}^3$. Let $B_R$ be a region satisfying \(\Omega \Subset B_R\) with \(R \in (2, \infty)\). Assume that  $\mathcal{N}_{1,1-\gamma_1}\left(\partial\Omega\right)$ and $\mathcal{N}_{2,\gamma_2-1}\left(\partial\Omega\right)$ are defined in  \eqref{df:bregion}, where $\gamma_1\in(0,1)$ and $\gamma_2\in(1,2)$ are constants. Under the assumptions \eqref{eq:sw}-\eqref{eq:wn1}, for a given boundary localization level $\epsilon \ll 1$, we select the incident wave defined in $\eqref{eq:3d_ui}$ with index $n$ satisfying the condition:
		\begin{equation}\label{eq:n thm 42}
			n \geq \max\left(\tilde{n}_1, \tilde{n}_2, \frac{1}{\delta^2}\right),
		\end{equation} 
		where $\tilde{n}_1$ and $\tilde{n}_2$ are defined in \eqref{2dn}, then  we have
		\begin{align}
			e_{u_n,\gamma_1 }:=\frac{\|\nabla u_n(\mathbf{x})\|_{L^2\left(\mathcal{N}_{1,1-\gamma_1}\left(\partial\Omega\right)\right)}}{\| u^{i}_{n}(\mathbf{x})\|_{L^2\left(\Omega\right)}}\ge \frac{n\delta}{2}\gg 1, \label{2d:gu_ratio} \\
			e_{u^s_n,\gamma_2 }:=\frac{\|\nabla u^s_n(\mathbf{x})\|_{L^2\left(B_{R}\setminus\left(\mathcal{N}_{2,\gamma_{2}-1}\left(\partial\Omega\right)\cup \Omega\right)\right)}}{\| u^{i}_{n}(\mathbf{x})\|_{L^2\left(\Omega\right)}}\ge \frac{n}{2}\gg 1.\label{2d:gu_ratio 1}
		\end{align}
		This indicates the occurrence of the surface resonances.
	\end{thm}

	% \(\frac{\|\nabla u_{n}\|^2_{L^2\left(\mathcal{N}_{1,1-\gamma_1}\right)}}{\| u^{i}_{n}\|^2_{L^2\left(\Omega\right)}} \gg 1\).  
	
	\begin{proof}
		We first establish the surface localization of the internal total field $ u_n $. Following the approach used to derive \eqref{eq:3d_gun} in Theorem \ref{thm:gradient_u}, we substitute the density function \eqref{eq:2d_phi} of the internal total field $ u_n $ into Lemma \ref{lem:2d_s_k}. Combining this with  \eqref{eq:jn} and \eqref{eq:hn}, we obtain the asymptotic expression for $ \nabla u_n(\mathbf{x}) $ in $ \Omega $
		\begin{equation}\label{eq:2d_gu}
			\begin{aligned}
				\nabla	u_{n}(\mathbf{x})&=\frac{\delta\beta_nk^{n}\vert\mathbf{x} \vert^{n-1}}{2^{n-1}(n-1)!}
				\left(\frac{i}{\delta +1}+\mathcal{O}\left(k^2\right)\right)e^{in\theta}\hat{e}_{\theta}\\
				&+\frac{\delta\beta_nk^{n}\vert\mathbf{x}\vert^{n-1}}{2^{n-1}(n-1)!}
				\left(\frac{1}{\delta +1}+\mathcal{O}\left(k^2\right)\right)e^{in\theta}\hat{e}_r,
			\end{aligned}
		\end{equation}
		where 
		\begin{equation}
			\hat{e}_r=\left(\begin{array}{c}\cos\theta\\\sin\theta\\\end{array}\right),
			\hat{e}_{\theta}=\left(\begin{array}{c}-\sin\theta\\\cos\theta\\\end{array}\right).
		\end{equation}

		%Using the gradient expression \eqref{eq:2d_gu} of the internal total field, then we compute the $L^2$ norm of $\nabla u_{n}(\mathbf{x})$. we exploit the orthogonality of the polar coordinate basis vectors \(\hat{e}_r\) and \(\hat{e}_\theta\),  yields:
		
		In view of \eqref{eq:2d_gu}, using the orthogonality  of \(\hat{e}_r\) and \(\hat{e}_\theta\), it yields that 
		
		\begin{align}\label{eq:2d_gu_n}
			\|\nabla u_{n}\|^2_{L^2\left(\mathcal{N}_{1,1-\gamma_1}\left(\partial\Omega\right)\right)}
			&=
			\frac{\pi\delta^2\vert\beta_n\vert^2k^{2n}(1-\gamma_{1}^{2n})}{\left[(\delta+1)(n-1)!\right]^22^{2n-3}n}\left(1+\mathcal{O}\left(k^2\right)\right).
		\end{align}
		Similarly, we can deduce that
		\begin{align}\label{eq:2d_gui_n}
			\| u^{i}_{n}\|^2_{L^2\left(\Omega\right)}
			&=
			\frac{\pi\vert\beta_n\vert^2k^{2n}}{(n!)^22^{2n}(n+1)}\left(1+\mathcal{O}\left(k^2\right)\right).
		\end{align}
		 By noting $n \geq \widetilde{n}_1$, $\delta \ll 1$ and $\epsilon \ll 1$, we know that $\frac{(1+\frac{1}{n})(1-\gamma^{2n}_{1})}{\left(1+\delta\right)^2} \ge \frac{1}{8}$. Under the sub-wavelength assumption \(k \ll 1\), according to \eqref{eq:2d_gu_n} and \eqref{eq:2d_gui_n}, it yields that
		\begin{align}
			\frac{\|\nabla u_n(\mathbf{x})\|^2_{L^2\left(\mathcal{N}_{1,1-\gamma_1}\left(\partial\Omega\right)\right)}}{\| u^{i}_{n}(\mathbf{x})\|^2_{L^2\left(\Omega\right)}}
			&=\frac{8\delta^2n(n+1)(1-\gamma^{2n}_{1})}{\left(1+\delta\right)^2}\left(1+\mathcal{O}\left(k^{2}\right)\right)\left(1-\mathcal{O}\left(k^{2}\right)\right)\label{eq:2d_gue1}\\
			&=\frac{8n^2\delta^2(1+\frac{1}{n})(1-\gamma^{2n}_{1})}{\left(1+\delta\right)^2}\left(1+\mathcal{O}\left(k^{2}\right)\right)\left(1-\mathcal{O}\left(k^{2}\right)\right)\notag\\
			&\ge n^2\delta^2\left(1+\mathcal{O}\left(k^{2}\right)\right)\left(1-\mathcal{O}\left(k^{2}\right)\right)\notag\\
			&= n^2\delta^2\left(1+\mathcal{O}\left(k^{2}\right)\right)\notag\\
			&\ge \frac{n^2\delta^2}{4}\notag. 
		\end{align}
		We prove  \eqref{2d:gu_ratio}.

		%which indicate a marked increase in the gradient of the total field energy compared to that of the incident wave.
		
		Next, we prove \eqref{2d:gu_ratio 1} to establish surface resonance of the external scattered field $ u^s_n $. By employing an argument analogous to \eqref{eq:2d_gu_n}, we derive
		
		% $ \frac{\|\nabla u_{n}^{s}\|^2_{L^2\left(\mathcal{N}_{2,\gamma_2-1}\left(\partial\Omega\right)\right)}}{\| u^{i}_{n}\|^2_{L^2\left(\Omega\right)}}\gg 1$. 

		\begin{align}\label{eq:2d_gus_n}
			\|\nabla u^s_n\|^2_{L^2\left(\mathcal{N}_{2,\gamma_2}\left(\partial\Omega\right)\right)}
			&= \frac{\pi(1-\delta)^2\vert \beta_{n}\vert^2k^{2n}\left(1-\frac{1}{\gamma^{2n}_{2}}\right)}{\left[(\delta+1)(n-1)!\right]^22^{2n-1}n}\left(1+\mathcal{O}\left(k^2\right)\right).
		\end{align}
		%By the  explicit expression $\eqref{eq:2d_gui_n}$ of $\| u^{i}_{n}\|^2_{L^2\left(\Omega\right)}$ and the result established in \eqref{eq:2d_gus_n} for $\|\nabla u^s_n\|^2_{L^2\left(\mathcal{N}_{2,\gamma_2-1}\left(\partial\Omega\right)\right)}$, we derive the ratio
		Combining \eqref{eq:2d_gus_n} and \eqref{eq:2d_gui_n}, we can deduce that
		\begin{align*}
			\frac{\|\nabla u^s_n(\mathbf{x})\|^2_{L^2\left(\mathcal{N}_{2,\gamma_2-1}\left(\partial\Omega\right)\right)}}{\| u^{i}_{n}(\mathbf{x})\|^2_{L^2\left(\Omega\right)}}
			&=\frac{2n(n+1)(1-\delta)^2\left(1-\frac{1}{\gamma^{2n}_{2}}\right)}{\left(1+\delta\right)^2}\left(1+\mathcal{O}\left(k^{2}\right)\right)\left(1-\mathcal{O}\left(k^{2}\right)\right)\\
			&=\frac{2n^2(1+\frac{1}{n})(1-\delta)^2\left(1-\frac{1}{\gamma^{2n}_{2}}\right)}{\left(1+\delta\right)^2}\left(1+\mathcal{O}\left(k^{2}\right)\right)\left(1-\mathcal{O}\left(k^{2}\right)\right)\\
			&\ge \frac{n^2}{2} \left(1+\mathcal{O}\left(k^{2}\right)\right)\left(1-\mathcal{O}\left(k^{2}\right)\right)\\
			&\ge \frac{n^2}{4},
		\end{align*}
			where the first inequality can be proved in a similar manner for \eqref{eq:2d_gue1} by noting $n \geq \widetilde{n}_2$. 
		
		The proof is complete.
	\end{proof}
	
	\begin{rem}
		When we choose the incident wave $u^i_n$ given by \eqref{eq:2d_ui}, for fixed parameters $\gamma_1\in(0,1)$ and $\gamma_2\in(1,2)$, recalling that $\tilde{n}_1$ and $\tilde{n}_2$ are defined in \eqref{2dn}, if  the boundary localization level $\epsilon$ is sufficient small and  $n$ satisfies \eqref{eq:n thm 42}, from Theorems \ref{2dthm1} and \ref{thm:2d_gradient}, one can directly conclude that both boundary localizations and surface resonances occur, which implies that quasi-Minnaert resonances are generated. 
	\end{rem}

	In Theorems \ref{2dthm1} and \ref{thm:2d_gradient}, for the high-contrast density parameter $\delta$ and the boundary localization level $\epsilon$, we can choose a specific incident wave that triggers the quasi-Minnaert resonance. Analogously to Corollary \ref{cor:3.2} and \ref{cor:gbr} in the three-dimensional case, in the two-dimensional scenario, the quasi-Minnaert resonance can also be induced by appropriately selecting the high-contrast parameter $\delta$.
	\begin{prop}\label{prop:3.2}
		Under the same assumptions as Theorem \ref{thm:2d_gradient}, consider the interaction between the incident wave \(u^i_n\) defined in \eqref{eq:3d_ui} and the scattering system \eqref{eq:system}. The high-contrast medium $\Omega$ behaves as a quasi-Minnaert resonator if the high-contrast parameter \(\delta\) satisfies:
		\begin{align}
			\delta \leqslant \beta =\min\left\{\frac{2\ln \gamma_1}{\ln\varepsilon - 3\ln \gamma_1}, \frac{2\ln \gamma_2}{3\ln{\gamma_2} - \ln \varepsilon} \right\},
		\end{align}
		and
		\begin{align}
			n \geqslant \frac{1}{\delta^2}.
		\end{align}Here, the first inequality guarantees the boundary localization as defined in \ref{df:blocalization}, while the second ensures the surface resonance specified in \ref{df:bresonance}.
		%	Consider the acoustic scattering problem \eqref{eq:system} in $\mathbb{R}^3$. Let $B_R$ denote the region containing the acoustic wave medium $\Omega$. We note that  $\mathcal{N}_{1,\gamma_1}\left(\partial\Omega\right)$ and $\mathcal{N}_{2,\gamma_2}\left(\partial\Omega\right)$ are defined in  \eqref{df:bregion}, where $\gamma_1\in(0,1)$ and $\gamma_2\in(1,R)$, respectively. Under the assumptions \eqref{eq:sw}-\eqref{eq:wn1}, for a sufficiently small level of boundary localization $\epsilon$ and we select the incident wave defined in $\eqref{eq:3d_ui}$ with n satisfying the condition
		%	\begin{equation*}
			%		n \geq \max\left(\tilde{n_1}, \tilde{n_2}, \frac{1}{\delta^2}\right),
			%	\end{equation*} 
		%	where $\tilde{n_1}$ and $\tilde{n_2}$ defined in \eqref{2dn}.
		%	With this choice of $n$, we have
		%	\begin{equation}
			%		\frac{\|\nabla u_n(\mathbf{x})\|_{L^2\left(\mathcal{N}_{1,\gamma_1}\left(\partial\Omega\right)\right)}}{\| u^{i}_{n}(\mathbf{x})\|_{L^2\left(\Omega\right)}}\ge 2n\delta\gg 1, \quad \frac{\|\nabla u^s_n(\mathbf{x})\|_{L^2\left(B_{R}\setminus\left(\mathcal{N}_{2,\zeta_2}\left(\partial\Omega\right)\cup \Omega\right)\right)}}{\| u^{i}_{n}(\mathbf{x})\|_{L^2\left(\Omega\right)}}\ge n\gg 1.
			%	\end{equation}
		%	This indicates the occurrence of quasi-Minnaert resonance as defined in Definition \ref{df:quasi-minnaert}.
	\end{prop}

	\section{Numerical examples}\label{sec6}

	In this section, we present extensive numerical examples to validate the theoretical findings of the preceding sections and to display that quasi-Minnaert resonance can induce the invisibility of the high-contrast medium $\Omega$. Specifically, for the incident wave $u^i_n$ defined by \eqref{eq:3d_ui} or \eqref{eq:2d_ui}, we examine various shapes of the high-contrast medium $\Omega$, although our primary analysis in previous sections focuses on radially symmetric geometries. Firstly, Our numerical results demonstrate that boundary localization and surface resonances for the internal total field $u|_{\Omega}$ and the external scattered field $u^s|_{\mathbb{R}^d \setminus \Omega}$ persist for general-shaped medium $\Omega$ when paired with a suitably chosen incident wave $u^i_n$. These findings underscore the critical role of high-contrast structures and the tuning of the incident wave in generating quasi-Minnaert resonances. Notably, as the index $n$ defining the incident wave increases, both the internal total field and the external scattered field exhibit increased boundary localizations. Similar observations hold for surface resonances. Furthermore, we demonstrate that under quasi-Minnaert resonance, as $n$ increases, the external scattered field decays rapidly, inducing a cloaking effect in the high-contrast medium $\Omega$.
	
	As noted in Remark~\ref{rem:21}, normalization of the $L^2$-norm of $u^i_n$ in $\Omega$ is advantageous, given our focus on the sub-wavelength regime. Accordingly, in our subsequent numerical experiments, we normalize the incident wave as $\frac{u^i_n(\mathbf{x})}{\|u^i_n\|_{L^2(\Omega)}}$.

	In two dimensions, the test domains for the high-contrast medium $\Omega$ include a disk, heart, Corner, and clover. In three dimensions, the validity of the theoretical results is verified exclusively through spherical geometries. The test domains are parameterized as follows:
	\begin{equation}
		\label{eq:test_domains}
		\begin{aligned}
			&\text{Circle:} \quad &\{(x_1,x_2) \mid x_1^2 + x_2^2 \leq 1\}, \\
			&\text{Heart:} \quad &\{(x_1,x_2) \mid \text{boundary parameterized by } s_1(t), \, 0 \leq t \leq 2\pi\}, \\
			&\text{Corner:} \quad &\{(x_1,x_2) \mid \text{boundary parameterized by } s_2(t), \, 0 \leq t \leq 2\pi\}, \\
			&\text{Clover:} \quad &\{(x_1,x_2) \mid \text{boundary parameterized by } s_3(t), \, 0 \leq t \leq 2\pi\}, \\
			&\text{Sphere:} \quad &\{(x_1,x_2,x_3) \mid x_1^2 + x_2^2 + x_3^2 \leq 1\},
		\end{aligned}
	\end{equation}
	where the boundary parameterizations are given by:
	\begin{equation}
		\label{eq:boundary_param}
		\begin{aligned}
			&s_1(t) = (2 \cos 2t + 3) (\cos t, -\sin t), \\
			&s_2(t) = (1 - \cos t) (\cos t, -\sin t), \\
			&s_3(t) = 4 (1 + \cos 3t + 3 \sin^2 3t) (\cos t, -\sin t).
		\end{aligned}
	\end{equation}
	
	In $\mathbb{R}^2$, triangular meshes for all domains are generated using FreeFem++. The acoustic scattering problem \eqref{eq:system} is discretized using the boundary finite element method. In $\mathbb{R}^3$, we employ COMSOL Multiphysics to generate high-quality tetrahedral meshes through its advanced meshing algorithms, followed by the application of the finite element method to numerically solve the scattering problem \eqref{eq:system}.

	\subsection{Boundary localizations} 
	In this subsection, we adopt the following physical parameters for the high-contrast medium $(\Omega; \rho_b, \kappa_b)$ and the homogeneous background $(\mathbb{R}^d \setminus \overline{\Omega}; \rho, \kappa)$:
	\begin{equation}
		\label{eq:physical_params}
		\kappa_b = \rho_b = 1, \quad \kappa = \rho = 1000, \quad \omega = 0.01, \quad \delta = \frac{\rho_b}{\rho} = 0.001,
	\end{equation}
	where $\omega$ denotes the incident frequency, and $\delta$ represents the high-contrast density ratio.
	
	We demonstrate through numerical examples that an appropriately selected incident wave $u^i_n$ can induce boundary localization in high-contrast medium $\Omega$ of general geometric shapes, paired with the incident wave $u^i_n$. In $\mathbb{R}^2$, the incident wave is given by:
	\begin{equation}\label{eq:uni 5-2}
		u^i_n(\mathbf{x})=J_{n}(k|\mathbf{x}|) \mathrm{e}^{{i}n\theta}.
	\end{equation}
	In $\mathbb{R}^3$, the incident wave is defined as:
	\begin{equation}\label{eq:in mn}
		u^i_n(\mathbf{x})=j_n(k|\mathbf{x}|) Y_n^m(\hat{\mathbf{x}}),\quad m\in \{0,\pm 1, \ldots, \pm n\}. 
	\end{equation}

	\begin{exm}\label{exm1}
		In the first example, we consider the high-contrast medium $\Omega$ to be a unit disk in $\mathbb{R}^2$. Figures~\ref{fig:circlei} and~\ref{fig:circlee} illustrate the modulus of the internal total field $|u_n(\mathbf{x})|$ and the external scattered field $|u_n^s(\mathbf{x})|$ associated with the incident wave $u_n^i$, for $n=10, 20, 40$, respectively. As depicted in Figure~\ref{fig:circlei}, the modulus of the internal total field $u_n$ is approximately zero far from the boundary $\partial \Omega$, while the red regions indicate that $|u_n|$ is predominantly concentrated near $\partial \Omega$. Similarly, Figure~\ref{fig:circlee} shows the modulus of the external scattered field $|u_n^s|$ for the same values of $n$. From left to right in Figures~\ref{fig:circlei} and~\ref{fig:circlee}, we observe the boundary localization of $|u_n|$ in $\Omega$ and $|u_n^s|$ in $B_2 \setminus \Omega$, where $B_2$ denotes the ball of radius 2 centered at the origin, with respect to the index $n$ defining the incident wave $u_n^i$.
		
		% Discussing boundary localization behavior
		It is evident that as the index $n$ increases, the regions of boundary localization become narrower. Specifically, the red regions corresponding to the maximum values of $|u_n|$ in $\Omega$ and $|u_n^s|$ in $B_2 \setminus \Omega$ become thinner. For given parameters $\gamma_1 = 0.1$ and $\gamma_2 = 1.1$, the boundary localization ratios, as defined in \eqref{eq:thm 41 ratio}, are reported for different values of $n$ in Table~\ref{tab:localization_ratios}.
		
		% Placeholder table for boundary localization ratios
		\begin{table}[h]
			\centering
			\caption{Boundary localization ratios $r_{u_n,\gamma_1}$ and $r_{u_n^s,\gamma_2}$ for different values of $n$, and given $\gamma_1=0.1$ and $\gamma_2=1.1$. }
			\label{tab:localization_ratios}
			\begin{tabular}{@{}ccc@{}}
				\toprule
				$n$ & $r_{u_n,\gamma_1}$ & $r_{u_n^s,\gamma_2}$ \\
				\midrule
				10 & 0.351158945548951 &  0.402881968027400 \\
				20 & 0.121947281445778 & 0.150714507543758 \\
				40 & 0.014794589503143 & 0.022200079874123 \\
				\bottomrule
			\end{tabular}
		\end{table}
		
		The values in Table~\ref{tab:localization_ratios} indicate that as $n$ increases, the boundary localization levels, namely $r_{u_n,\gamma_1}$ and $r_{u_n^s,\gamma_2}$, correspondingly decrease. We emphasize that the index $n$, as defined in \eqref{eq:thm 41 n1n2} in Theorem~\ref{2dthm1}, is a decreasing function of the the boundary localization level $\varepsilon$, as derived in \eqref{2dn}. Consequently, as $n$ increases, the generated wave fields $u_n$ and $u_n^s$ associated with $u_n^i$ exhibit increasingly pronounced boundary localization, consistent with the conclusions of Theorem~\ref{2dthm1}. These numerical results, as presented in Table~\ref{tab:localization_ratios}, provide robust validation of our theoretical findings in Theorem~\ref{2dthm1}.
	\end{exm}

	\begin{exm}
		For high-contrast medium $\Omega$ shaped as a heart, Cassini, and clover, as parameterized in \eqref{eq:test_domains}, we illustrate the boundary localization of the internal total wave field $u_n$ and the external scattered wave field $u_n^s$ in Figures~\ref{fig:corneri}, \ref{fig:cornere}, \ref{fig:cassinii}, \ref{fig:cassinie}, \ref{fig:cloveri}, and \ref{fig:clovere}, respectively, for different choices of the index $n=10, 20, 40$. Similar observations and conclusions can be drawn as those presented in Example~\ref{exm1} for the unit disk, including the narrowing of boundary localization regions as $n$ increases, indicating enhanced localization near $\partial \Omega$. 
		
	\end{exm}

	\begin{exm}
		In this example, we numerically validate the boundary localization properties established by Theorems~\ref{3dthm1} and~\ref{3dthm2}   for a spherical high-contrast medium $  \Omega  $. As illustrated in Figures~\ref{fig:spherei} and~\ref{fig:spheree}, both the internal total field $  u_n  $ and the external scattered field $  u^s_n  $ exhibit pronounced boundary localization near the plane $  z = 0  $ for the incident waves $  u^i_n  $ with indices $  n = 10, 20, 40  $ and corresponding $  m = n  $. 
		\iffalse	
		By appropriately selecting spherical harmonic functions $Y_n^m(\hat{\mathbf{x}})$, we can achieve targeted boundary localization in specified domains.
		In three-dimensional case, we consider an incident wave characterized by the following form 
		\begin{equation*}
			u^i_n(\mathbf{x})=j_n(k|\mathbf{x}|) Y_n^n(\hat{\mathbf{x}}),
		\end{equation*}
		and numerically verify Theorems \ref{3dthm1} and \ref{3dthm2}. As shown in Fig \ref{fig:spherei} and \ref{fig:spheree}, both the internal total field and the external scattered field exhibit pronounced boundary localization near the \(z=0\) plane as the parameter $n$ increases.
		\fi
	\end{exm}
	
	\subsection{Surface resonances}

		In this subsection, we conduct numerical experiments to validate the surface resonance established in Theorems~\ref{thm:gradient_u} and~\ref{thm:2d_gradient}, where the high-contrast medium $\Omega$ is either a unit disk or a corner in $\mathbb{R}^2$.

	The physical parameters for the high-contrast medium $\Omega$ and the homogeneous background $\mathbb{R}^2 \setminus \overline{\Omega}$ are set as follows:
	\begin{equation}
		\label{eq:physical_params_surface}
		\quad \kappa = \rho = 2,\quad \kappa_b = \delta\kappa, \quad \rho_b = \delta \rho
	\end{equation}
	where $\delta$ denotes the high-contrast ratio. Similar to \eqref{eq:physical_params}, the incident frequency is denoted by $\omega$. By varying the high-contrast parameter $\delta$, the density and bulk modulus of the medium $\Omega$ are changed accordingly.

	\begin{exm}\label{exm:4}

	For a high-contrast medium \(\Omega\), modeled as a unit disk, we numerically verify the occurrence of surface resonances using the incident wave \(u_n^i\) given by \eqref{eq:uni 5-2}, as defined in Definition~\ref{df:bresonance}. For parameters \(\delta = 0.1\), \(\omega = 0.01\), \(\gamma_1 = 0.1\), and \(\gamma_2 = 1.1\), the surface resonance ratios, as defined in \eqref{tab:resonance_ratios1}, are reported for different values of \(n\) in Table~\ref{tab:resonance_ratios1}. The quantities \(e_{u_n,\gamma_1}\) and \(e_{u_n^s,\gamma_2}\) characterize the highly oscillatory behavior of the internal total field \(u_n\) and the external scattered field \(u_n^s\) near the boundary of the high-contrast medium \(\Omega\), respectively. Additional numerical results for high-contrast parameter values \(\delta = 0.01\) and \(\delta = 0.001\), with varying \(n\) and fixed parameters \(\omega = 0.01\), \(\gamma_1 = 0.9\), and \(\gamma_2 = 1.1\), as in Table~\ref{tab:resonance_ratios1}, are presented in Tables~\ref{tab:resonance_ratios2} and~\ref{tab:resonance_ratios3}, respectively. Furthermore, Table~\ref{tab:resonance_ratios4} provides results for varying incident frequencies \(\omega\), with all other parameters fixed as in Table~\ref{tab:resonance_ratios2}.

		\begin{table}[h]
			\centering
			\caption{Surface resonance ratios $  e_{u_n, \gamma_1}  $ and $  e_{u_n^s, \gamma_2}  $ for varying values of $  n  $, with fixed parameters $  \delta = 0.1  $, $  \omega = 0.01  $, $  \gamma_1 = 0.9  $, and $  \gamma_2 = 1.1  $.}
			\label{tab:resonance_ratios1}
			\begin{tabular}{@{}ccc@{}}
				\toprule
				\(n\) & \(e_{u_n,\gamma_1}\) & \(e_{u_n^s,\gamma_2}\) \\
				\midrule
				10 & 1.147178455352589 & 10.985366617640477 
				
				 \\
				15 & 1.756433810650922
				 & 17.139246855660900  \\
				20 & 2.361301865270747 & 23.629359015615890
				 \\
				25 & 3.252528585448301 & 33.493521692552847 \\
				30 & 3.655900622080926
				 & 38.707162049879436 \\
				35 & 4.389561180031250 &  47.596569066804250 \\
				40 & 5.201989781299396 & 57.476053847873516 \\
				\bottomrule
			\end{tabular}
		\end{table}

		\begin{table}[h]
			\centering
			\caption{Surface resonance ratios \(e_{u_n,\gamma_1}\) and \(e_{u_n^s,\gamma_2}\) for different values of \(n\), with fixed parameters $  \delta = 0.01  $, $  \omega = 0.01  $, $  \gamma_1 = 0.9  $, and $  \gamma_2 = 1.1  $.}
			\label{tab:resonance_ratios2}
			\begin{tabular}{@{}ccc@{}}
				\toprule
				\(n\) & \(e_{u_n,\gamma_1}\) & \(e_{u_n^s,\gamma_2}\) \\
				\midrule
				10 &  0.124939746325538 & 11.964214562647447 \\
				15 & 0.191294034991458 & 18.666434610099657 \\
				20 & 0.257170501559110 & 25.734846520614006 \\
				25 & 0.325252858544830 & 33.493521692552847 \\
				30 & 0.398165869728157 & 42.156153667713298 \\
				35 &  0.448236731551825 & 51.839172544045205 \\
				40 & 0.566551184932064 & 62.597444018897221 \\
				\bottomrule
			\end{tabular}
		\end{table}

		\begin{table}[h]
			\centering
			\caption{Surface resonance ratios \(e_{u_n,\gamma_1}\) and \(e_{u_n^s,\gamma_2}\) for different values of \(n\), with fixed parameters $  \delta = 0.001  $, $  \omega = 0.01  $, $  \gamma_1 = 0.9  $, and $  \gamma_2 = 1.1  $.}
			\label{tab:resonance_ratios3}
			\begin{tabular}{@{}ccc@{}}
				\toprule
				\(n\) & \(e_{u_n,\gamma_1}\) & \(e_{u_n^s,\gamma_2}\) \\
				\midrule
				10 & 0.012606302732214 & 12.071779810809083 \\
				15 & 0.019301387972241 & 18.834256723156997
				 \\
				20 & 0.025948261422412 & 25.966217786629809
				 \\
				25 & 0.032817337921849 & 33.794267773392093
				 \\
				30 & 0.040174561358143 & 42.535162124719747 \\
				35 &  0.048236731551825 & 52.303700316366907 \\
				40 & 0.057164481250159 & 63.160231768805282 \\
				\bottomrule
			\end{tabular}
		\end{table}
		
\begin{table}[h]
	\centering
	\caption{Surface resonance ratios \(e_{u_n,\gamma_1}\) and \(e_{u_n^s,\gamma_2}\) for different values of \(\omega\), with fixed parameters \(\delta=0.01\) \(n=30\),\(\gamma_1 = 0.9\) and \(\gamma_2 = 1.1\).}
	\label{tab:resonance_ratios4}
	\begin{tabular}{@{}ccc@{}}
		\toprule
		\(\omega\) & \(e_{u_n,\gamma_1}\) & \(e_{u_n^s,\gamma_2}\) \\
		\midrule
		1e-1 & 0.398166679443199
		& 42.156225788946784 \\
		5e-2 & 0.398160912815661 & 42.155625551462421
		\\
		1e-2 &  0.398165869728157 & 42.156153667713298
		\\
		5e-3 & 0.398168747619778 & 42.156458470054325
		\\
		1e-3 & 0.398166183121177
		& 42.156186984549379 \\
		5e-4 &  0.398168825968573 & 42.156466799320640 \\
		\bottomrule
	\end{tabular}
\end{table}

Table \ref{tab:resonance_ratios1} presents numerical results demonstrating that, as the index \( n \) increases, the surface resonance ratios \( e_{u_n, \gamma_1} \) and \( e_{u_n^s, \gamma_2} \) exhibit a consistent monotonic increase. The values of \( e_{u_n, \gamma_1} \) in the second column are consistently \( \mathcal{O}(\delta) \) times larger than those of \( e_{u_n^s, \gamma_2} \) in the third column, consistent with the lower bounds established in equations \eqref{2d:gu_ratio} and \eqref{2d:gu_ratio 1}. Specifically, \( e_{u_n, \gamma_1} \) is of order \( n \delta \), while \( e_{u_n^s, \gamma_2} \) is of order \( n \). Theorem \ref{thm:2d_gradient} establishes that generating surface resonance for \( u_n \) requires \( n > \frac{1}{\delta^2} \). For \( \delta = 0.1 \), this implies \( n > 100 \) for the incident wave \( u_n^i \). As $n$  increases, solving the forward problem necessitates finer grid refinement, leading to a substantial escalation in computational resource requirements.  Due to computational resource constraints, we restrict \( n \leq 40 \) in the numerical examples presented below. Consequently, the values of \( e_{u_n^s, \gamma_2} \) in Table \ref{tab:resonance_ratios1} remain moderate, as \( n \leq 40 \).

		Similar conclusions apply to Tables \ref{tab:resonance_ratios2} and \ref{tab:resonance_ratios3}, which reveal additional observations. As the high-contrast parameter \( \delta \) decreases from 0.1 to 0.001 across Tables \ref{tab:resonance_ratios1}--\ref{tab:resonance_ratios3}, the numerical values of \( e_{u_n, \gamma_1} \) exhibit a significant decaying trend, while those of \( e_{u_n^s, \gamma_2} \) remain relatively stable with respect to \( \delta \). Specifically, the values of \( e_{u_n, \gamma_1} \) in the second column of Tables \ref{tab:resonance_ratios2} and \ref{tab:resonance_ratios3} are of order \( n \delta \), consistent with \eqref{2d:gu_ratio}, and the values of \( e_{u_n^s, \gamma_2} \) in the third column are of order \( n \), consistent with \eqref{2d:gu_ratio 1}. Additionally, Table \ref{tab:resonance_ratios4} demonstrates that both \( e_{u_n, \gamma_1} \) and \( e_{u_n^s, \gamma_2} \) are independent of the incident frequency \( \omega \). By Theorem \ref{thm:2d_gradient}, the lower bound for \( e_{u_n, \gamma_1} \) depends on \( n \) and \( \delta \), whereas the lower bound for \( e_{u_n^s, \gamma_2} \) depends solely on \( n \).

		\end{exm}

		\begin{exm}\label{ex:high ossi}
				
				In this example, we adopt the physical parameter settings from Table \ref{tab:resonance_ratios1} and fix the index \( n = 35 \) for the incident wave \( u_n^i \), defined in \eqref{eq:uni 5-2}. The high-contrast medium \( \Omega \) is a unit disk in \( \mathbb{R}^2 \). We aim to illustrate the high oscillatory behavior of \( \text{Re}(\nabla u_n) \) and \( \text{Re}(\nabla u_n^s) \), the real parts of the gradients of the total field \( u_n \) and scattered field \( u_n^s \), respectively, when surface resonance occurs, as shown in Table \ref{tab:resonance_ratios1}. 

Figures \ref{fig:gu_circle} and \ref{fig:gus_circle} display global and local vector field plots of \( \text{Re}(\nabla u_n) \) and \( \text{Re}(\nabla u_n^s) \), respectively, using arrows to represent vector magnitudes. Longer arrows indicate larger magnitudes, with red arrows denoting large magnitudes and blue arrows denoting small magnitudes. Figure \ref{fig:gu_circle_B} provides a local illustration of \( \text{Re}(\nabla u_n) \) in the region \( [0.5, 1] \times [0.5, 1] \), where \( \text{Re}(\nabla u_n) \) attains a maximum magnitude of 1.1832 in \( \mathcal{N}_{1, 1-\gamma_1} \). Similarly, Figure \ref{fig:gus_circle} shows that \( \text{Re}(\nabla u_n^s) \) reaches a maximum magnitude of 118.3301 in \( \mathcal{N}_{2, \gamma_2-1} \). Notably, the oscillations of \( \text{Re}(\nabla u_n^s) \) are significantly stronger than those of \( \text{Re}(\nabla u_n) \), as evidenced by comparing Figures \ref{fig:gu_circle} and \ref{fig:gus_circle}. This aligns with the numerical results in Table \ref{tab:resonance_ratios1}, where the lower bound of the surface resonance ratio \( e_{u_n^s, \gamma_2} \) is independent of the high-contrast parameter \( \delta = 0.01 \), while the lower bound of \( e_{u_n, \gamma_1} \) depends on \( \delta \), as established in Theorem \ref{thm:2d_gradient}.

	\end{exm}

	\begin{exm}
	
	In this example, we illustrate the oscillatory behavior of \( \text{Re}(\nabla u_n) \) and \( \text{Re}(\nabla u_n^s) \), the real parts of the gradients of the total field \( u_n \) and scattered field \( u_n^s \), respectively, for a high-contrast medium \( \Omega \) with a corner, as defined in \eqref{eq:test_domains}. The parameters are set as in Example \ref{ex:high ossi}. Analogous observations to those presented in Example \ref{ex:high ossi} hold. Notably, the maximum magnitude of \( \text{Re}(\nabla u_n^s) \) in \( \mathcal{N}_{2, \gamma_2-1} \) is 1078.962, approximately ten times larger than the corresponding maximum for a unit disk \( \Omega \). This suggests that geometric singularities, such as a corner, induce stronger oscillations in \( \text{Re}(\nabla u_n^s) \), which need to be further investigated in future studies.

	\begin{table}[h]
		\centering
		\caption{The quantity \(f_{u_n^s,\gamma_2}\) for different values of \(n\), with fixed parameters \(\delta=0.001\),  \(\omega=0.01\) and \(\gamma_2 = 1.1\) when $\Omega$ is a unit disk.}
		\label{tab:cloaking}
		\begin{tabular}{@{}ccc@{}}
			\toprule
			\(n\) & \(f_{u_n^s,\gamma_2}\) \\
			\midrule
			10 & 2716.801056817247
			 \\
			15 & 559.9223667039771 
			\\
			20 &  104.0501802886530 
			\\
			25 & 18.22848378817822 
			\\
			30 & 3.069068703133348
			 \\
			35 &  0.501958876145285
			 \\
			40 &  0.080304000018998
			\\
			\bottomrule
		\end{tabular}
	\end{table}
	\iffalse

	When the high-contrast medium $\Omega$ is a corner domain, we observe that the real value of the gradient of the internal total field, $\text{Re}(\nabla u_n)$, and the real value of the gradient of the external scattered field, $\text{Re}(\nabla u_n^s)$, exhibit highly oscillatory behavior. This is clearly illustrated in Figures~\ref{fig:gu_corner} and~\ref{fig:gus_corner}, which depict results for the incident wave $u^i_n$ with $n=35$ under the physical parameter $\delta=0.01$ and $\omega=0.01$. This suggests that surface resonance can occur in more general domains. 
	Consistent with the findings in Example~\ref{exm:4} for the unit disk, similar conclusions regarding  high oscillation can be drawn.
	
\fi 
	\end{exm}

	\subsection{Applications to invisibility cloaking}\label{sec5}

In this subsection, we demonstrate through numerical examples that quasi-Minnaert resonance can induce an invisibility effect for the high-contrast medium $\Omega$. We consider the interaction of an incident wave $u_n^i$, as defined in \eqref{eq:uni 5-2}, with the scattering system described by \eqref{eq:system}. The physical parameters for the high-contrast medium $\Omega$ and the homogeneous background $\mathbb{R}^2 \setminus \overline{\Omega}$ are specified in \eqref{eq:physical_params_surface}.

To quantify the invisibility cloaking effect for the high-contrast medium $\Omega$ under the influence of the incident wave $u_n^i$, we evaluate the following quantity:
\begin{equation}
\label{eq:invisibility_metric}
f_{u_n^s,\gamma_2} := \frac{\| u_n^s(\mathbf{x}) \|_{L^2(B_R \setminus (\mathcal{N}_{2,\gamma_2-1}(\partial \Omega) \cup \Omega))}}{\omega^2},
\end{equation}
where $u_n^s(\mathbf{x})$ is the scattered wave field associated with $u_n^i$ and $\Omega$, $R \in \mathbb{R}_+$ is the radius of a ball or disk $B_R$ centered at the origin, and $\gamma_2 \in \mathbb{R}_+$ with $\gamma_2 > 1$ is a predefined parameter such that $\mathcal{N}_{2,\gamma_2-1}(\partial \Omega) \cup \Omega \Subset B_R$. Here $\mathcal{N}_{2,\gamma_2-1}(\partial \Omega)$ denotes a $(\gamma_2-1)$-neighborhood of the boundary $\partial \Omega$ in $\mathbb{R}^n$.

A smaller value of $f_{u_n^s,\gamma_2}$ indicates a stronger invisibility cloaking effect. Specifically, $f_{u_n^s,\gamma_2}$ represents the ratio of the $L^2$-norm of the scattered field $u_n^s$ in the region $B_R \setminus (\mathcal{N}_{2,\gamma_2-1}(\partial \Omega) \cup \Omega)$ to $\omega^2$. Given that $\omega \ll 1$, a small $f_{u_n^s,\gamma_2}$ implies that the $L^2$-norm $\| u_n^s(\mathbf{x}) \|_{L^2(B_R \setminus (\mathcal{N}_{2,\gamma_2-1}(\partial \Omega) \cup \Omega))}$ is very small, confirming the occurrence of invisibility cloaking.

\begin{exm}\label{exm:7}

In this example, we consider the high-contrast medium $\Omega$ to be a unit disk in $\mathbb{R}^2$ and numerically report the values of $f_{u_n^s,\gamma_2}$, as defined in \eqref{eq:invisibility_metric}, with parameters $R=2$ and $\gamma_2=1.1$. As demonstrated in Example~\ref{exm1}, we have established boundary localization for the internal total wave field $u_n|_\Omega$ and the external scattered wave field $u_n^s$ associated with the incident wave $u_n^i$ and $\Omega$, where the physical parameters of the homogeneous background $\mathbb{R}^2 \setminus \overline{\Omega}$ are set as $\kappa = \rho = 1$ and the high-contrast ratio $\delta = 0.001$. The surface resonances of $u_n|_\Omega$ and $u_n^s$ are numerically demonstrated in Table~\ref{tab:resonance_ratios3} when the parameters $\kappa = \rho = 2$ for the homogeneous background $\mathbb{R}^2 \setminus \overline{\Omega}$, as specified in \eqref{eq:physical_params_surface}, with $\delta = 0.001$. 

We emphasize that surface resonances also occur for the case $\kappa = \rho = 1$ and $\delta = 0.001$, as confirmed by numerical experiments. For brevity, we choose not to report these additional results. Therefore, when the incident wave $u_n^i$ is selected as given by \eqref{eq:uni 5-2} and the physical parameters are set according to \eqref{eq:physical_params_surface}, we conclude that quasi-Minnaert resonances occur for the unit disk high-contrast medium $\Omega$. 

 We present numerical values of the invisibility cloaking quantity $f_{u_n^s,\gamma_2}$, defined in \eqref{eq:invisibility_metric}, with $R=2$ and $\gamma_2=1.1$. Table~\ref{tab:cloaking} reports these values for $n=10, 20, 40$, where $n$ is the index defining the incident wave $u_n^i$ as given by \eqref{eq:uni 5-2}. As $n$ increases from 10 to 40, $f_{u_n^s,\gamma_2}$ decreases significantly from $\mathcal{O}(10^3)$ to $\mathcal{O}(10^{-2})$. For a fixed high-contrast parameter $\delta$, a larger $n$ reduces the boundary localization ratio $r_{u_n^s,\gamma_2}$, defined by \eqref{eq:thm 41 ratio}, as reported in Table~\ref{tab:localization_ratios}, while increasing the surface resonance ratio $e_{u_n^s,\gamma_2}$, defined by \eqref{2d:gu_ratio 1}, as shown in Table~\ref{tab:resonance_ratios3}. This indicates a stronger quasi-Minnaert resonance, which, as illustrated in Table~\ref{tab:cloaking}, enhances the invisibility cloaking effect for the high-contrast medium $\Omega$.

% Analyzing boundary localization and cloaking
As shown in Figure~\ref{fig:circlee}, as $n$ increases, the boundary localization region of the scattered wave field $u_n^s$ narrows and approaches the boundary $\partial \Omega$. The blue regions, where $|u_n^s|$ is nearly vanishing, expand with increasing $n$. For $n=10$, there is an overlap between the boundary localization region of $u_n^s$ and the region $B_2 \setminus (\mathcal{N}_{2,\gamma_2-1}(\partial \Omega) \cup \Omega)$ with $\gamma_2=1.1$, resulting in large values of $|u_n^s|$ in this region. Given $\omega \ll 1$, the definition of $f_{u_n^s,\gamma_2}$ in \eqref{eq:invisibility_metric} yields a value of approximately 2800 for $n=10$, as reported in the first row of Table~\ref{tab:cloaking}. For $n=40$, Figure~\ref{fig:circlee} shows no overlap between the boundary localization region and $B_2 \setminus (\mathcal{N}_{2,\gamma_2-1}(\partial \Omega) \cup \Omega)$, indicating that $|u_n^s|$ is nearly vanishing in this region. Consequently, the value of $f_{u_n^s,\gamma_2}$ in the last row of Table~\ref{tab:cloaking} is of order $\mathcal{O}(10^{-2})$. This demonstrates that the high-contrast medium $\Omega$ is nearly invisible to the incident wave $u_n^i$ for large $n$.

While the cloaking effect is illustrated through numerical results in Table~\ref{tab:cloaking}, rigorous validation requires extensive and intricate analyses, which we defer to future research. This discovery opens a novel direction for the study and application of acoustic cloaking technologies.

\end{exm}

\begin{exm}

In this example, we consider the high-contrast medium $\Omega$ to be a general shape, such as a heart, Cassini oval, or clover, as parameterized in \eqref{eq:test_domains}. To compute the numerical values of the invisibility cloaking quantity $f_{u_n^s,\gamma_2}$, defined in \eqref{eq:invisibility_metric}, a refined mesh is required due to singularities, such as corners or high-curvature points, on the boundary $\partial \Omega$. However, due to computational resource limitations, we do not calculate the numerical values of $f_{u_n^s,\gamma_2}$ for these shapes in this example.

Based on observations of Figures~\ref{fig:corneri}, \ref{fig:cornere}, \ref{fig:cassinii}, \ref{fig:cassinie}, \ref{fig:cloveri}, and \ref{fig:clovere}, it is evident that the scattered wave field $u_n^s$, generated by the incident wave $u_n^i$ as defined in \eqref{eq:uni 5-2}, is nearly vanishing away from the high-contrast medium $\Omega$. This behavior suggests that an invisibility cloaking effect occurs for $\Omega$ when quasi-Minnaert resonances are induced, consistent with the findings for the unit disk in Example~\ref{exm1}. Rigorous analysis of this type of invisibility cloaking for general-shaped high-contrast media requires further investigation in future work.

%For high-contrast medium $\Omega$, parameterized as heart, Cassini, and clover shapes in \eqref{eq:test_domains}, observations analogous to those in Example~\ref{exm:7} are illustrated in Figures~\ref{fig:corneri}, \ref{fig:cornere}, \ref{fig:cassinii}, \ref{fig:cassinie}, \ref{fig:cloveri}, and \ref{fig:clovere}, respectively. When quasi-Minnaert resonance occurs, the magnitude of the external scattered field $|u^s_n|$ exhibits increasingly rapid decay near the boundary as $n$ increases, leading to a pronounced cloaking effect for the high-contrast medium $\Omega$ across general geometries.
\end{exm}

\bigskip
	\noindent\textbf{Acknowledgment.}
	The work of H. Diao is supported by National Natural Science Foundation of China  (No. 12371422) and the Fundamental Research Funds for the Central Universities, JLU. The work of H. Liu is supported by the Hong Kong RGC General Research Funds (projects 11311122, 11300821, and 11303125), the NSFC/RGC Joint Research Fund (project  N\_CityU101/21), the France-Hong Kong ANR/RGC Joint Research Grant, A-CityU203/19.

	\bigskip
	\noindent\textbf{CRediT authorship contribution statement.}  The authors declare that they have contributed equally to this work.

		\bigskip
	\noindent\textbf{Declaration of competing interest.}
	The authors declare that they have no known competing financial interests or personal relationships that could have appeared to
	influence the work reported in this paper.

	\bigskip
	\noindent\textbf{Data availability.}

	The numerical experimental data are not publicly available but are available from the authors under a reasonable request. 
	
	%No data was used for the research described in the article.

	\begin{figure}[htbp]
		\centering
		\begin{minipage}[t]{0.3\textwidth}
			\centering
			\includegraphics[width=\linewidth]{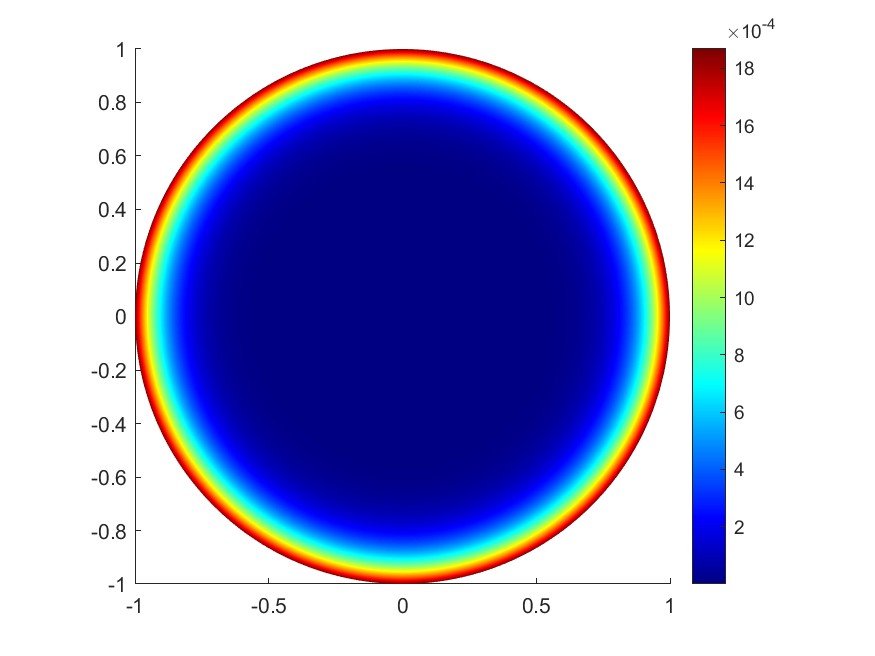}
		\end{minipage}%
		\begin{minipage}[t]{0.3\textwidth}
			\centering
			\includegraphics[width=\linewidth]{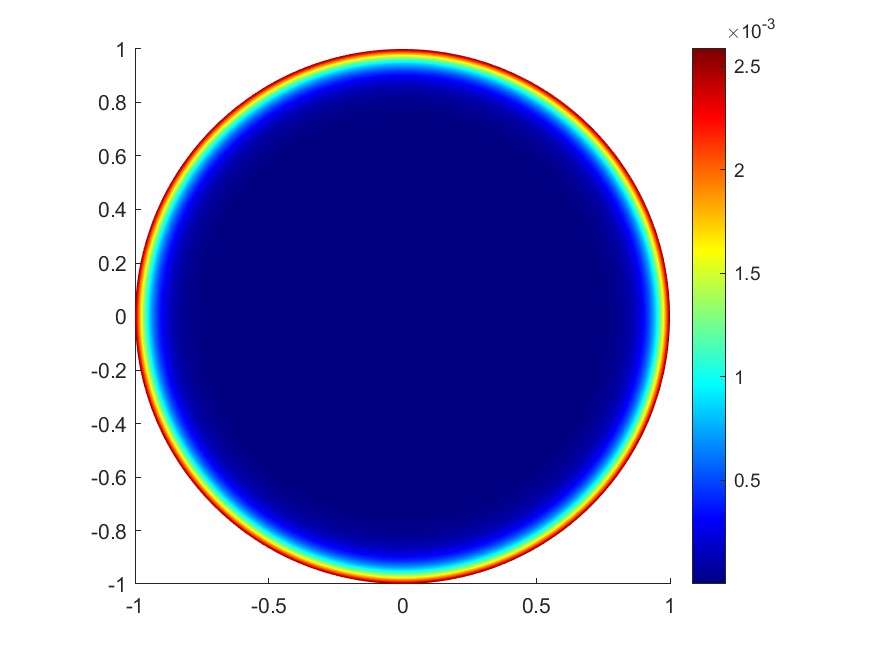}
		\end{minipage}%
		\begin{minipage}[t]{0.3\textwidth}
			\centering
			\includegraphics[width=\linewidth]{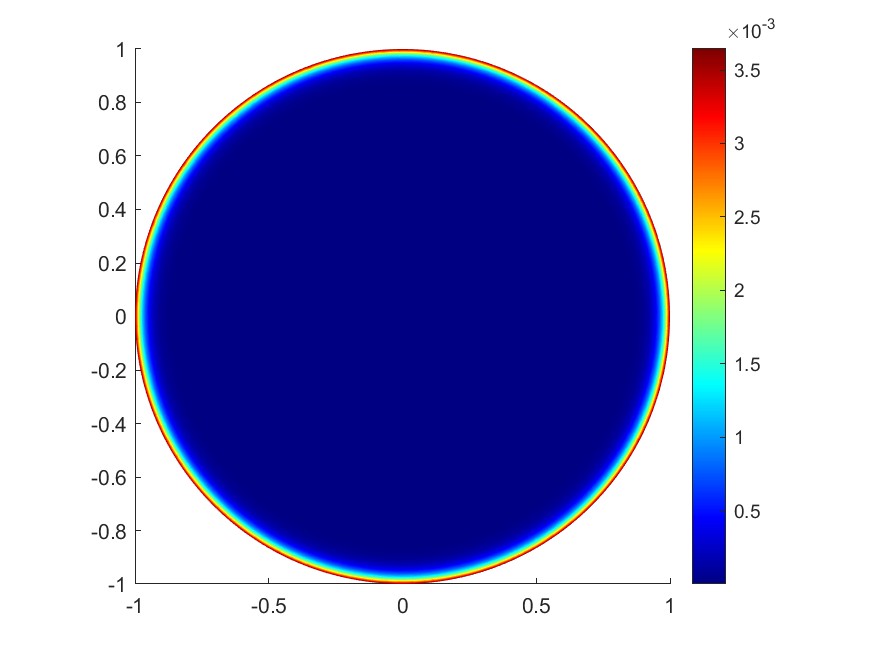}
		\end{minipage}%
		\caption{Boundary localization of the internal total field $|u_n(\mathbf{x})|$ for a unit disk high-contrast medium $\Omega$ in $\mathbb{R}^2$, with respect to different choices of the index $n=10, 20, 40$ defining the incident wave $u_n^i$. }
		\label{fig:circlei}
	\end{figure}

	\begin{figure}[htbp]
		\centering
		\begin{minipage}[t]{0.3\textwidth}
			\centering
			\includegraphics[width=\linewidth]{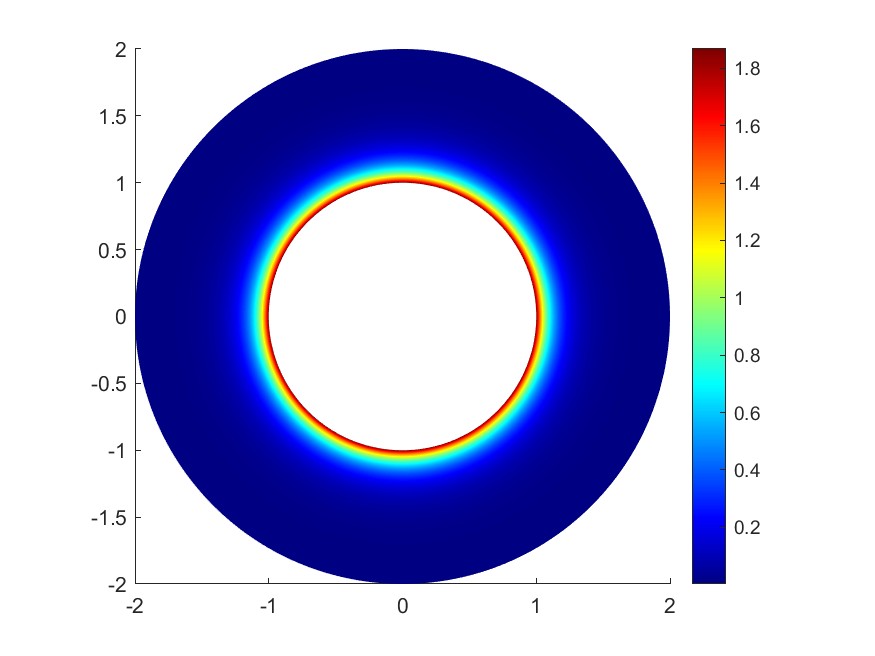}
		\end{minipage}%
		\begin{minipage}[t]{0.3\textwidth}
			\centering
			\includegraphics[width=\linewidth]{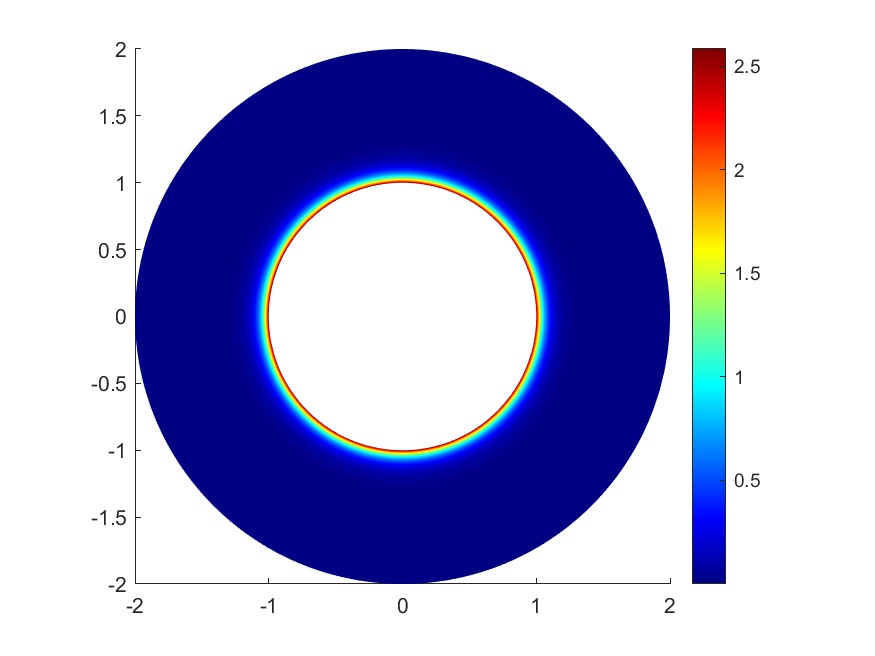}
		\end{minipage}%
		\begin{minipage}[t]{0.3\textwidth}
			\centering
			\includegraphics[width=\linewidth]{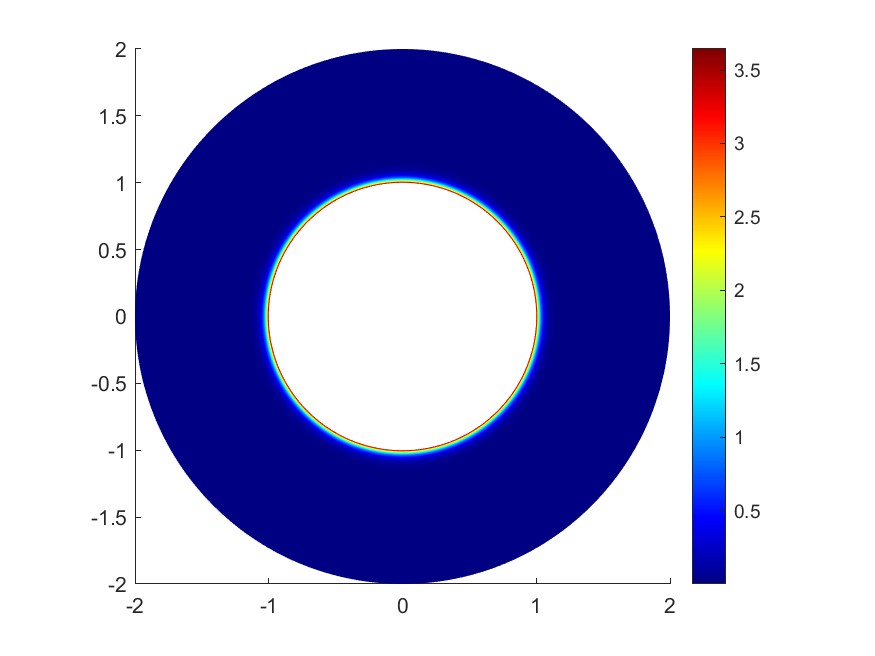}
		\end{minipage}%
		\caption{Boundary localization of the external scattered field $|u_n^s(\mathbf{x})|$ for a unit disk high-contrast medium $\Omega$ in $\mathbb{R}^2$, with respect to different choices of the index $n=10, 20, 40$ defining the incident wave $u_n^i$.}
		\label{fig:circlee}
	\end{figure}

	\begin{figure}[htbp]
		\centering
		\begin{minipage}[t]{0.3\textwidth}
			\centering
			\includegraphics[width=\linewidth]{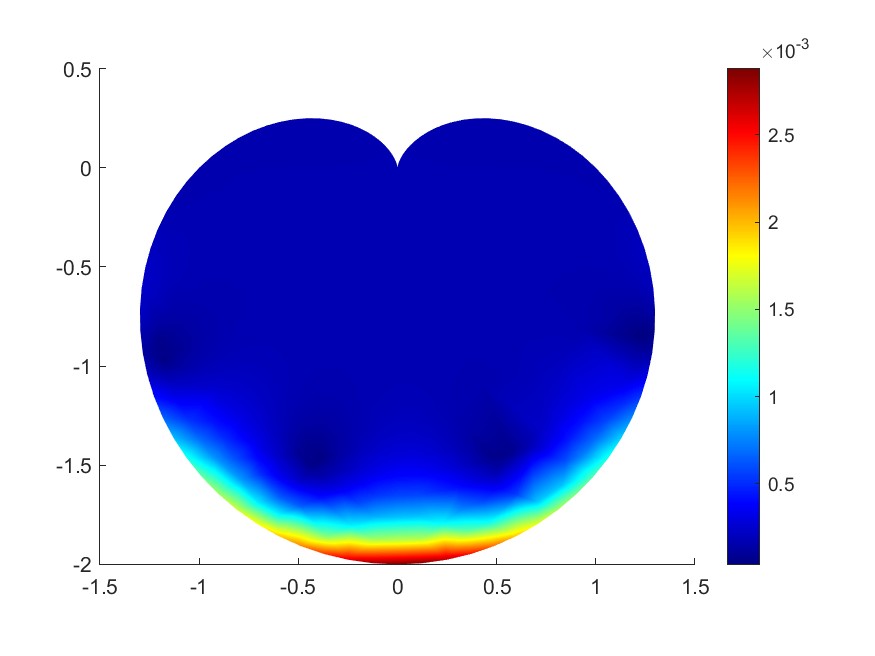}
		\end{minipage}%
		\begin{minipage}[t]{0.3\textwidth}
			\centering
			\includegraphics[width=\linewidth]{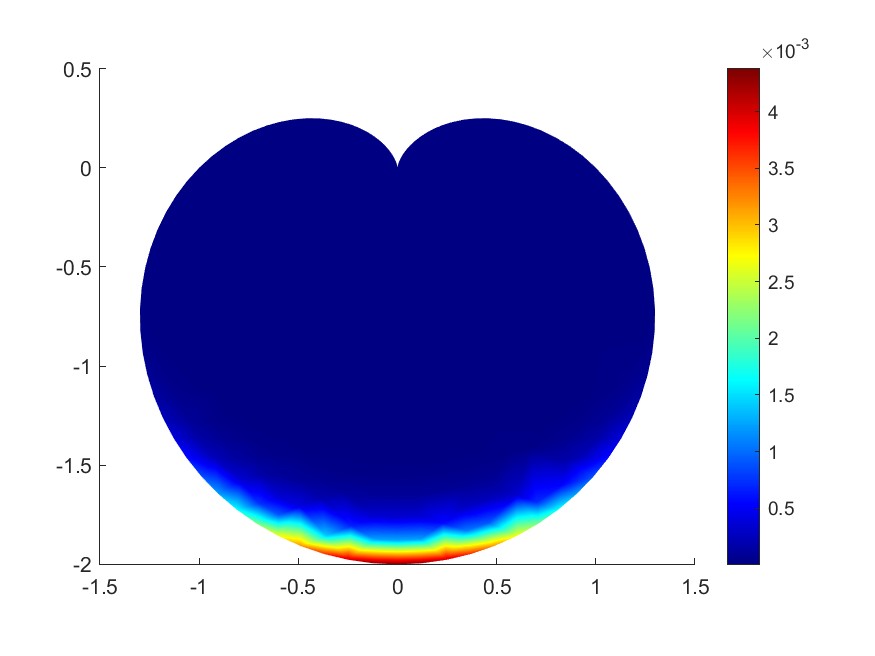}
		\end{minipage}%
		\begin{minipage}[t]{0.3\textwidth}
			\centering
			\includegraphics[width=\linewidth]{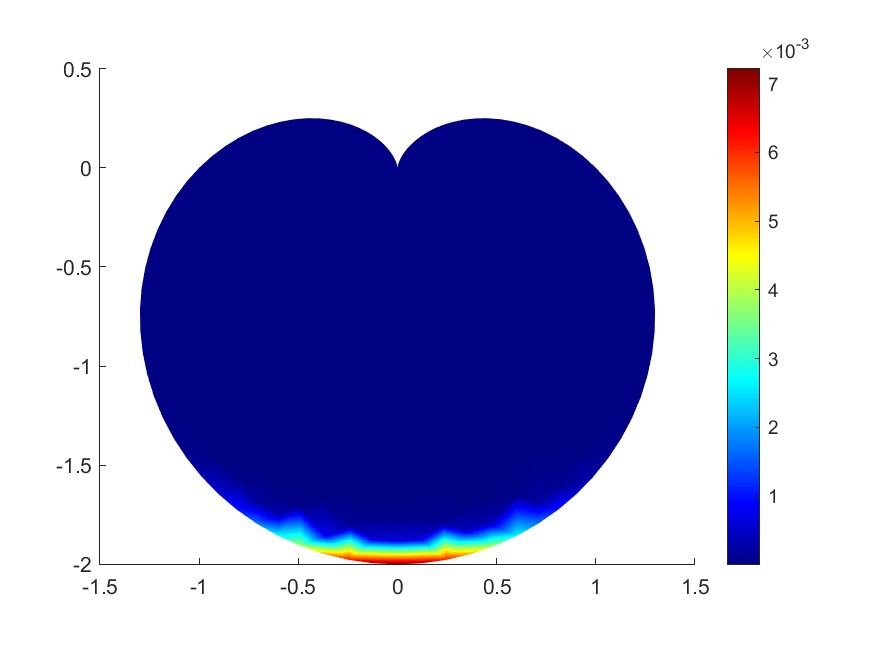}
		\end{minipage}%
		\caption{Boundary localization of the internal total field $|u_n(\mathbf{x})|$ for a heart-shaped high-contrast medium $\Omega$ in $\mathbb{R}^2$, with respect to different choices of the index $n=10, 20, 40$ defining the incident wave $u_n^i$.}
		\label{fig:corneri}
	\end{figure}
	
	\begin{figure}[htbp]
		\centering
		\begin{minipage}[t]{0.3\textwidth}
			\centering
			\includegraphics[width=\linewidth]{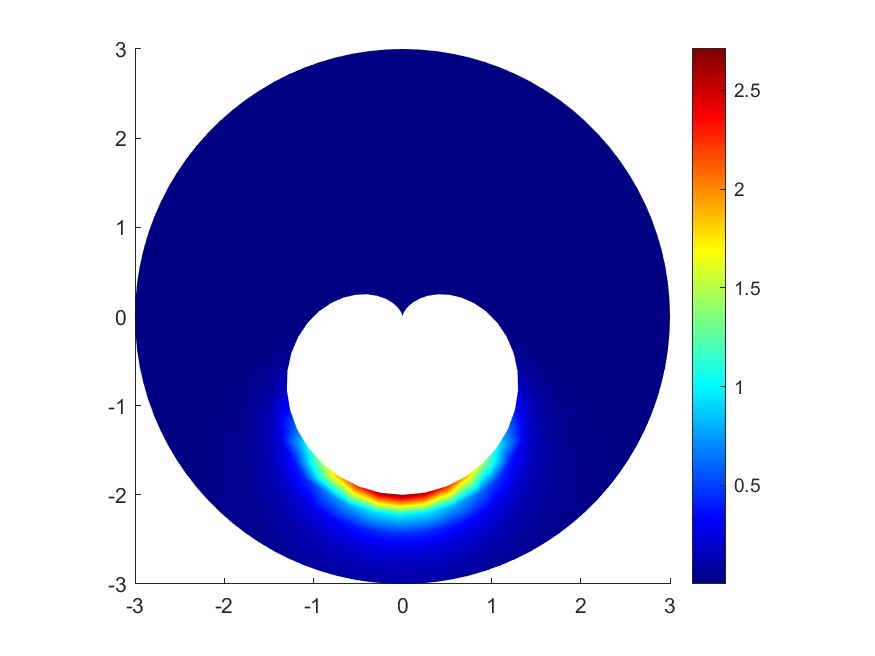}
		\end{minipage}%
		\begin{minipage}[t]{0.3\textwidth}
			\centering
			\includegraphics[width=\linewidth]{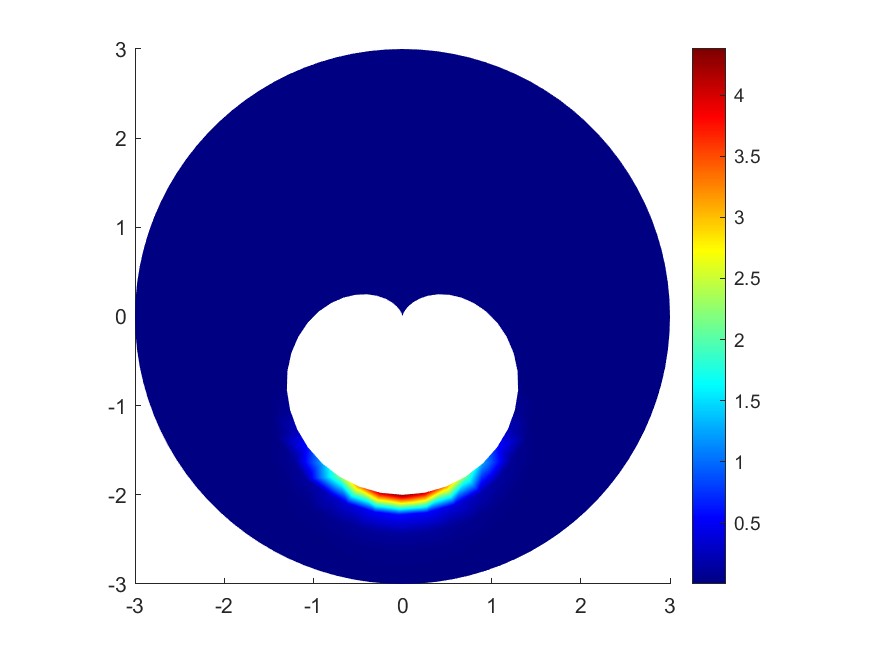}
		\end{minipage}%
		\begin{minipage}[t]{0.3\textwidth}
			\centering
			\includegraphics[width=\linewidth]{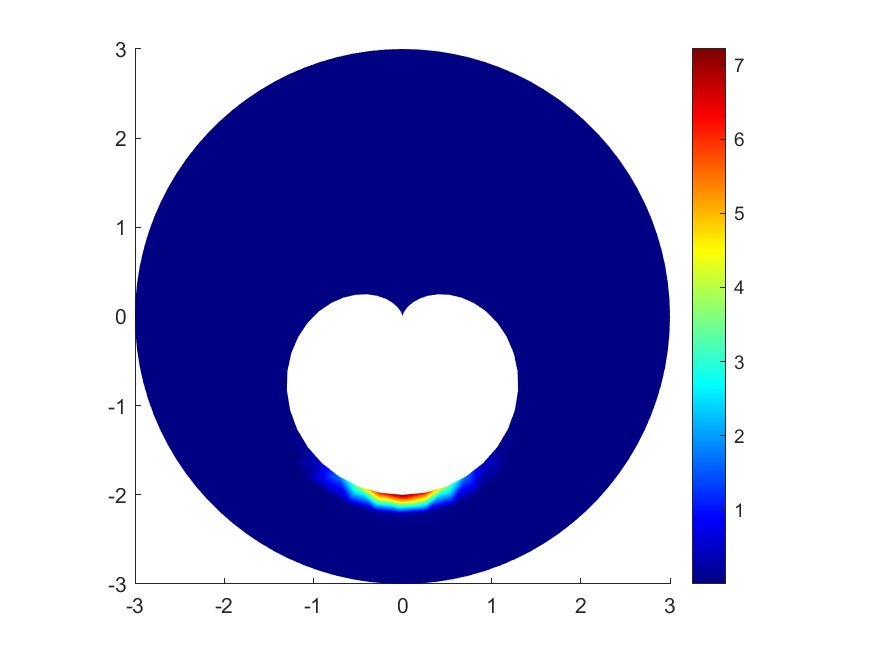}
		\end{minipage}%
		\caption{Boundary localization of the external scattered field $|u_n^s(\mathbf{x})|$ for a heart-shaped high-contrast medium $\Omega$ in $\mathbb{R}^2$, with respect to different choices of the index $n=10, 20, 40$ defining the incident wave $u_n^i$.}
		\label{fig:cornere}
	\end{figure}

	\begin{figure}[htbp]
		\centering
		\begin{minipage}[t]{0.3\textwidth}
			\centering
			\includegraphics[width=\linewidth]{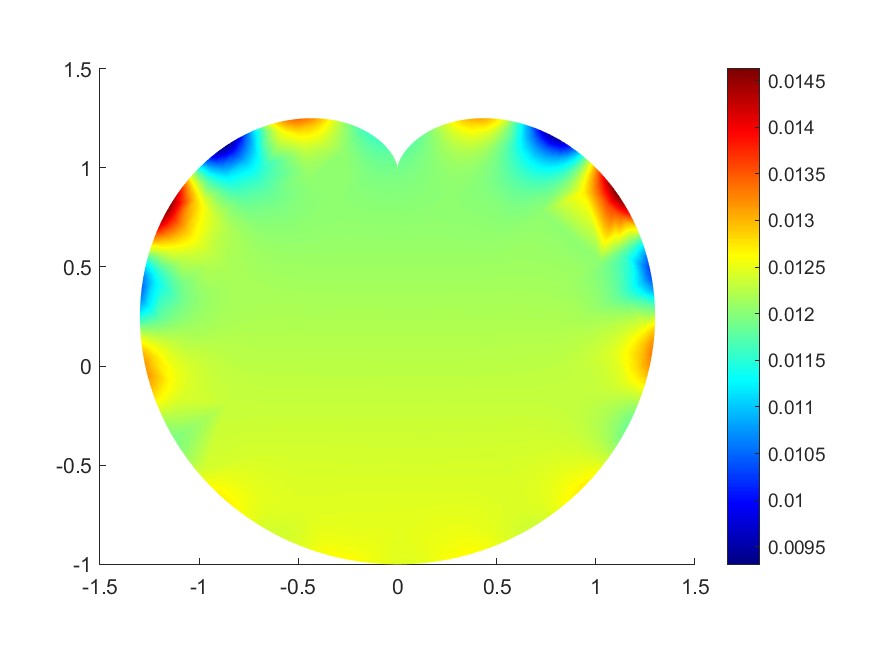}
		\end{minipage}%
		\begin{minipage}[t]{0.3\textwidth}
			\centering
			\includegraphics[width=\linewidth]{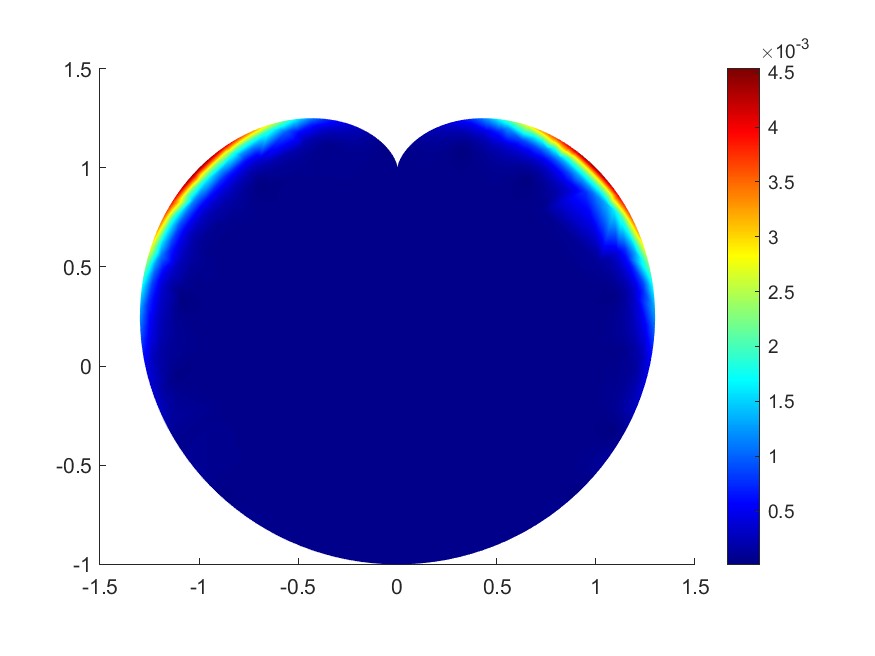}
		\end{minipage}%
		\begin{minipage}[t]{0.3\textwidth}
			\centering
			\includegraphics[width=\linewidth]{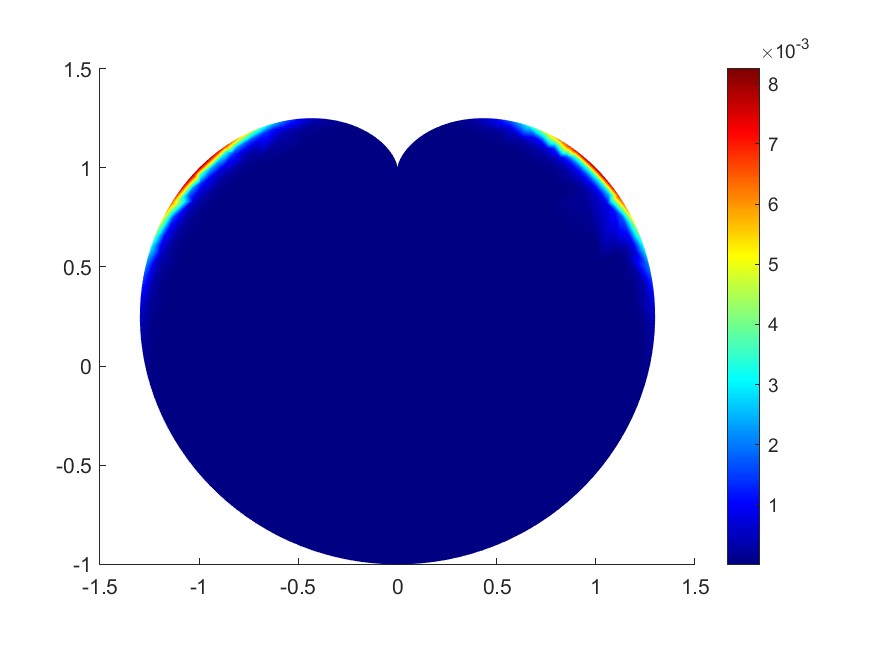}
		\end{minipage}%
		\caption{Boundary localization of the internal total field $|u_n(\mathbf{x})|$ for a heart-shaped high-contrast medium $\Omega$ in $\mathbb{R}^2$, with respect to different choices of the index $n=10, 20, 40$ defining the incident wave $u_n^i$.}
		\label{fig:cornerin}
	\end{figure}

	\begin{figure}[htbp]
		\centering
		\begin{minipage}[t]{0.3\textwidth}
			\centering
			\includegraphics[width=\linewidth]{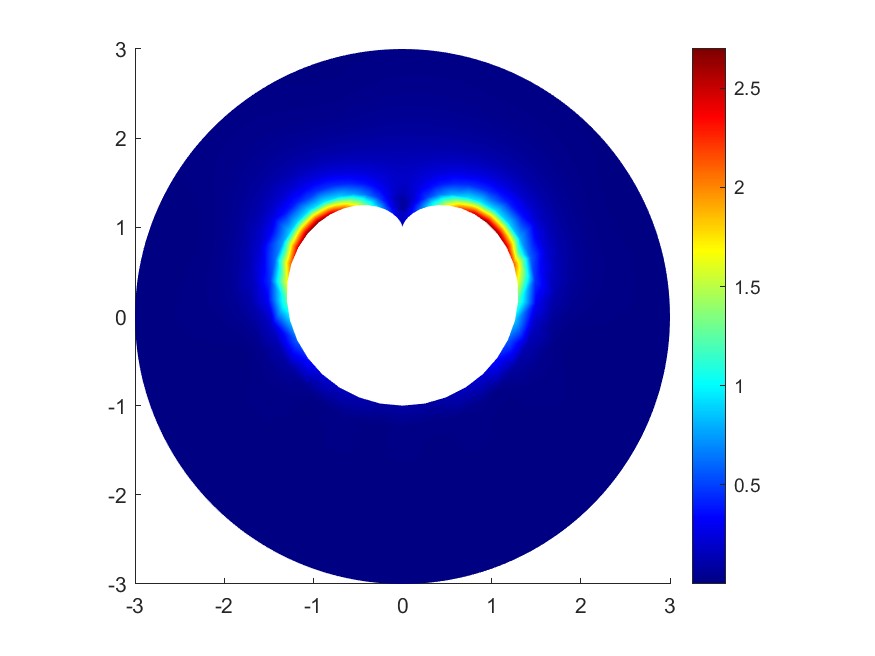}
		\end{minipage}%
		\begin{minipage}[t]{0.3\textwidth}
			\centering
			\includegraphics[width=\linewidth]{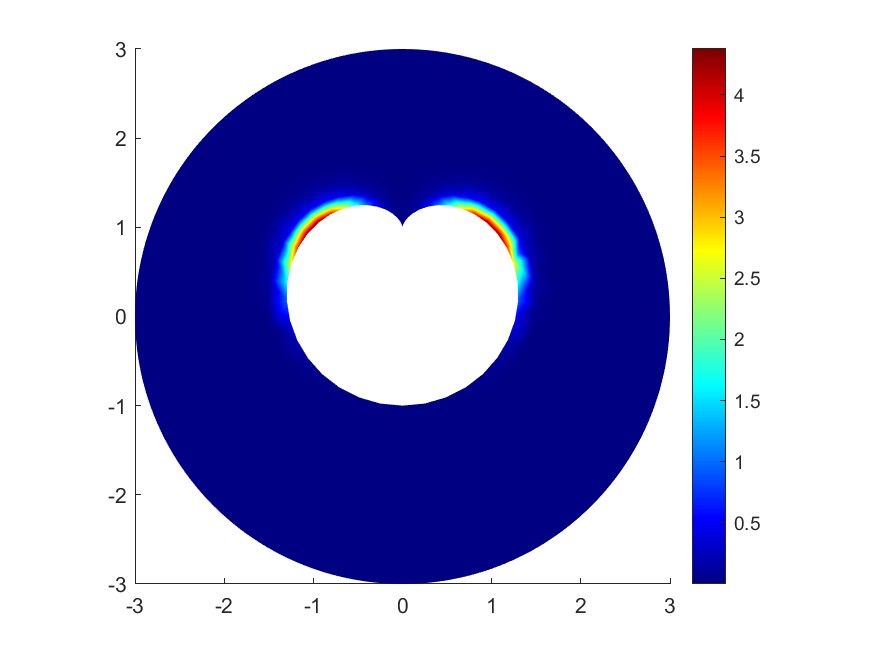}
		\end{minipage}%
		\begin{minipage}[t]{0.3\textwidth}
			\centering
			\includegraphics[width=\linewidth]{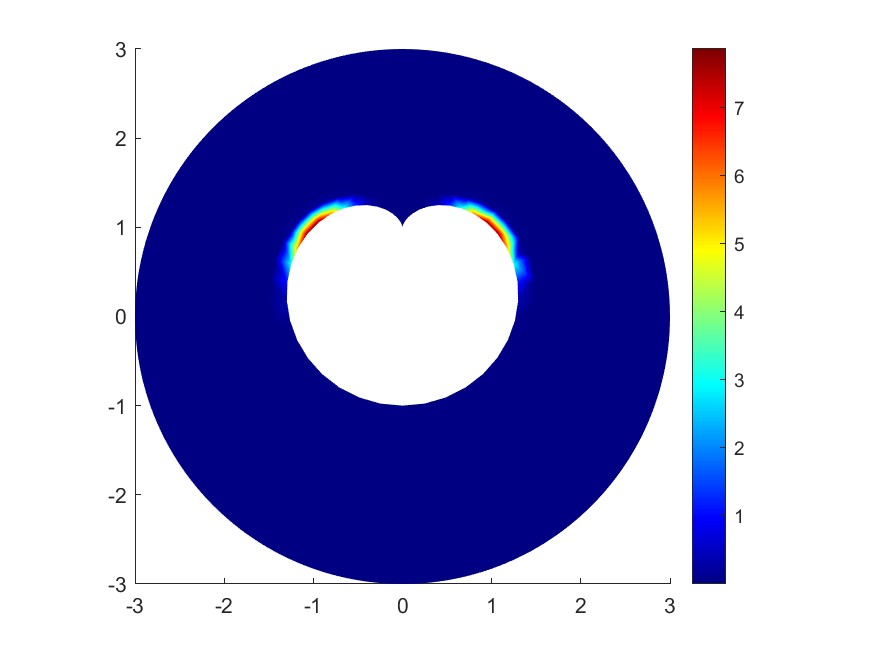}
		\end{minipage}%
		\caption{Boundary localization of the external scattered field $|u_n^s(\mathbf{x})|$ for a heart-shaped high-contrast medium $\Omega$ in $\mathbb{R}^2$, with respect to different choices of the index $n=10, 20, 40$ defining the incident wave $u_n^i$.}
		\label{fig:corneren}
	\end{figure}

	\begin{figure}[htbp]
		\centering
		\begin{minipage}[t]{0.3\textwidth}
			\centering
			\includegraphics[width=\linewidth]{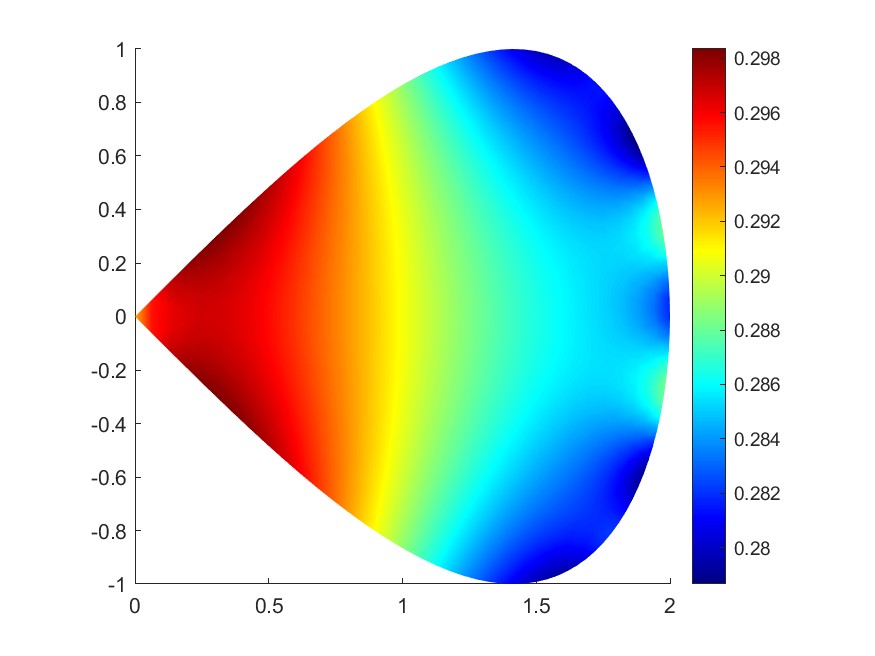}
		\end{minipage}%
		\begin{minipage}[t]{0.3\textwidth}
			\centering
			\includegraphics[width=\linewidth]{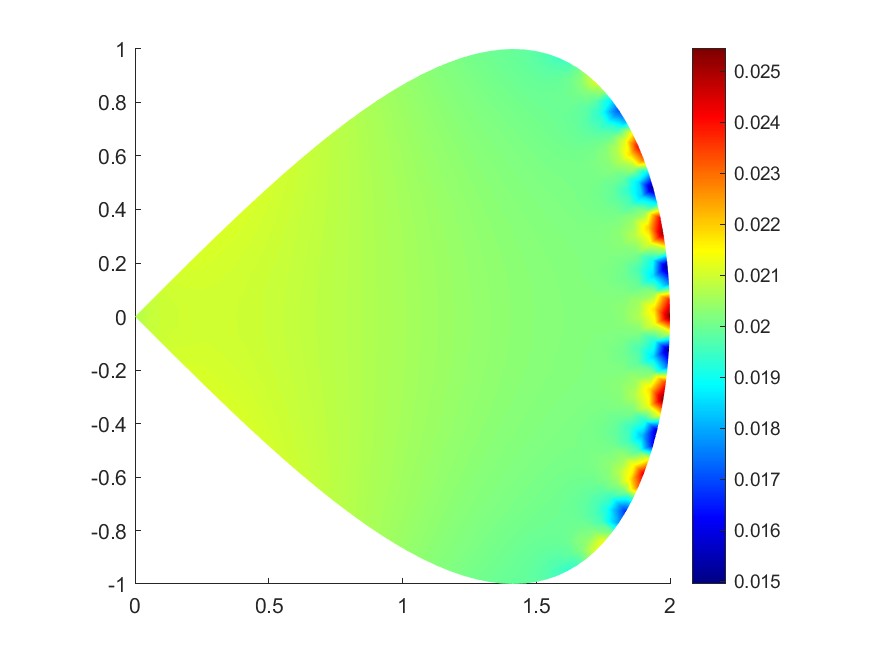}
		\end{minipage}%
		\begin{minipage}[t]{0.3\textwidth}
			\centering
			\includegraphics[width=\linewidth]{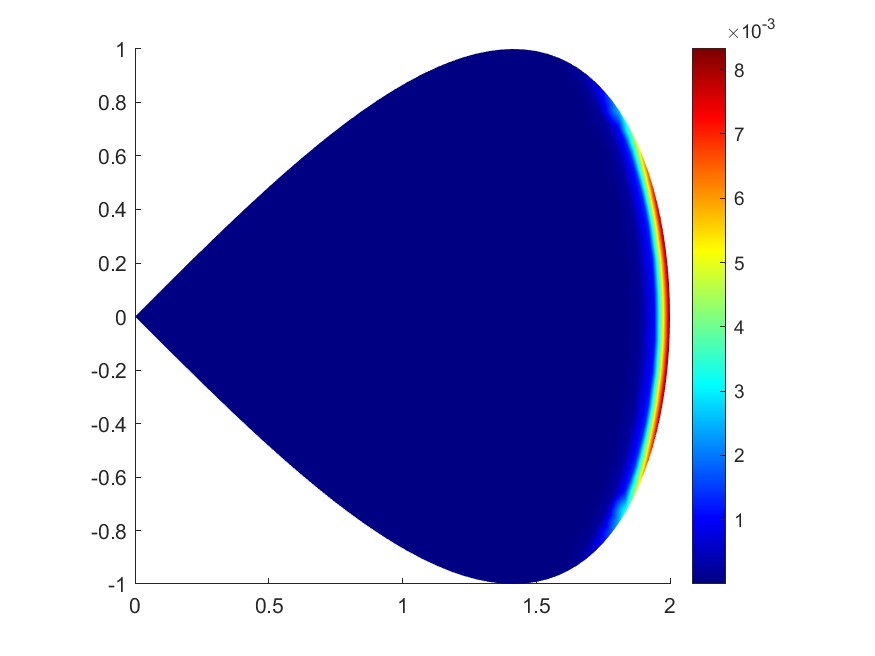}
		\end{minipage}%
		\caption{Boundary localization of the internal total field $|u_n(\mathbf{x})|$  for a corner-shaped high-contrast medium $\Omega$ in $\mathbb{R}^2$, with respect to different choices of the index $n=10, 20, 40$ and $m=0$ defining the incident wave $u_n^i$.}
		\label{fig:cassinii}
	\end{figure}
	
	\begin{figure}[htbp]
		\centering
		\begin{minipage}[t]{0.3\textwidth}
			\centering
			\includegraphics[width=\linewidth]{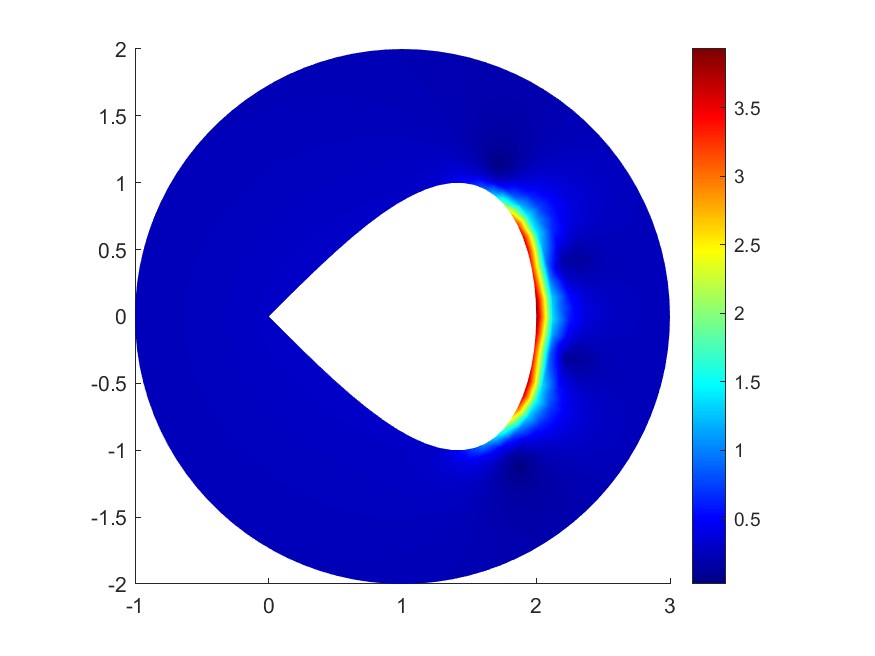}
		\end{minipage}%
		\begin{minipage}[t]{0.3\textwidth}
			\centering
			\includegraphics[width=\linewidth]{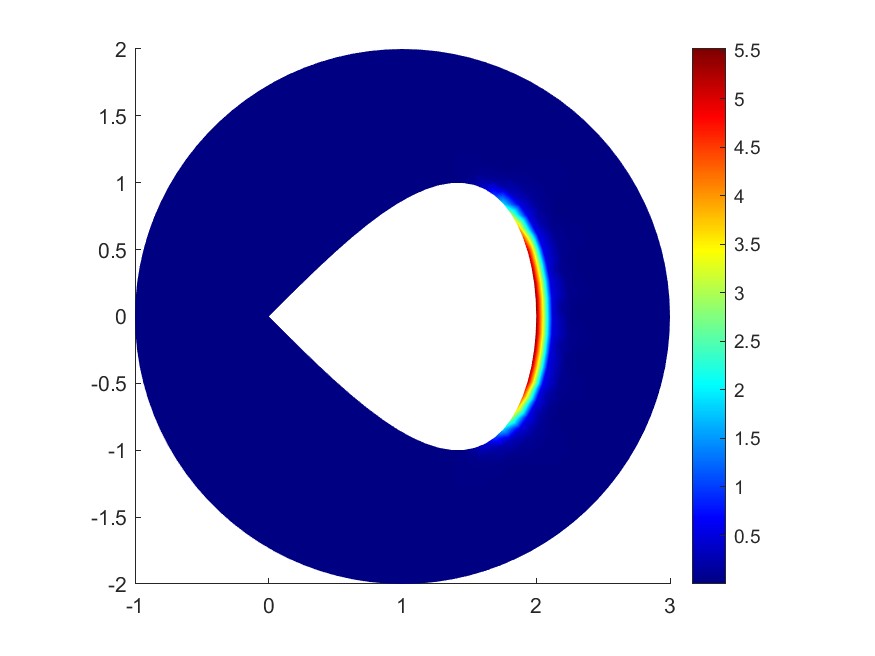}
		\end{minipage}%
		\begin{minipage}[t]{0.3\textwidth}
			\centering
			\includegraphics[width=\linewidth]{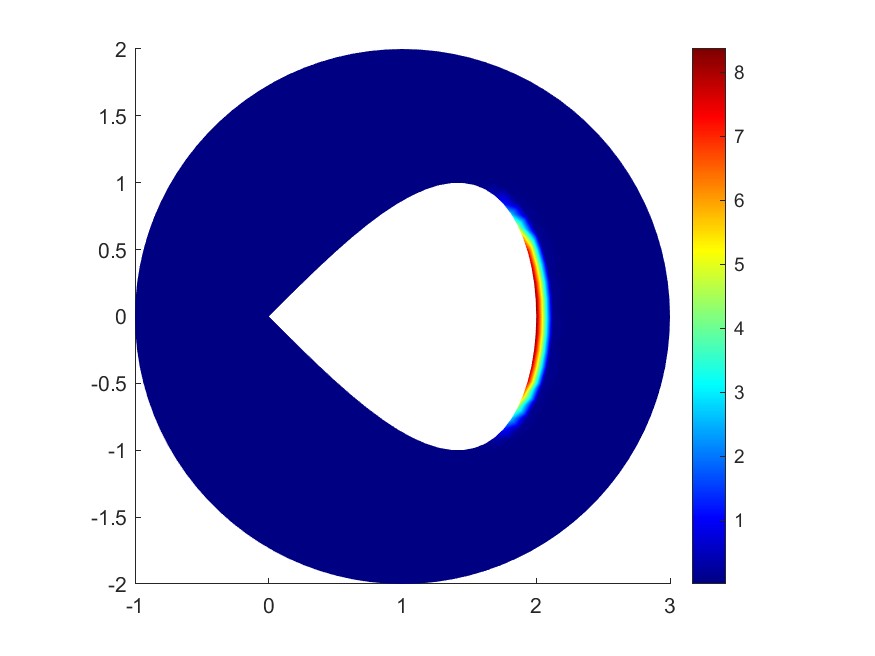}
		\end{minipage}%
		\caption{Boundary localization of the external scattered field $|u_n^s(\mathbf{x})|$ for a corner-shaped high-contrast medium $\Omega$ in $\mathbb{R}^2$, with respect to different choices of the index $n=10, 20, 40$ defining the incident wave $u_n^i$.}
		\label{fig:cassinie}
	\end{figure}
	
	\begin{figure}[htbp]
		\centering
		\begin{minipage}[t]{0.3\textwidth}
			\centering
			\includegraphics[width=\linewidth]{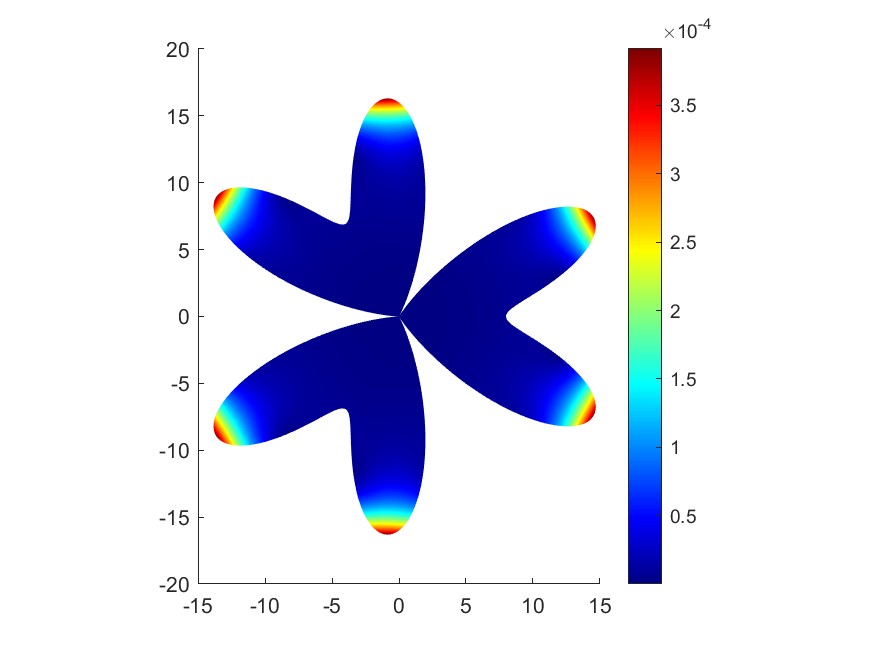}
		\end{minipage}%
		\begin{minipage}[t]{0.3\textwidth}
			\centering
			\includegraphics[width=\linewidth]{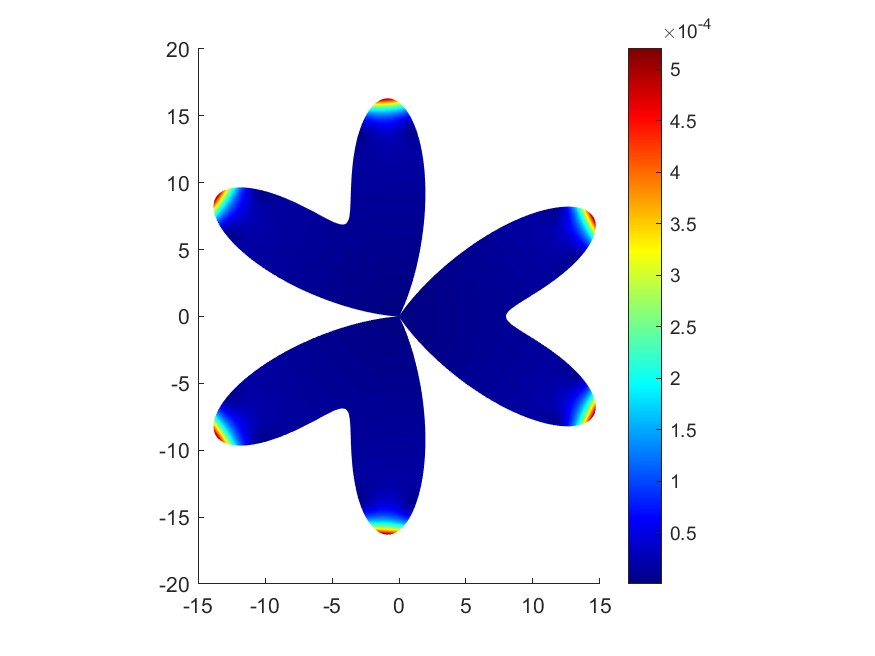}
		\end{minipage}%
		\begin{minipage}[t]{0.3\textwidth}
			\centering
			\includegraphics[width=\linewidth]{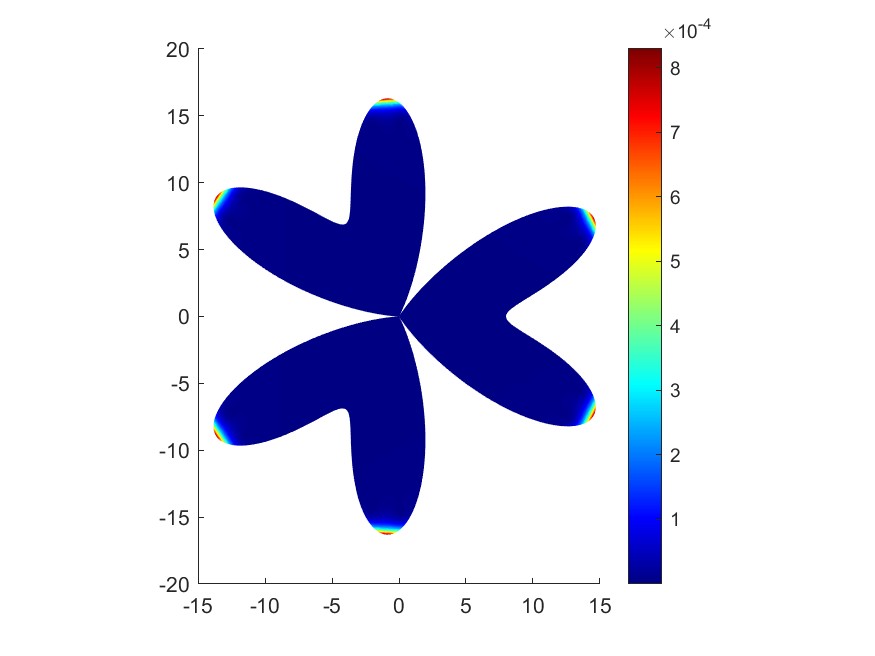}
		\end{minipage}%
		\caption{Boundary localization of the internal total field $|u_n(\mathbf{x})|$  for a clover-shaped high-contrast medium $\Omega$ in $\mathbb{R}^2$, with respect to different choices of the index $n=10, 20, 40$ defining the incident wave $u_n^i$.}
		\label{fig:cloveri}
	\end{figure}
	
	\begin{figure}[htbp]
		\centering
		\begin{minipage}[t]{0.3\textwidth}
			\centering
			\includegraphics[width=\linewidth]{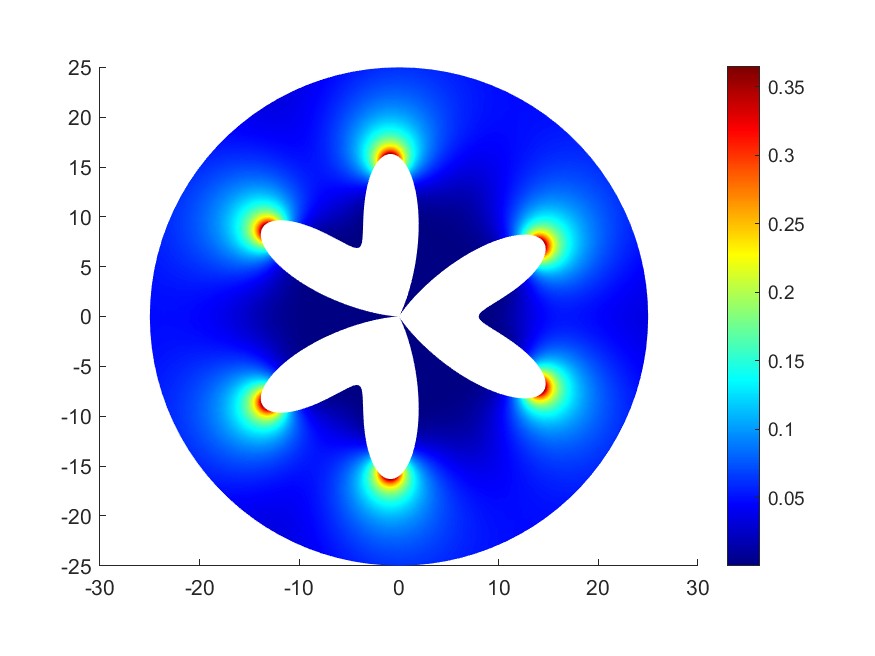}
		\end{minipage}%
		\begin{minipage}[t]{0.3\textwidth}
			\centering
			\includegraphics[width=\linewidth]{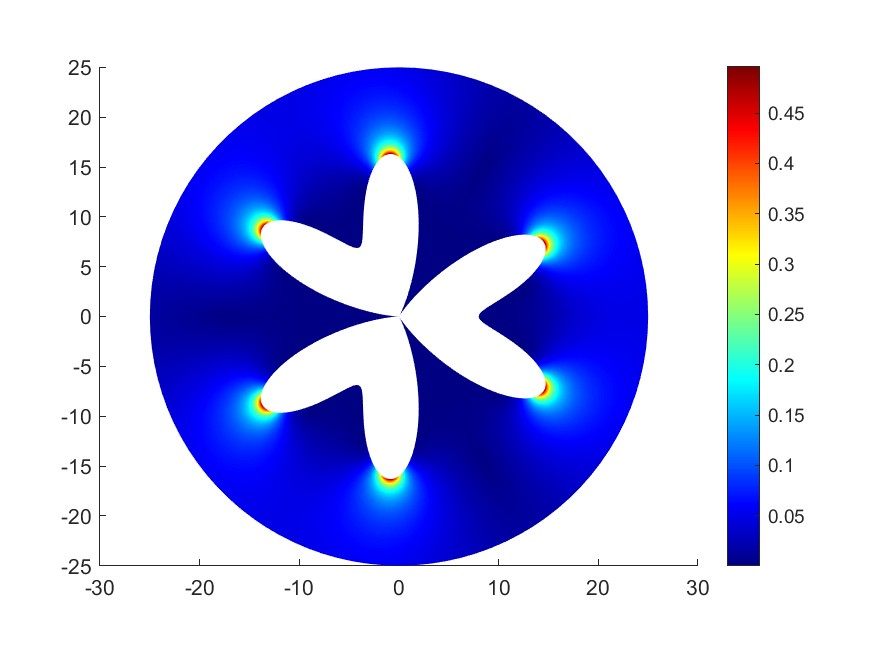}
		\end{minipage}%
		\begin{minipage}[t]{0.3\textwidth}
			\centering
			\includegraphics[width=\linewidth]{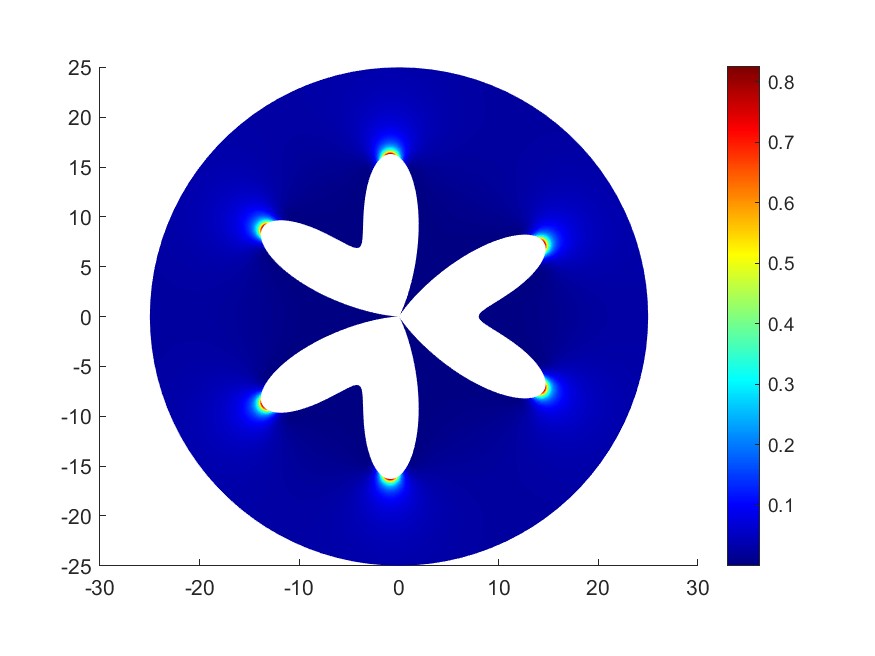}
		\end{minipage}%
		\caption{Boundary localization of the external scattered field $|u_n^s(\mathbf{x})|$ for a clover-shaped high-contrast medium $\Omega$ in $\mathbb{R}^2$, with respect to different choices of the index $n=10, 20, 40$ defining the incident wave $u_n^i$.}
		\label{fig:clovere}
	\end{figure}
	
	\begin{figure}[htbp]
		\centering
		\begin{minipage}[t]{0.3\textwidth}
			\centering
			\includegraphics[width=\linewidth]{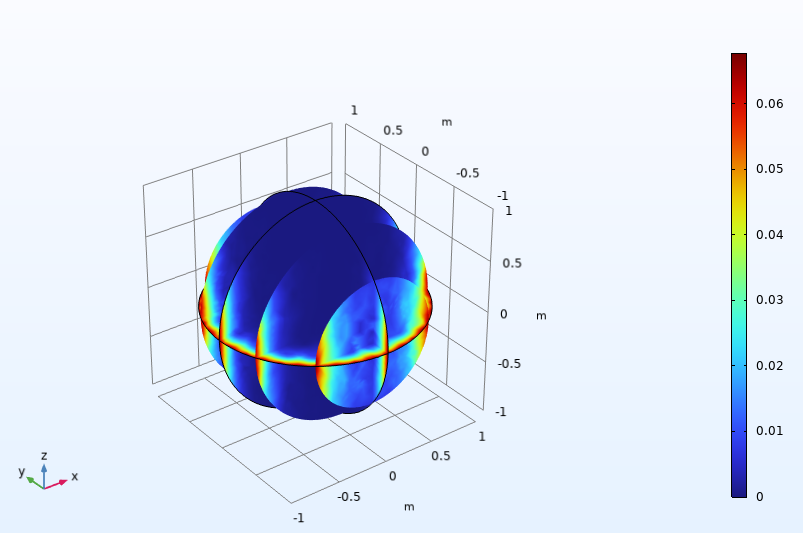}
		\end{minipage}%
		\begin{minipage}[t]{0.3\textwidth}
			\centering
			\includegraphics[width=\linewidth]{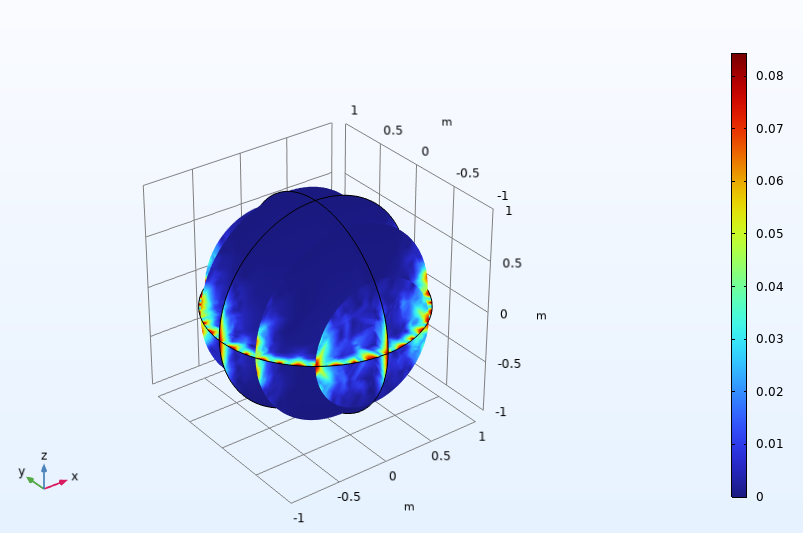}
		\end{minipage}%
		\begin{minipage}[t]{0.3\textwidth}
			\centering
			\includegraphics[width=\linewidth]{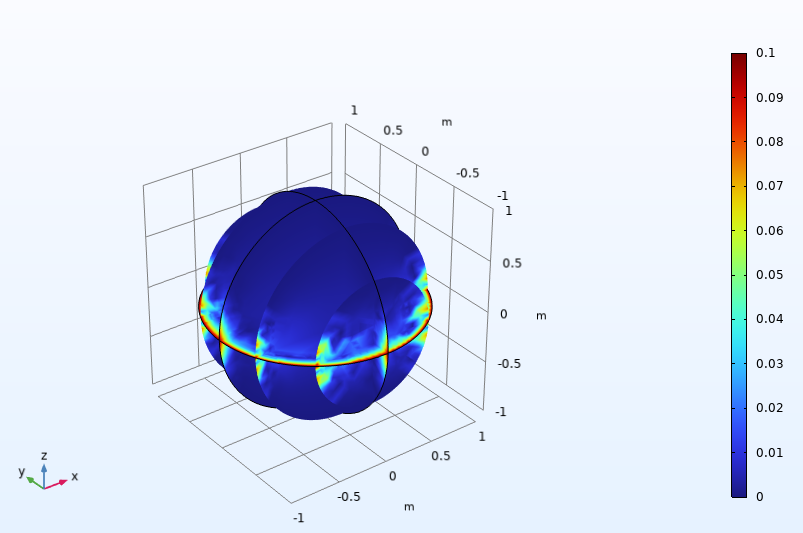}
		\end{minipage}%
		\caption{Boundary localization of the internal total field $|u_n(\mathbf{x})|$ for a ball high-contrast medium $\Omega$ in $\mathbb{R}^2$, with respect to different choices of the index $n=10, 20, 40$ and $m=n$ defining the incident wave $u_n^i$.}
		\label{fig:spherei}
	\end{figure}

	\begin{figure}[htbp]
		\centering
		\begin{minipage}[t]{0.3\textwidth}
			\centering
			\includegraphics[width=\linewidth]{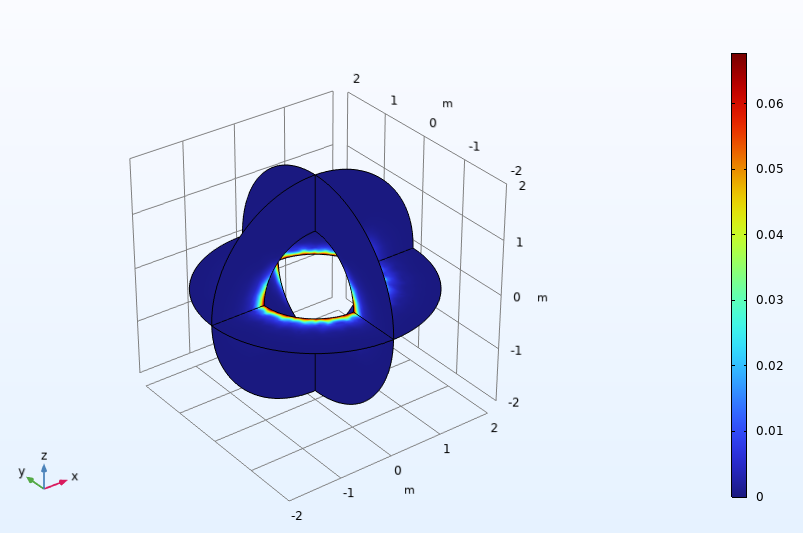}
		\end{minipage}%
		\begin{minipage}[t]{0.3\textwidth}
			\centering
			\includegraphics[width=\linewidth]{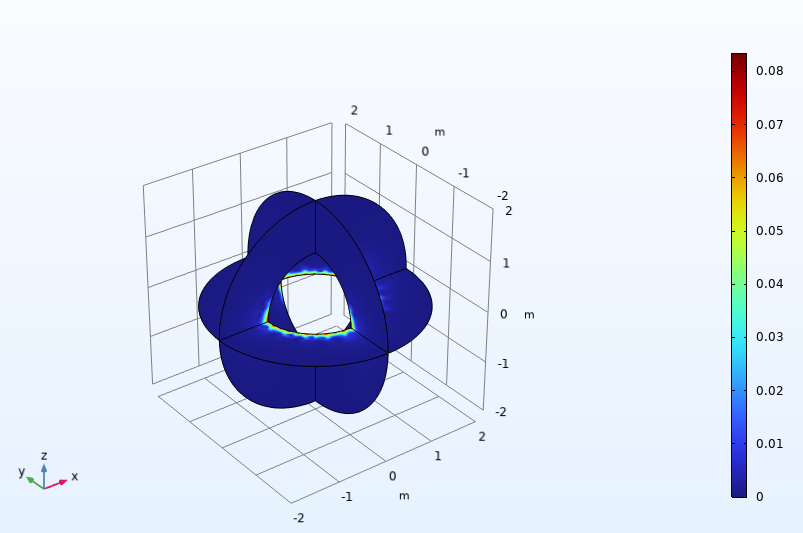}
		\end{minipage}%
		\begin{minipage}[t]{0.3\textwidth}
			\centering
			\includegraphics[width=\linewidth]{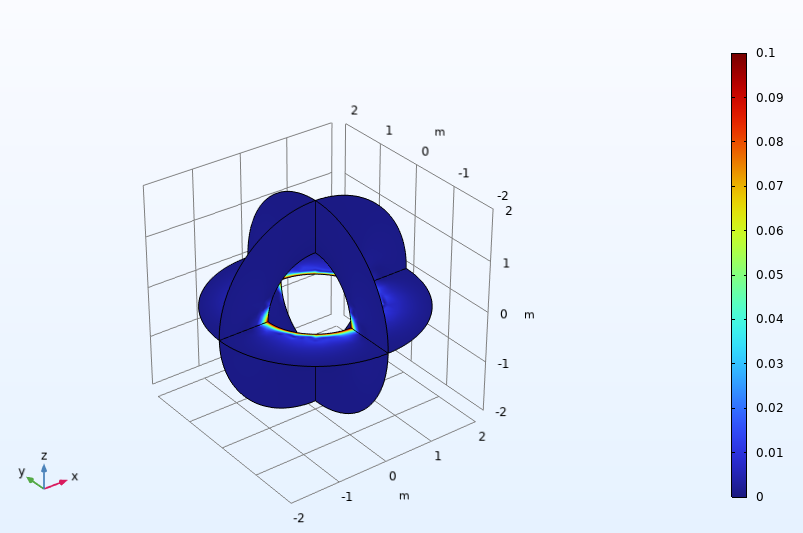}
		\end{minipage}%
		\caption{Boundary localization of the external scattered field $|u_n^s(\mathbf{x})|$ for a ball high-contrast medium $\Omega$ in $\mathbb{R}^2$, with respect to different choices of the index $n=10, 20, 40$ and $m=n$ defining the incident wave $u_n^i$.}
		\label{fig:spheree}
	\end{figure}

	\begin{figure}[h]
		\centering
		\begin{subfigure}{0.45\textwidth}
			\centering
			% Placeholder for Cassini internal field
			\includegraphics[width=\textwidth]{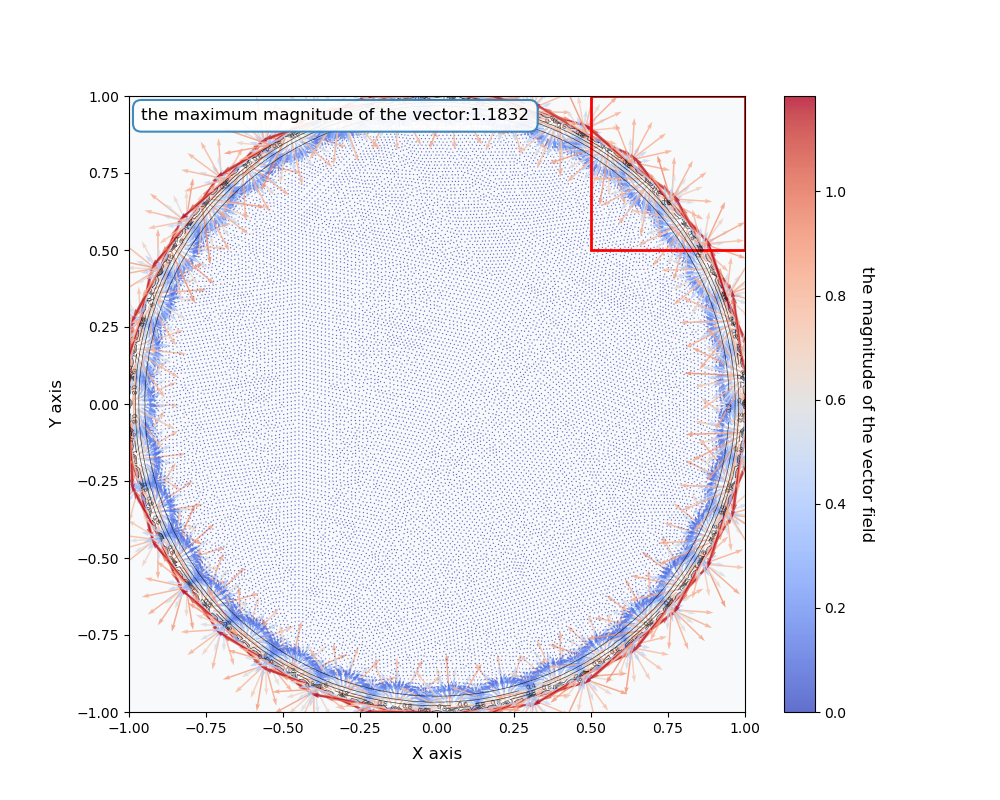}
			\caption{High oscillations of the real part of the gradient of the internal total field \(  u_n \).}  
			\label{fig:gu_circle_A}
		\end{subfigure}
		\hfill
		\begin{subfigure}{0.45\textwidth}
			\centering
			% Placeholder for Cassini external field
			\includegraphics[width=\textwidth]{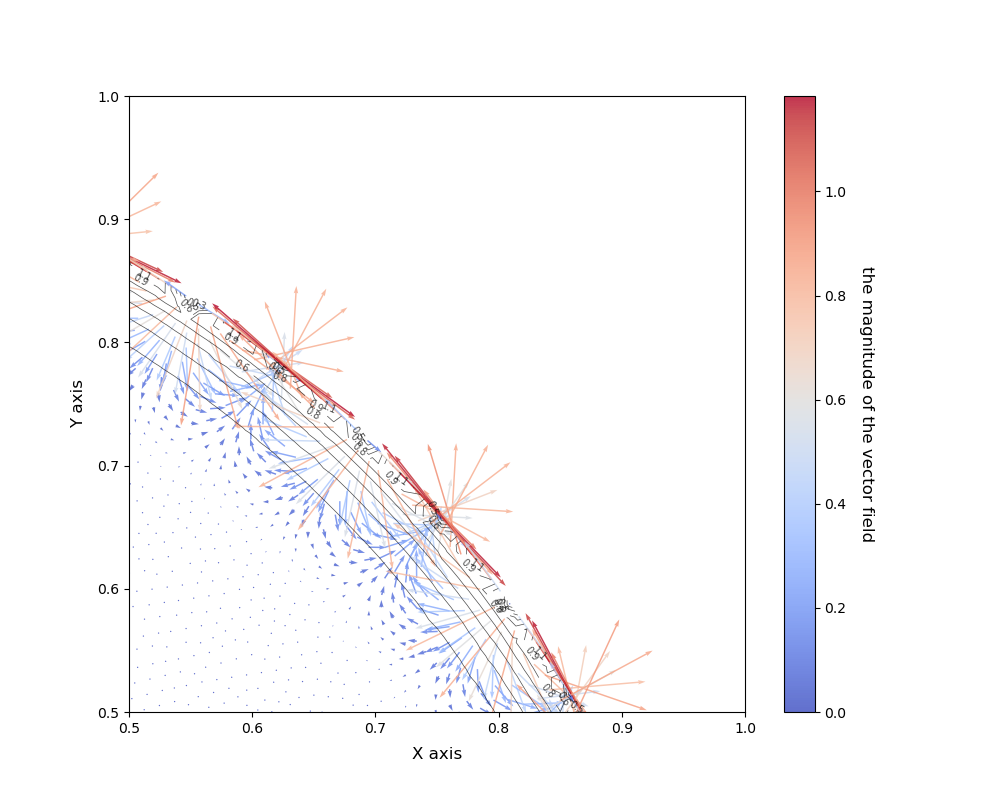}
			\caption{Local illustrations of \eqref{fig:gu_circle_A} in the region  \([0.5,1] \times [0.5,1]\).}
			\label{fig:gu_circle_B}
		\end{subfigure}
		\caption{Global and local illustrations of ${\rm Re }(\nabla u_n)$ for the occurrence of the surface resonance in  a high-contrast unit disk  $\Omega$ with the high-contrast parameter $\delta=0.01$ and the fixed index $n=35$.}
		\label{fig:gu_circle}
	\end{figure}

	\begin{figure}[h]
	\centering
	\begin{subfigure}{0.45\textwidth}
		\centering
		% Placeholder for Cassini internal field
		\includegraphics[width=\textwidth]{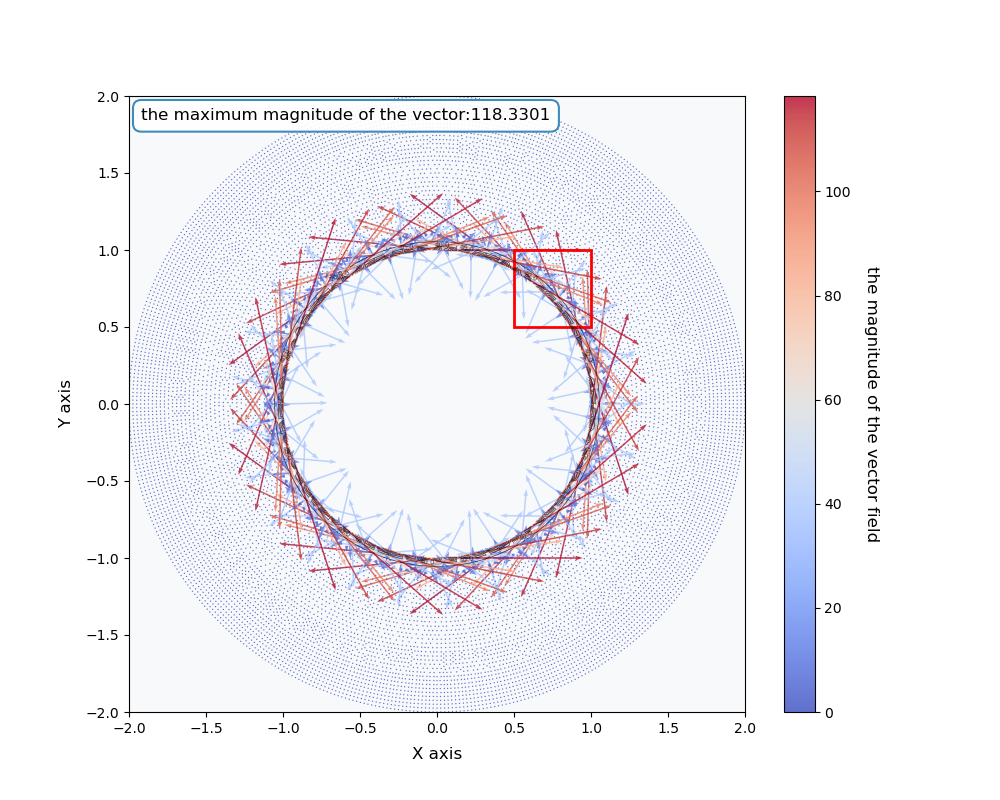}
		\caption{High oscillations of the real part of the gradient of  the external scattered field \( u_n^s \).}
		\label{fig:gus_circle_A}
	\end{subfigure}
	\hfill
	\begin{subfigure}{0.45\textwidth}
		\centering
		% Placeholder for Cassini external field
		\includegraphics[width=\textwidth]{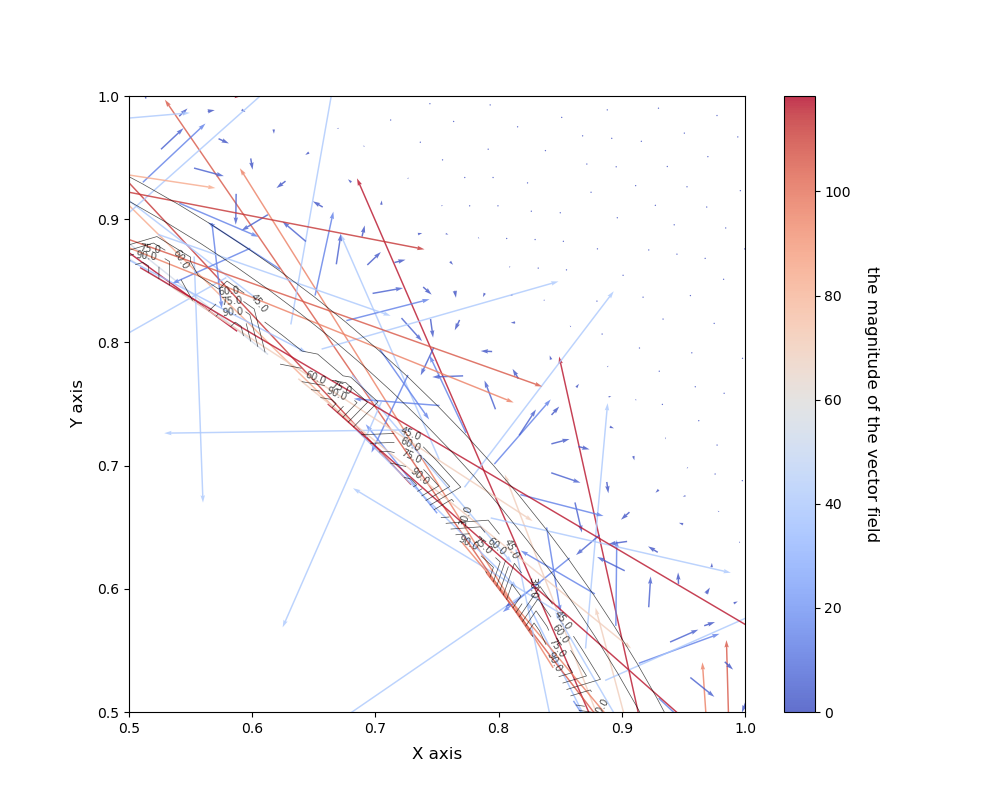}
		\caption{Local illustrations of \eqref{fig:gus_circle_A} in the region  \([0.5,1] \times [0.5,1]\).}
		\label{fig:gu_circle_B}
	\end{subfigure}
	\caption{Global and local illustrations of ${\rm Re }(\nabla u_n^s)$ for the occurrence of the surface resonance in  a high-contrast unit disk  $\Omega$ with the high-contrast parameter $\delta=0.01$ and the fixed index $n=35$.}
	\label{fig:gus_circle}
\end{figure}

\begin{figure}[h]
	\centering
	\begin{subfigure}{0.45\textwidth}
		\centering
		% Placeholder for Cassini internal field
		\includegraphics[width=\textwidth]{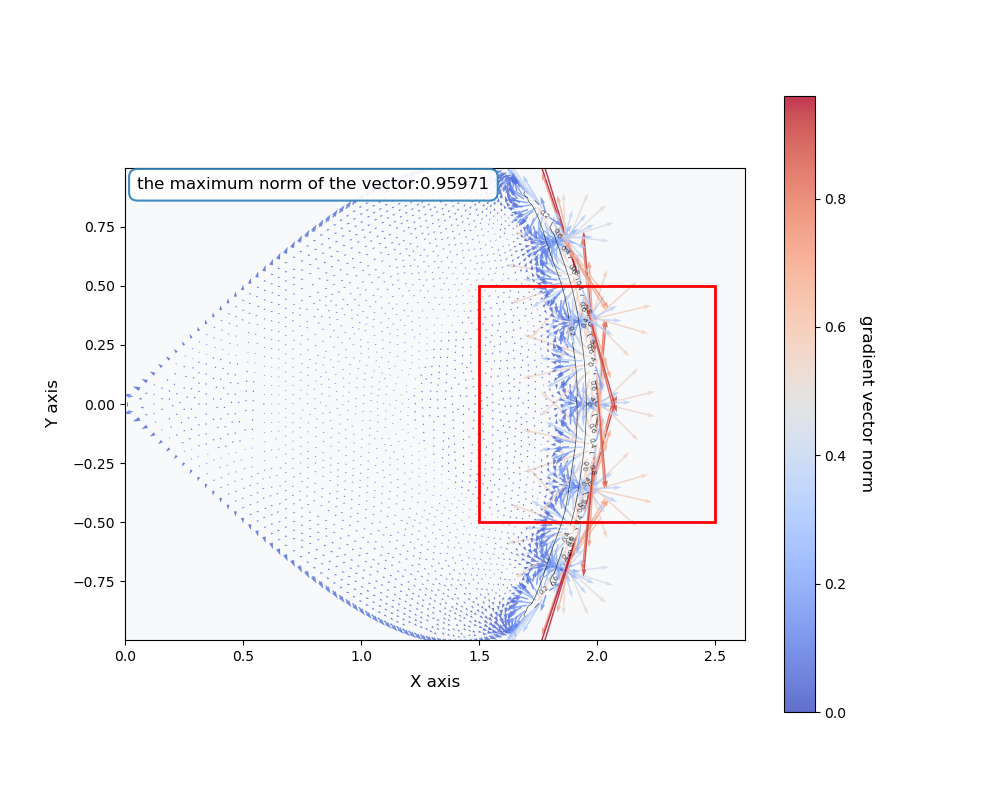}
		\caption{High oscillations of the real part of the gradient of  the internal total field \( u_n \).}
		\label{fig:gu_corner_A}
	\end{subfigure}
	\hfill
	\begin{subfigure}{0.45\textwidth}
		\centering
		% Placeholder for Cassini external field
		\includegraphics[width=\textwidth]{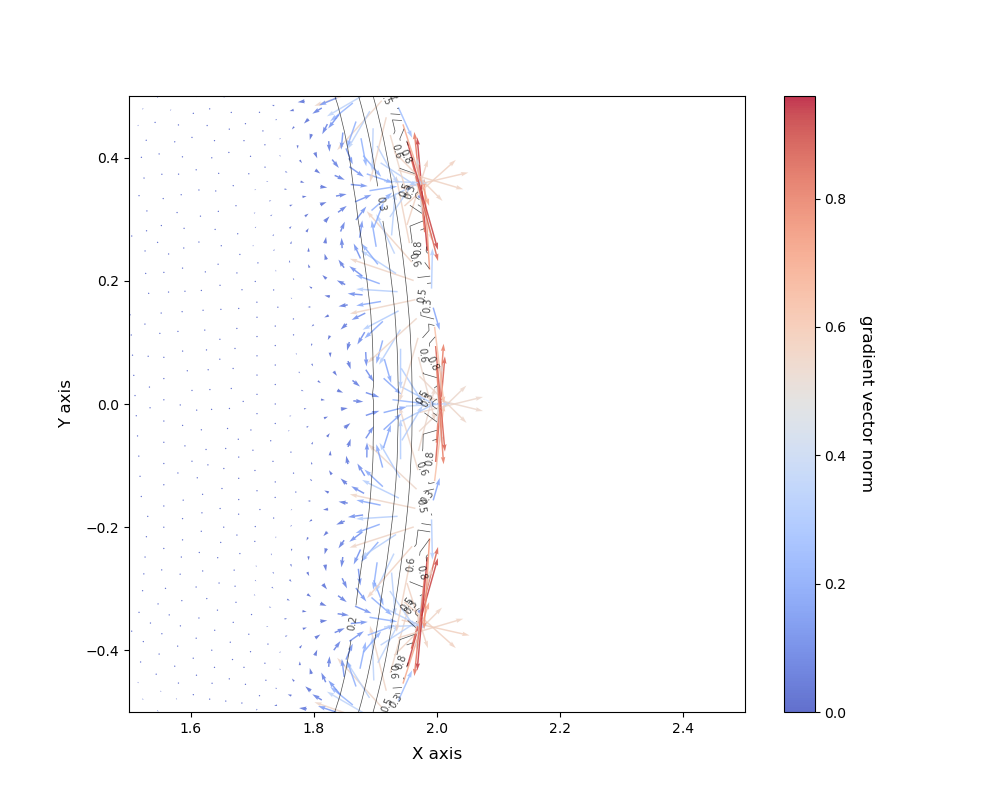}
		\caption{Local illustrations of \eqref{fig:gu_corner_A} in the region   \([1.5,2.5] \times [-0.5,0.5]\).}
		\label{fig:gu_corner_B}
	\end{subfigure}
	\caption{Global and local illustrations of ${\rm Re }(\nabla u_n)$ for  the corner high-contrast medium $\Omega$ in $\mathbb{R}^2$ with the high-contrast parameter $\delta=0.01$ and the fixed index $n=35$.}
	\label{fig:gu_corner}
\end{figure}

	\begin{figure}[h]
		\centering
		\begin{subfigure}{0.45\textwidth}
			\centering
			% Placeholder for Cassini internal field
			\includegraphics[width=\textwidth]{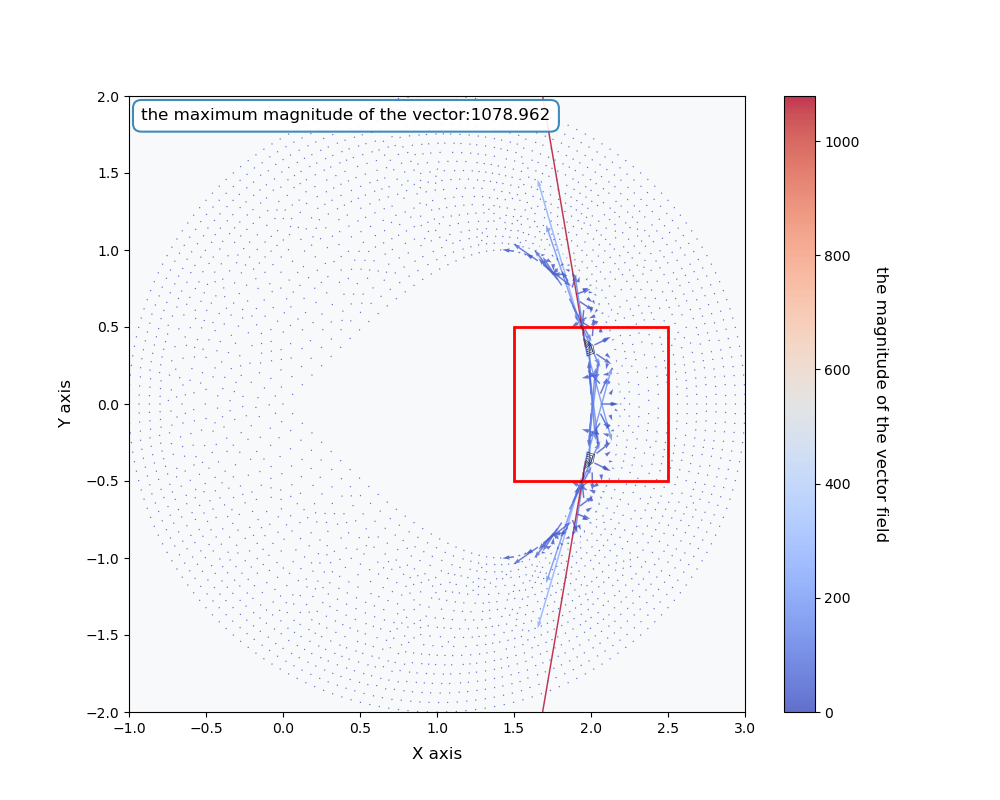}
			\caption{High oscillations of the real part of the gradient of  the external scattered field \( u_n^s \).}
			\label{fig:gus_corner_A}
		\end{subfigure}
		\hfill
		\begin{subfigure}{0.45\textwidth}
			\centering
			% Placeholder for Cassini external field
			\includegraphics[width=\textwidth]{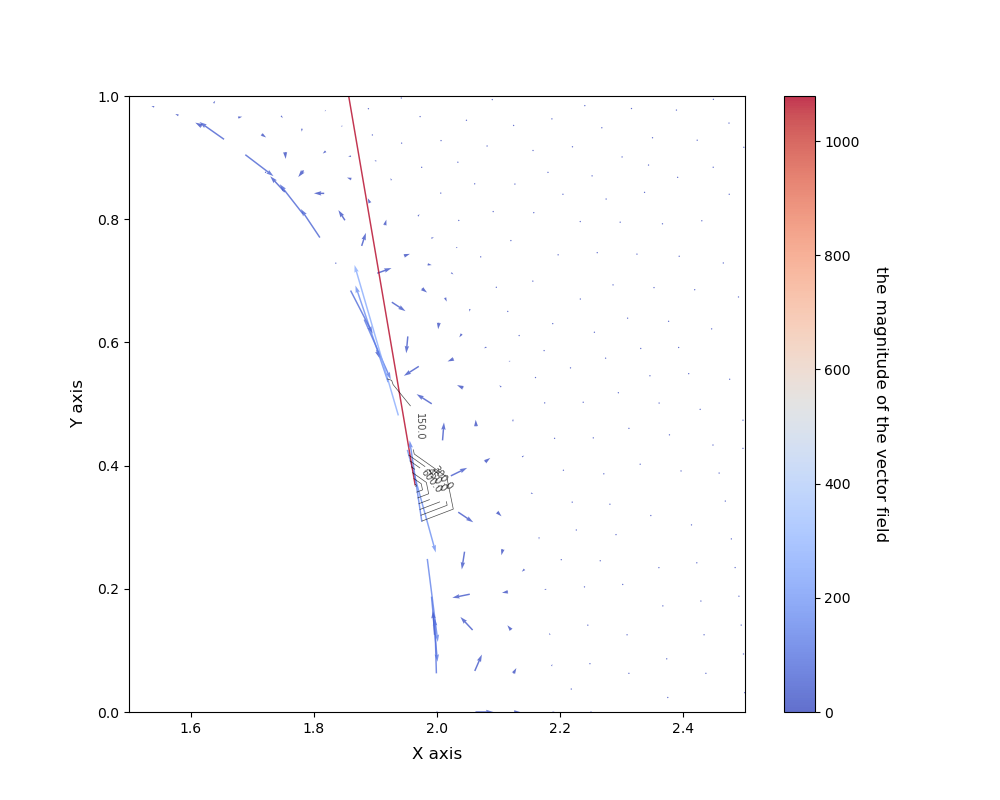}
			\caption{Local illustrations of \eqref{fig:gus_corner_A} in the region \([1.5,2.5] \times [0,1]\).}
			\label{fig:gus_corner_B}
		\end{subfigure}
		\caption{Global and local illustrations of ${\rm Re }(\nabla u_n^s)$ for  the corner high-contrast medium $\Omega$ in $\mathbb{R}^2$  with the high-contrast parameter $\delta=0.01$ and the fixed index $n=35$.}
		\label{fig:gus_corner}
	\end{figure}

\end{document}